\documentclass[11pt]{article}

\usepackage{color}
\usepackage{amsmath}
\usepackage{mathrsfs}
\usepackage{amssymb}
\usepackage{amsthm}
\usepackage{graphicx}
\usepackage{epstopdf}
\usepackage{mathtools}
\usepackage[usenames,dvipsnames, table]{xcolor}
\usepackage[hang]{caption}
\usepackage{cite}
\usepackage{caption}

\usepackage{oldgerm}  
\usepackage[yyyymmdd,hhmmss]{datetime}

\usepackage{mathabx} 
\usepackage{enumerate}

\usepackage[colorlinks=false]{hyperref}

\def\bee{\begin{enumerate}}\def\eee{\end{enumerate}}
\def\bei{\begin{itemize}}\def\eei{\end{itemize}}
\newcommand{\nco}{\newcommand}
\def\R{\mathbb{R}}\def\C{\mathbb{C}}
\nco{\red}{\color{red}}
\nco{\blue}{\color{blue}}
\definecolor{ngreen}{cmyk}{.7,0,.3,0}

\nco{\cyan}{\color{cyan}}
\nco{\brown}{\color{Magenta}}

\nco{\magenta}{\color{magenta}}

\nco{\violet}{\color{violet}}
\nco{\olive}{\color{Emerald}}
\nco{\orange}{\color{orange}}
\nco{\redend}{\normalcolor}
\nco{\blueend}{\normalcolor}
\def\inv#1{\frac{1}{#1}}
\def\Tr{{\rm Tr}}

\def\E{\Bbb{E}}

\def\ommit#1{{}}
\def\({\left(}
\def\){\right)}
\def\ie{{\it i.e.,\/}\ }

\definecolor{cb}{rgb}{.8,.5,0}

\nco{\rnc}{\renewcommand}
\rnc{\title}[1]{{\Large\bf\mbox{}\\\medskip#1\bigskip\medskip\\}}
\rnc{\author}[1]{{\large #1\smallskip\\}}
\nco{\address}[1]{{\em #1\medskip\\}}

\def\be{\begin{equation}}\def\ee{\end{equation}}
\def\bea{\begin{eqnarray}}\def\eea{\end{eqnarray}}
\def\bee{\begin{enumerate}}\def\eee{\end{enumerate}}
\def\bei{\begin{itemize}}\def\eei{\end{itemize}}
\def\oh{\frac{1}{2}}
\def\ommit#1{{}}

\def\tX{\tilde X}
\def\X1{\X1}
\def\X1{Y_2}
\def\tX1{\tilde \X1}
\def\tX2{\tilde X_2}
\def\tX{\tilde X}

\def\[{[\![} \def\]{]\!]} 
\def\U{\mathrm{U}}
\def\Ud{U^\dagger}
\def\Wg{\mathrm{Wg}}
\def\Pc{\mathcal{P}}

\def\ah{\widehat{\alpha}}

\def\Ee{\mathbb{E}}
\def\C{\mathbb{C}}
\def\AcN{\mathcal{A}^{(N)}}
\def\McN{\mathcal{M}^{(N)}}
\def\OcN{\mathcal{O}^{(N)}}
\def\Cl{\text{Cl}}
\def\fb{\bar{f}}
\def\gb{\bar{g}}
\def\mb{\overline{m}}
\def\kb{\overline{\kappa}}
\def\Kb{\overline{K}}
\def\Db{\overline{\Delta}}
\def\Zb{\overline{Z}}
\def\Kcgl{K^{\text{CGL}}}

\usepackage{amsthm}
\theoremstyle{plain}
\newtheorem{theorem}{Theorem}
\newtheorem{corollary}{Corollary}
\newtheorem{lemma}{Lemma}

\newtheorem{proposition}{Proposition}
\newtheorem{definition}{Definition}
\theoremstyle{remark}
\newtheorem{remark}{Remark}
\def\comment#1{{\blue #1}}
\def\comment#1{}
\def\adots{\mathinner{\mkern2mu\raise1pt\hbox{.}\mkern3mu\raise4pt\hbox{.}\mkern1mu\raise7pt\hbox{.}}}  
\def\dddots{\mathinner{\mkern2mu\raise9pt\hbox{.}\mkern3mu\raise3pt\hbox{.}\mkern2mu\raise-4pt\hbox{.}}}

\def\?{\,\buildrel ?\over =\,}

\def\moi#1{{#1}}\def\moi#1{}

\def\Ud{U^\dagger}\def\Vd{V^\dagger}

\DeclareFontEncoding{LS1}{}{}
\DeclareFontSubstitution{LS1}{stix}{m}{n}
\DeclareSymbolFont{stixsymbols}{LS1}{stixscr}{m}{n}
\SetSymbolFont{stixsymbols}{bold}{LS1}{stixscr}{b}{n}
\DeclareMathSymbol{\kay}{\mathalpha}{stixsymbols}{"6B}

\begin{document}

\begin{titlepage}

\begin{center}
\title{Finite $\boldsymbol{N}$ precursors of the free cumulants}
\medskip
\author{Sylvain Lacroix and Jean-Bernard Zuber}
\address{ \vskip4pt 
 Laboratoire de Physique Th\'eorique et Hautes \'Energies (LPTHE) \\
 Sorbonne Universit\'e \& CNRS, UMR 7589\\ F-75005 Paris, France}~\\[-6pt]
 {\href{mailto:lacroix@lpthe.jussieu.fr}{{\ttfamily lacroix@lpthe.jussieu.fr}} \hspace{20pt} \href{mailto:zuber@lpthe.jussieu.fr}{{\ttfamily zuber@lpthe.jussieu.fr}}}
\bigskip\bigskip
  \begin{abstract}
We study $\U(N)$ invariant polynomials on the space of $N\times N$ matrices first introduced by Capitaine and Casalis, that are precursors of free cumulants in various respects. First, they are polynomials of deterministic matrices, that are not yet evaluated over some probability law, 
contrary to what is usually meant by cumulants. Secondly, they converge towards the algebraic expression of free cumulants in terms of moments as $N\to \infty$, with $1/N^2$ corrections expressed in terms of monotone Hurwitz numbers. Their most crucial property is their additivity with respect to averaging over sums of $\U(N)$ conjugacy orbits, providing a finite $N$ version of the well-known additivity of free cumulants in free probability. Finally, they extend several properties of free cumulants at finite $N$, including a Wick rule for their average over a Gaussian weight and their appearance in various matrix integrals. Building on the additivity property of these precursors, we also define and compute a coproduct describing the behaviour of general invariant polynomials with respect to the addition of $\U(N)$ conjugacy orbits, as well as their expectation values on sums of $\U(N)$-invariant random matrices. In our construction, a central role is played by the so-called HCIZ integral, both for the construction of the precursors and for the derivation of their properties. 

\vskip 1.cm
\noindent {\it Keywords}: Free cumulants, Random Matrix Theory, large N limit, symmetric functions, Hurwitz numbers \\
\noindent \textit{Mathematics Subject Classification:}  
05E05, 46L54, 15B52, 60B20,  60Cxx
\end{abstract}
\end{center}
 \end{titlepage}


\tableofcontents
\newpage


\section{Introduction}

\subsection{Free cumulants and their additivity}

The notion of cumulants plays a key role in probability theory, statistics and statistical mechanics. There exist different variants of these objects, depending on the domain of application. Among them are the \textit{free cumulants}, which find their origin in the counting of non-crossing partitions by Kreweras~\cite{Krw}. 
They were initially introduced  under the name of ``connected planar correlation functions" 
in the study of matrix integrals in the large size limit \cite{BIPZ}. 
They reappeared in the context of
free probability \cite{Voi86} in the work of Speicher \cite{Spe93}, who proved the fundamental result that the free cumulants of two non-commutative random variables are additive if these variables are ``freely independent" (or ``free" in short).

To explain this connection, let us first recall how free cumulants arise in the theory of large size matrices. We define the $n$-th moment of a $N\times N$ matrix $A_N$ by $m_n(A_N):=\frac{1}{N}\Tr(A_N^n)$. The free cumulants $\kappa_n(A_N)$ of this matrix are then related to the moments through the combinatorics of non-crossing set-partitions. More precisely, if $\alpha=(\alpha_1,\dots,\alpha_\ell)\,\vdash\, n$ is a partition of the integer $n$, we let $\text{NC}(\alpha)$ be the number of non-crossing set-partitions of $\lbrace 1,\dots,n\rbrace$ formed by $\ell$ subsets of sizes $\alpha_1,\dots,\alpha_\ell$. The free cumulants are then characterised by the moment-cumulant relation
\begin{equation}\label{eq:MomCumRelIntro}
m_n(A_N) = \sum_{\alpha=(\alpha_1,\dots,\alpha_\ell) \,\vdash\, n} \text{NC}(\alpha)\, \kappa_{\alpha_1}(A_N)\cdots\kappa_{\alpha_\ell}(A_N)\,,
\end{equation}
which can be inverted order by order, expressing $\kappa_n(A_N)$ as a polynomial in the moments $m_k(A_N)$ with $k\leq n$.

The link with free probability arises in the large $N$ limit. Suppose that we are given two sequences of $N\times N$ matrices $(A_N)_{N\geq 1}$ and $(B_N)_{N\geq 1}$, with appropriate assumptions on the existence of their moments $m_n(A_N)$ and $m_n(A_N)$ as $N$ grows large (``$A_N$ and $B_N$ have a limit eigenvalue distribution"). Then, the variables $A_N$ and $U_NB_N\Ud_N$, with $U_N$ chosen randomly in $\U(N)$ following the Haar measure, are asymptotically free when $N\to\infty$ \cite{Voi91}. In expectation values, the additivity of the free cumulants then translates to the following crucial identity for all $n\geq 1$:
\begin{equation}\label{eq:AddFree}
\int_{\U(N)} DU_N\,\kappa_n \bigl(A_N+U_NB_N\Ud_N\bigr) \approx \kappa_n(A_N) + \kappa_n(B_N) \qquad \text{ as } N\to\infty\,,
\end{equation}
where $DU_N$ stands for the Haar measure on $\U(N)$. The question then naturally arises: is there a finite $N$ analogue of the free cumulants $\kappa_n$ that 
satisfies this additivity property? Such ``precursors'' of the free cumulants were defined by Capitaine and Casalis in~\cite{CC06}. In this article, we will propose a new approach to construct and study these precursors based on the HCIZ integral, allowing us to recover various results on their properties and derive new ones, summarised in the next section. We also refer to section \ref{sec:review} for a brief review and discussion of pre-existing works in related directions.\\

A technical point is in order before proceeding, though. In random matrix theory, moments are usually associated with a random variable $A_N\sim\mu_N$, as the expectation values of $\frac{1}{N}\Tr(A_N^n)$, from which the free cumulants are subsequently constructed. These moments and cumulants are then scalar integrated quantities, which are functionals of the probability law $\mu_N$. In this paper we adopt a slightly different standpoint: the moments $m_n$ and free cumulants $\kappa_n$ are defined as polynomials on the space of $N\times N$ matrices. Depending on the context, they can either be evaluated on a purely deterministic matrix $A_N$, producing numbers $m_n(A_N)$ and $\kappa_n(A_N)$, or on a random matrix $A_N\sim\mu_N$, in which case $m_n(A_N)$ and $\kappa_n(A_N)$ should be understood as scalar random variables. In the latter case, one can further consider the expectation values of these variables, which we will denote by $\mathbb{E}_{A_N\sim \mu_N}\bigl(m_n(A_N)\bigr)$ and $\mathbb{E}_{A_N\sim \mu_N}\bigl(\kappa_n(A_N)\bigr)$, hence recovering the more standard objects used in random matrix theory. We refer to the section \ref{sec:Random} for a more detailed discussion of this distinction between polynomials and integrated quantities.


\subsection{Main results and plan of the paper}

In this subsection we summarise the main results of the paper and outline its plan. We postpone details and proofs to the main text and refer to Section \ref{notations} for a complete list of notations and conventions, keeping those to a minimum here. For a given finite $N$, consider the so-called HCIZ integral
\be
Z(A, B\,; z):= 
 \int_{\U(N)} DU \exp \bigl( N z\, \Tr(U A U^\dagger B) \bigr)\,,\ee
where $A$ and $B$ are two $N\times N$ complex matrices (the subscript $N$ on $A, B, U$ will be dropped hereafter). This integral can be expanded in an infinite series, whose coefficients are invariant polynomials in the matrices $A$ and $B$ (see section \ref{HCIZexp} for details). The main protagonists of this paper are defined as specific coefficients in this expansion.\\

\begin{definition}The HCIZ integral $Z(A,B\,; z)$ possesses a unique expansion in powers of $z$ and products of traces $\Tr(B^k)$, $k\in\lbrace 1,\dots,N\rbrace$, of its second argument. For $n\in\lbrace 1,\dots,N\rbrace$, we define
\begin{equation}\label{eq:Kn}
K^{(N)}_n(A) := \frac{n}{N} \bigl[ z^n\, \Tr(B^n) \bigr] Z(A,B\,; z)
\end{equation}
as the appropriately normalised coefficient of $z^n \Tr(B^n)$ in this expansion.

\noindent More generally, if $\alpha=(\alpha_1,\dots,\alpha_{\ell})\vdash n$ is a partition of $n\in\lbrace 1,\dots,N\rbrace$ of length $\ell$, \ie $n=\alpha_1+\dots+\alpha_{\ell}$, we define
\begin{equation}\label{eq:Kalpha}
K^{(N)}_\alpha(A) := \frac{\prod_{k=1}^n k^{\ah_k}\, \ah_k!}{N^\ell} \bigl[ z^n\,  \Tr(B^{\alpha_1}) \cdots \Tr(B^{\alpha_\ell}) \bigr] Z(A,B\,; z)\,,
\end{equation}
where $\ah_k$ denotes the number of times $k$ appears in the partition $\alpha$.
\end{definition}

\begin{remark}
We note that we restricted the definition of $K_\alpha^{(N)}(A)$ to partitions of integers $n$ lesser or equal to $N$: this ensures that the polynomials $\Tr(B^{\alpha_1}) \cdots \Tr(B^{\alpha_\ell})$ are linearly independent and thus that the construction \eqref{eq:Kalpha} is non-ambiguous. The polynomial $K_n^{(N)}(A)$ coincides with $K_\alpha^{(N)}(A)$ for $\alpha=(n)$ the partition of $n$ formed by a single number. 
To lighten the notation, we will often omit the superscript $(N)$ of $K^{(N)}_\alpha(A)$ in the main text, when the context makes it clear that we are working at a fixed finite size $N$. Finally, the normalisation in the definitions \eqref{eq:Kn} and \eqref{eq:Kalpha} has been chosen to obtain a specific large $N$ limit, as shown in point 4 of Theorem \ref{thm:summary} below.
\end{remark}

The polynomials $K_n^{(N)}$, and their generalisations $K_\alpha^{(N)}$, possess many remarkable properties making them natural ``precursors'' at finite $N$ of the free cumulants $\kappa_n$. We summarise these properties in the following theorem.

\begin{theorem}\label{thm:summary}
Let $N\in\mathbb{Z}_{\geq 1}$, $n\in\lbrace 1,\dots,N\rbrace$ and $\alpha=(\alpha_1,\dots,\alpha_\ell)$ a partition of $n$.\vspace{-3pt}
\begin{enumerate}
\item $K^{(N)}_n$ and $K^{(N)}_\alpha$ are invariant polynomials of degree $n$ on the space of $N\times N$ matrices.
\item They admit natural expansions \eqref{eq:KSchur}, \eqref{Kn} and \eqref{eq:Kp2} in the bases of Schur and Newton polynomials. 
\item The $K^{(N)}_n$ satisfy the following additivity property, for any $N\times N$ matrices $A$ and $B$:
\begin{equation}\label{eq:AddFinite}
\int_{\U(N)} DU K^{(N)}_n(A+U B\Ud)= K_n^{(N)}(A)+K_n^{(N)}(B)\,,
\end{equation}
generalising \eqref{eq:AddFree} for finite $N$.
\item Suppose we are given a sequence $(A_N)_{N\geq 1}$ with a limit eigenvalue distribution when $N\to\infty$, so that the normalised moments $m_k(A_N)=\frac{1}{N}\Tr(A_N^k)$ all converge to finite values $m_k$. Then
\begin{equation}
K_n^{(N)}(A_N) \approx \kappa_n(A_N) \qquad \text{ and } \qquad K_\alpha^{(N)}(A_N) \approx \kappa_{\alpha_1}(A_N)\cdots\kappa_{\alpha_\ell}(A_N) \qquad \text{ as }\, N\to\infty\,,
\end{equation}
where $\kappa_n$ is the free cumulant defined through equation \eqref{eq:MomCumRelIntro}.
\item Under the same assumptions, $K^{(N)}_\alpha$ admits a ``topological'' expansion in terms of moments:
{\small\begin{equation*}
K^{(N)}_\alpha(A_N) = \frac{\prod_{k=1}^n k^{\ah_k}\, \ah_k!}{n!} \sum_{g \geq 1-\ell} N^{2(1-\ell-g)}  \left( \sum_{(\beta_1,\dots,\beta_{\kay})\vdash n} (-1)^{\ell+\kay}  H_g^{\bullet \hspace{1pt},\hspace{1pt}\leq}(\alpha,\beta) \,m_{\beta_1}(A_N)\cdots m_{\beta_{\kay}}(A_N) \right)\,,
\end{equation*}}where $H_g^{\bullet \hspace{1pt},\hspace{1pt}\leq}(\alpha,\beta)$ are the disconnected weakly monotone double Hurwitz numbers of genus $g$ and type $\alpha,\beta$ (see section \ref{sec:Hur} for definitions and details).
\item The relation in point 5 can be inverted to reconstruct the moments, yielding
\begin{equation}
m_n(A_N) = \sum_{g \geq 0}  N^{-2g} \sum_{\alpha\vdash n} P_g(\alpha)\,  K^{(N)}_\alpha(A_N)\,,
\end{equation}
where $P_g(\alpha)$ is the number of permutations of cycle type $\alpha$ and genus $g$. This provides the $1/N^2$ corrections to the moment-cumulant relations \eqref{eq:MomCumRelIntro} (noting that $P_0(\alpha)=\text{NC}(\alpha)$).
\item $K_n^{(N)}(A)$ can be extracted as the coefficient of $x^n$ in the generating function
\begin{equation}
\mathcal{K}^{(N)}(A;x) = \frac{1}{N} \,  \int_{\U(N)} DU e^{N\Tr(A\Ud)}\, \Tr \left(\frac{1}{1-xU}\right)\,.
\end{equation} 
In particular, $\dfrac{1}{x}(\mathcal{K}^{(N)}(A;x)-1)$  converges to the Voiculescu $\mathcal{R}$-transform when $N\to\infty$.
\item If $A$ is a random Hermitian matrix in the Gaussian Unitary Ensemble GUE$(N,\sigma)$, then
\begin{equation}
\Ee_{A\sim \text{GUE}(N,\sigma)}\bigl( K_\alpha^{(N)}(A) \bigr) = \left\lbrace \begin{array}{ll}
\sigma^{n}\, & \text{ if } \alpha=[2^{n/2}],\\[4pt] 
0 & \text{ otherwise}\,,
\end{array} \right.
\end{equation}
where $[2^{n/2}]=(2,\dots,2)$ is the partition of $n$ into pairs (which exists only for even $n$). In particular,
\begin{equation}
\Ee_{A\sim \text{GUE}(N,\sigma)}\bigl( K_n^{(N)}(A) \bigr) = \delta_{n,2}\,\sigma^2\,.
\end{equation}
\item  Let $\sigma\in S_n$ be a permutation of cycle type $\alpha$. Then
\begin{equation}
\int_{\U(N)} DU \,A^U_{1\sigma(1)}\,A^U_{2\sigma(2)}\cdots A^U_{n\sigma(n)} = N^{\ell-n} K_\alpha^{(N)}(A)\,,
\end{equation}
where $A^U_{ij}$ is the $(i,j)$ entry of the matrix $A^U = UAU^\dagger$.\vspace{4pt}\\
In particular, suppose $A\sim\mu$ is a random matrix following a $\U(N)$-invariant distribution $\mu$ and let $\sigma\in S_n$ be a permutation of cycle type $\alpha$. Then
\begin{equation}
\Ee_{A\sim\mu}\bigl( A_{1\sigma(1)}\,A_{2\sigma(2)}\cdots A_{n\sigma(n)}\bigr) = N^{\ell-n}\, \Ee_{A\sim\mu}\bigl( K_\alpha^{(N)}(A)\bigr)\,.
\end{equation}
\end{enumerate}
\end{theorem}

Points 1 and 2 will be proven in Section \ref{HCIZint}, which is devoted to the HCIZ integral and its representation in terms of standard invariant polynomials. Property 4 shows that the $K^{(n)}_\alpha$ are directly related to the free cumulants in the large $N$ limit, with point 5 providing the corresponding $1/N^2$ corrections in terms of the moments and point 6 the inverted moment-cumulant relation. They will be proven in Section \ref{LargeNlim}, which concerns the large $N$ limit of the HCIZ integral and its relation to weakly monotone Hurwitz numbers. Properties 7 and 8 will be derived in Sections \ref{genfn} and \ref{int-wick}, making use of the origin of $K_\alpha^{(N)}$ from the HCIZ integral. Property 9 is a finite $N$ generalisation of a well-known characterisation of free cumulants~\cite{CMSS06,Maillard-etal,BH24}, which we will show in Section \ref{cyclicpr} using the Weingarten--Samuel--Collins calculus.

Finally, the most important property of $K_n^{(N)}$ is the point 3 of the theorem, which is the finite $N$ version of the additivity \eqref{eq:AddFree} of free cumulants. It follows directly from the remarkable identity
\begin{equation}\label{eq:AddHCIZ}
\int_{\U(N)} DU \,Z(A+UB\Ud,C\,; z) = Z(A,C\,; z)\,Z(B,C\,; z)\,. 
\end{equation}
obeyed by the HCIZ integral, by projecting onto the single trace $z^n\,\Tr(C^n)$. The identities \eqref{eq:AddFinite}, \eqref{eq:AddHCIZ} and their applications will be the main subjects of Section \ref{Additivity}. They are deeply tied to the addition of conjugacy orbits and thus to the Horn problem, as we shall discuss in Section \ref{Horn}. In view of these results, it is natural to wonder if the more general polynomials $K^{(N)}_\alpha$ also have a simple behaviour with respect to the addition of conjugacy orbits. Although they do not satisfy the additivity property \eqref{eq:AddFinite}, they obey a generalisation thereof \eqref{eq:OrbitK}, which follows from the identity \eqref{eq:AddHCIZ} and still takes a quite simple form. More generally and more formally, if $\AcN$ denotes the algebra of invariant polynomials on $N\times N$ matrices and $f\in\AcN$, it is natural to study the quantity
\begin{equation}\label{eq:OrbitCoIntro}
\Delta f(A,B) = \int_{\U(N)} DU\, f(A+UB\Ud)\,.
\end{equation}
This is a sum of products of invariant polynomials of $A$ and $B$, which can then be interpreted as an element of the tensor product $\AcN \otimes \AcN$. This procedure thus defines a coproduct $\Delta : \AcN \to \AcN\otimes\AcN$ on the algebra of invariant polynomials\footnote{Note that this coproduct is not a morphism with respect to the standard product on $\AcN$ and thus does not define a Hopf algebra structure on it. See Section \ref{Sec:Coproduct} for details.}. The results mentioned above show that the $K^{(N)}_\alpha$'s define a family of polynomials in $\AcN$ on which this coproduct is particularly simple. In particular, the $K^{(N)}_n$'s are so-called primitive elements, satisfying $\Delta K^{(N)}_n = K^{(N)}_n \otimes 1 + 1 \otimes K^{(N)}_n$. We discuss these aspects and the general properties of the coproduct $\Delta$ in the Section \ref{Sec:Coproduct} (and in particular show that $\Delta$ is dual to the standard product on $\AcN$ with respect to some natural scalar product).\\

So far, we have mainly focused on algebraic properties of invariant polynomials and generalised precursors $K_\alpha^{(N)}$. Section \ref{sec:Proba} turns to more probabilistic considerations and applications of our formalism to random matrix theory. If $f\in\AcN$ is an invariant polynomial and $A\sim\mu$ is a random matrix with $\U(N)$-invariant probability measure $\mu$, it is natural to consider the expectation value
\begin{equation}
\fb[\mu] = \E_{A\sim \mu}[f(A)] = \int f(A)\,\mu(A)\,.    
\end{equation}
Given two independent random matrices $A\sim \mu$ and $B\sim\nu$, their sum $C=A+B$ follows the convoluted distribution $\mu\star\nu$. Then, a natural question is whether we can reconstruct the expectation values
\begin{equation}
\fb[\mu\star\nu]=\E_{C\sim \mu\star\nu}[f(C)] = \iint f(A+B) \mu(A)\,\nu(B) 
\end{equation}
of invariant polynomials on this sum from the knowledge of the expectation values for the measures $\mu$ and $\nu$. This is the main subject of Section \ref{sec:Proba}, which answers this question positively. Namely, the Theorem \ref{Thm:fConvolution} expresses $\fb[\mu\star\nu]$ in terms of the orbit coproduct introduced in equation \eqref{eq:OrbitCoIntro}. Since the latter can be computed using the results of Section \ref{Sec:Coproduct}, this gives a general procedure to determine the quantity $\fb[\mu\star\nu]$. In particular, the averaged precursors $\Kb_n$ are additive with respect to the convolution, \ie $\Kb_n[\mu\star\nu]=\Kb_n[\mu]+\Kb_n|\nu]$.


\subsection{A short review of related work} 
\label{sec:review}

{ The construction of finite $N$ analogues or precursors of the free cumulants is a natural question that has attracted the attention of several authors. Let us quickly review these works and discuss the relations and differences with the approach developed in the present article.
 
\bei
\item The precursors of free cumulants $K_n$ and their generalisations $K_\alpha$ (or rather their expectation values $\Kb_\alpha$) were first introduced by Capitaine and Casalis in~\cite{CC06}. In that reference, they are defined using the convolution on the symmetric group (see also~\cite{CC08} for a more geometric interpretation). In contrast, our construction is based on the HCIZ integral. We recover the expression of $K_\alpha$ as a convolution in equation \eqref{eq:Kp2}, thus making the link with~\cite{CC06} explicit. Some of the properties of the generalised precursors summarised in our main Theorem \ref{thm:summary} were already found in~\cite{CC06,CC08}, including the additivity of $K_n$ (point 3), their large $N$ limit (point 4) and their expectation value under a Gaussian distribution (point 8). Our approach allows us to rederive these results from a different route and to obtain various new ones, in particular the relation between the large-$N$ expansion of $K_\alpha$, Hurwitz numbers and the enumeration of permutations by genus (points 5 and 6 in Theorem \ref{thm:summary}), as well as the generating function of $K_n$ (point 7) and the behaviour of general invariant polynomials with respect to sums of conjugacy orbits (see Sections \ref{Sec:Coproduct} and \ref{sec:Proba}).
\item The notion of free convolution is a quite important tool in free probability~\cite{Voi86,Spe93,Nica}. A variant of this construction for finite size matrices was introduced in \cite{Marcus} through the convolution of their characteristic polynomials. The associated
finite free cumulants, that are additive with respect to that convolution, were studied in \cite{AP}.
Despite some similarities, they differ from our functions $K_n$, as testified by their
approach to infinite $N$ with  $1/N$ corrections, 
in contrast with our $1/N^2$ expansion (see point 5 above in Theorem \ref{thm:summary}). For example, their finite free cumulant of
degree 4 reads $f_4= m_4 -  (2N-3)/(N-1) m_2^2$ (for traceless matrices), see also 
\cite{Mergny-Potters,Mergny-th}, while $K_4$ is proportional to $m_4-(2N^2-3)/(N^2+1)m_2^2$. Although the precursors $K_n$ introduced here satisfy a number of other interesting properties, we will see in Section \ref{sec:convolution} that they do not seem to be associated with a well-behaved notion of finite free convolution, in contrast with the cumulants of \cite{AP}. 
\item In a recent work\cite{KMW}, mostly devoted to random tensors,
a concept of cumulants for finite tensors is
developed. When applied to 2-tensors regarded as  matrices, these objects seem to coincide with the generalised precursors $K_\alpha$. It would be interesting to study the relation between their construction and that of~\cite{CC06,CC08} or the present paper.
\item Even more recently~\cite{CGL}, Collins, Gurau and Lionni have,
also in the context of tensor theories, developed a
concept of free cumulants at finite $N$, which are additive under the addition of two independent random matrices. Their construction, however, depends in an essential way on the evaluation of moments and cumulants against some probability law. This is in contrast with the (generalised) precursors considered in this paper, which are defined as invariant polynomials on matrices and satisfy various interesting algebraic properties, before integration over some probability law. We return in Section \ref{sec:CGL} to similarities and differences between their approach and ours.
\eei}



\subsection{Notations and prerequisites}
\label{notations}

A brief review of the notations and prerequisites used in that paper is in order.

\paragraph{Partitions of integers and Young diagrams.} $\alpha\vdash n$ denotes a \textit{partition of the integer} $n$, which we write in two equivalent ways
\begin{equation}
    \alpha=\bigl(\alpha_1,\dots,\alpha_{\ell(\alpha)}\bigr) = [n^{\ah_n} \dots 1^{\ah_1}]\,,
\end{equation}
either in terms of its components $\alpha_1 \geq \alpha_2 \geq \dots \geq \alpha_{\ell(\alpha)} > 0$ or in terms of its multiplicities $\ah_k$ (\ie the number of times $k$ appears in $\alpha$). In these notations, the \textit{degree} of $\alpha$ is
\begin{equation}
    d(\alpha) = n = \sum_{i=1}^{\ell(\alpha)} \alpha_i = \sum_{k=1}^n k\,\ah_k\,,
\end{equation}
while its \textit{length} is the number of its components
\begin{equation}
    \ell(\alpha) = \sum_{k=1}^n \ah_k\,.
\end{equation}
When a multiplicity $\ah_k$ is 0, we will sometimes omit $k$ from the notation $[n^{\hat\alpha_n} \dots 1^{\hat\alpha_1}]$. For instance, the partition $(3,3,1)=[3^2\,1]$ has degree 7 and length 3.

Equivalently, the partition $\alpha\vdash n$ can be thought of as a Young diagram with $n$ boxes, formed by $\ell(\alpha)$ rows of sizes $\alpha_1 \geq \dots \geq \alpha_{\ell(\alpha)}$. By a slight abuse of terminology, we will identitfy the partition $\alpha$ with its associated Young diagram. We will then use the notation $\square\in\alpha$ to designate a box in this Young diagram.

\paragraph{Permutations.} $S_n$ denotes the \textit{group of permutations} of $\lbrace 1,\dots,n\rbrace$, of size $n!$. If $\sigma\in S_n$ is formed by $\hat\alpha_k$ cycles of length $k$, for $k\in\lbrace 1,\dots,n\rbrace$, we will say that $\sigma$ has \textit{cycle type} $[\sigma]=[n^{\hat\alpha_n} \dots 1^{\hat\alpha_1}]$, which is a partition $\alpha\vdash n$. The subset of $S_n$ composed by all permutations of cycle type $\alpha$ forms a \textit{conjugacy class} of $S_n$ and will be denoted by $\Cl_\alpha$. Its size is
\begin{equation}\label{eq:SizeClass}
    |\Cl_\alpha| = \frac{n!}{\prod_{k=1}^n k^{\ah_k}\, \ah_k!}\,.
\end{equation}
The partition $[1^n]$ labels the conjugacy class of the identity formed by only 1 element, while the partition $[n]$ labels the conjugacy class of $n$-cycles, formed by $|\Cl_{[n]}|=(n-1)!$ elements.

The \textit{irreducible representations} of $S_n$ are naturally labeled by partitions / Young diagrams of degree $n$. The representation associated with $\lambda\vdash n$ has \textit{character} $\hat\chi_\lambda$, taking the value $\hat\chi_\lambda(\alpha)$ on the conjugacy class $\Cl_\alpha$. Its dimension is given by
\begin{equation}
    \hat d_\lambda = \hat\chi_\lambda([1^n]) = \frac{n!}{H_\lambda}\,, \qquad \text{ where } \qquad H_\lambda :=  \prod_{\square \in \lambda }{ h(\square) }
\end{equation}
and $h(\square)$ denotes the \textit{hook length} of the box $\square$ in the Young diagram $\lambda$. If the box $\square$ has coordinate $(i,j)$, the latter is defined as
\begin{equation}
    h(\square):=\lambda_i-j+\tilde\lambda_j-i+1\,,
\end{equation}
where $\tilde\lambda$ denotes the transposed Young diagram of $\lambda$ (with rows and columns exchanged). The characters satisfy the \textit{orthogonality relations}
\begin{equation}\label{eq:OrthoChi}
    \frac{1}{n!} \sum_{\lambda\vdash n} \hat\chi_\lambda(\alpha)\,\hat\chi_\lambda(\beta) = \delta_{\alpha\beta} \qquad \text{ and } \qquad \sum_{\alpha\vdash n} |\Cl_\alpha|\, \hat\chi_\lambda(\alpha)\,\hat\chi_\mu(\alpha) = \delta_{\lambda\mu}\,. 
\end{equation}

\paragraph{Matrices, invariant polynomials and unitary group.} For a matrix $X=(X_{ij})_{1\leq i,j \leq N}$ of size $N\times N$, we write its trace as $\Tr(X)=\sum_{i=1}^N X_{ii}$ and its $n$-th moment as $m_n(X)=\inv{N} \Tr(X^n)$. If $\alpha\vdash n$ is a partition of $n$, we also let
\begin{equation}
    p_\alpha(X) := \prod_{i=1}^{\ell(\alpha)} \Tr(X^{\alpha_i}) = \prod_{k=1}^{n} \Tr^{\ah_k}(X^k)\,.
\end{equation}
This defines a conjugacy-invariant polynomial $p_\alpha$ of degree $n=d(\alpha)$ on the space of $N\times N$ matrices, which we call the \textit{Newton polynomial} associated with $\alpha$.

Polynomial irreducible representations of the general linear group $\text{GL}(N)$ are labeled by partitions $\lambda$ of length $\ell(\lambda) \leq N$, \ie Young tableaux with at most $N$ rows. The corresponding characters will be denoted by $s_\lambda$ and define invariant polynomials of degree $d(\lambda)$ called \textit{Schur polynomials}. The dimension of the $\text{GL}(N)$-representation associated with $\lambda$ is 
\begin{equation}\label{eq:dC}
    d_\lambda=s_\lambda(\mathbb{I})=\frac{N^n C_\lambda}{H_\lambda} \qquad \text{ where } \qquad C_\lambda(N):= \prod_{\square \in \lambda } (1+ c(\square)/N)\,,
\end{equation}
in terms of the \textit{content function}, defined on each box $\square\in\lambda$ of coordinates $(i,j)$ by
\be\label{hc} c(\square) := j-i\,. \ee
The Schur and Newton polynomials are related by the Frobenius-Schur formula
\be\label{frobenius}  p_\alpha(X)= \sum_{\lambda\vdash n\atop \ell(\lambda)\le N}  \hat\chi_\lambda(\alpha)\, s_\lambda(X) \,, \qquad \forall\,\alpha\vdash n\,,
\ee
or its reverse form, which will be of much use in this article,
\be\label{frobeniusinv}  s_\lambda(X)= \sum_{\alpha\vdash n} \frac{|\Cl_\alpha|}{n!} \hat\chi_\lambda(\alpha)\, p_{\!\alpha}(X) \, , \qquad \forall\,\lambda\vdash n\,,\ \ell(\lambda)\le N. \ee

The compact real form of $\text{GL}(N)$ is the \textit{unitary group}  $\U(N)$, formed by $N\times N$ matrices $U$ satisfying $\Ud=U^{-1}$, where $\Ud$ is the conjugate transpose of $U$. If $f$ is a function on $\U(N)$, we denote by $\int_{\U(N)} DU\, f(U)$ its integral with respect to the \textit{Haar measure} on $\U(N)$, normalised by $\int_{\U(N)} DU=1$. Fixing partitions $\alpha,\beta,\lambda,\mu\vdash n$ and $N\times N$ matrices $X,Y$, the Schur functions then satisfy the following useful integration formulae:
 \begin{equation}
     \int_{\U(N)} DU\,s_\lambda(U X) s_\mu( \Ud Y) = \frac{\delta_{\lambda\,\mu}}{d_\lambda}  s_\lambda(XY)
 \end{equation} 
 \begin{equation}
     \int_{\U(N)} DU\, s_\lambda(U X \Ud Y) = \frac{1}{d_\lambda}  s_\lambda(X) s_\lambda(Y)\,,
 \end{equation}
while the Newton functions obey the orthogonality relation
\begin{equation}\label{eq:OrthoP}
    \int_{\U(N)} DU\, p_\alpha(U)\, p_\beta(\Ud) = \frac{n!}{|\Cl_\alpha|} \delta_{\alpha\beta}\,.
\end{equation}

\paragraph{Weingarten--Samuel--Collins calculus.}\hspace{-8pt}\footnote{In our opinion, the usual name Weingarten calculus does not make justice to these two other important contributors. We thus propose this new name, but keep the symbol $\text{Wg}$ for the coefficient in \eqref{Wg}.} 
We have~\cite{W, Sam, Collins03}
 \be\label{Wg}\int_{\U(N)} DU\, \prod_{a=1}^n U_{i_a, j_a}  \Ud_{k_a\ell_a}= \sum_{\sigma, \tau \in S_n}\Wg([\sigma \tau^{-1}])
  \prod_{a=1}^n \delta_{i_a, \ell_{\sigma(a)}}  \delta_{j_a k_{\tau(a)}}\,,\ee
  implying 
  \be\label{Wg2}  \int_{\U(N)} DU\, \Tr^n (A U B U^\dagger) =  \sum_{\sigma, \tau \in S_n}\Wg([\sigma \tau^{-1}]) \, p_{[\sigma]}(A)\, p_{[\tau]}(B)\,.\ee
where we recall that $[\sigma] \vdash n$ is the partition describing the cycle type of $\sigma\in S_n$. In these formulae, $\Wg$ denotes the Weingarten function, whose value on a partition $\alpha\vdash n$ is defined as \cite{Sam}
\begin{equation}
    \Wg(\alpha) = \frac{1}{n!^2} \sum_{\lambda\vdash n \atop \ell(\lambda)\leq N} \frac{\hat d_\lambda^2}{d_\lambda} \hat\chi_\lambda(\alpha) =  \frac{1}{N^n} \sum_{\lambda\vdash n \atop \ell(\lambda)\leq N} \frac{\hat\chi_\lambda(\alpha)}{H_\lambda\,C_\lambda} \,.
\end{equation}


\section{Precursors of free cumulants from the HCIZ integral}
\label{HCIZint}
\subsection{The HCIZ integral and its expansions}
\label{HCIZexp}
We refer to the section \ref{notations} for conventions and notations on partitions, matrices and invariant polynomials, of which we will make extensive use in this section. Consider the HCIZ integral
\be\label{hciz}
Z(A, B \,; z):= 
 \int_{\U(N)} DU \exp \bigl( N z\, \Tr(U A U^\dagger B) \bigr)\,,
\ee
where $DU$ stands for the $\U(N)$ Haar measure and $A$ and $B$ are two complex matrices of size $N\times N$. It is clear that $Z(A,B\,; z)$ is invariant under independent conjugations of $A$ and $B$ by elements of $\U(N)$. Moreover, it can be expanded as a power series in $z$ with an infinite radius of convergence and with the coefficient of $z^n$ being expressed in terms of invariant polynomials of degree $n$ in both $A$ and $B$. This expansion is the simplest in the basis of Schur polynomials $s_\lambda$, where it reads
\begin{equation}\label{Z2ss}
    Z(A,B\,; z) = \sum_{n=0}^\infty z^n \sum_{\lambda\vdash n\atop \ell(\lambda) \le N }\frac{s_\lambda(A)s_\lambda(B)}{C_\lambda} \,,
\end{equation}
with $C_\lambda$ defined in equation \eqref{eq:dC}. Another useful family of invariant functions are the Newton polynomials $p_\alpha(X)=\prod_{i=1}^{\ell(\alpha)} \Tr(X^{\alpha_i})$, related to Schur polynomials by the identity \eqref{frobeniusinv}. In terms of those, the expansion of the HCIZ integral can be rewritten as
\begin{eqnarray}
    Z(A,B\,; z) &=& \sum_{n=0}^\infty {z^n}\sum_{\lambda\vdash n\atop \ell(\lambda) \leq N } \frac{1}{C_\lambda} 
\sum_{\alpha\vdash n}\frac{|\Cl_\alpha |}{n!} \hat\chi_\lambda(\alpha)  s_\lambda(A) p_\alpha(B) \label{Z2sp} \\
 &=&  \sum_{n=0}^\infty {z^n} \sum_{\lambda\vdash n\atop \ell(\lambda) \leq N } \frac{1}{C_\lambda} 
\sum_{\alpha,\beta\vdash n}\frac{|\Cl_\alpha |\cdot|\Cl_\beta |}{n!^2} \hat\chi_\lambda(\alpha)\hat\chi_\lambda(\beta)  p_\beta(A) p_\alpha(B)\,.\label{expand}
\end{eqnarray}
Equivalently, the latter form can also be written in terms of the Weingarten coefficients as
\begin{equation}\label{Wgexpand}
Z(A,B\,; z) = \sum_{n=0}^\infty \frac{z^n}{n!} N^n \sum_{\sigma,\tau\in S_n }\Wg([\sigma\tau^{-1}])\, p_{[\sigma]}(A)\, p_{[\tau]}(B)\,.
\end{equation}


\subsection{Generalised precursors of free cumulants}

Let $F(B\,; z)$ be a power series in $z$, whose $z^n$-coefficient belongs to the space of invariant polynomials of $B$ of degree $n$. For $n\leq N$, the Newton polynomials $\lbrace p_\alpha(B) \rbrace_{\alpha\vdash n}$ form a basis of this space: we can thus define unambiguously the coefficient of $z^n\, p_\alpha(B)$ in $F(B\,; z)$, which we denote by $[z^n\, p_\alpha(B)] F(B\,; z)$. With these conventions, we can introduce the main protagonists of this paper:

\begin{definition}
For $\alpha$ a partition of $n\leq N$, we define the corresponding ``generalised free cumulant precursor'' as the following invariant polynomial of degree $n$:
    \begin{equation}\label{eq:Kalpha2}
        K_\alpha(A) := \frac{n!}{N^{\ell(\alpha)}|\mathrm{Cl}_\alpha|} [z^n p_\alpha(B)]\, Z(A,B\,; z)\,.
    \end{equation}
\end{definition}

\noindent We recall that $|\Cl_\alpha|$ is the size of the conjugacy class of $S_n$ formed by permutations of cycle type $\alpha$. Using its explicit expression \eqref{eq:SizeClass}, we see that the above definition of $K_\alpha$ coincides with the one \eqref{eq:Kalpha} given in the Introduction (up to a slight simplification in notation, as we now omit the $N$-dependence of $K_\alpha$). Of particular importance to us will be the polynomial associated with the maximal partition $\alpha=[n]$, which we will simply denote by $K_n$. It corresponds to the coefficient of a single trace of degree $n$ in the HCIZ integral, namely
\begin{eqnarray}\label{eq:Kn2}
    K_n(A) = \frac{n}{N} [z^n \Tr(B^n)]\,Z(A,B\,; z)\,,
\end{eqnarray}
where we have used $|\Cl_{[n]}|=(n-1)!$. We will call $K_n$ a \textit{free cumulant precursor}, in anticipation of Section \ref{LargeNlim} where we will identify its large $N$ limit with the celebrated free cumulant of degree $n$. We then keep the term \textit{generalised precursor} introduced earlier to designate the polynomial $K_\alpha$ associated with an arbitrary partition $\alpha$. We finally note that the definition \eqref{eq:Kalpha2} can be rephrased as another way of organising the expansion of the HCIZ integral (at least up to order $z^N$):
\be\label{expand2}
 Z(A,B\,; z)=  \sum_{n=0}^N {z^n}   
\sum_{\alpha\vdash n}   \frac{N^{\ell(\alpha)}|\Cl_\alpha |}{n!} K_\alpha(A)  p_\alpha(B) + {\rm O}(z^{N+1})\,.
\ee

Let us now give more explicit expressions of these generalised free cumulant precursors, in terms of standard invariant polynomials. For instance, using the expansion \eqref{Z2sp} of the HCIZ integral, it is easy to express $K_\alpha$ in the basis of Schur polynomials:
\begin{equation}\label{eq:KSchur}
    K_\alpha(A) = \inv{N^{\ell(\alpha)}}\sum_{\lambda\vdash n\atop |\lambda| \le N } \frac{1}{C_\lambda}    \hat\chi_\lambda(\alpha) s_\lambda(A)\,.
\end{equation}
This expression further simplifies in the case $\alpha=[n]$. Indeed, it is a well-known fact that the only representations of $S_n$ that have a non-vanishing character on the class $[n]$ of
cyclic permutations  are
those associated to ``hook" Young diagrams, namely those diagrams with at most one line of length longer than 1.
In the notations of section \ref{notations}, they correspond to $\lambda_{n,t}=[n-t, 1^t]$, with $0\le t \le n-1$, and one has $\hat\chi_{\lambda_{n,t}}([n])=(-1)^t$. We then conclude that 
\be \label{Kn}   K_n(A)= \inv{N} \sum_{t=0}^{n-1} \frac{(-1)^t}{C_{\lambda_{n,t}}}  s_{\lambda_t}(A)\,, \qquad \text{ where } \qquad C_{\lambda_{n,t}}=N^{-n} \frac{(N+n-t-1)!}{(N-t-1)!}\, .\ee

Similarly, one can use the expansion \eqref{expand} to express the generalised precursors in terms of Newton polynomials. This gives
\begin{eqnarray}\label{eq:Kp1}
    K_\alpha(A) = \inv{N^{\ell(\alpha)}}\sum_{\lambda,\beta\vdash n\atop |\lambda| \le N} \frac{|\Cl_\beta|}{n! C_\lambda}    \hat\chi_\lambda(\alpha)\hat\chi_\lambda(\beta) p_\beta(A)\,.
\end{eqnarray}
From the alternative way \eqref{Wgexpand} to expand the HCIZ integral along Newton polynomials, we can give another equivalent expression of $K_\alpha$ in terms of Weingarten coefficients. Namely, fixing any permutation $\sigma\in S_n$ of cycle type $[\sigma]=\alpha$, we get
\begin{eqnarray}\label{eq:Kp2}
   K_{\alpha}(A)= N^{n-\ell(\alpha)}  \sum_{\tau\in S_n}  \Wg([\sigma\tau^{-1}])\,  p_{[\tau]}(A)\,.
\end{eqnarray}
The above expressions for $K_\alpha(A)$ and $K_n(A)$ form the point 2 of our main Theorem \ref{thm:summary}. The right-hand side of \eqref{eq:Kp2} can be rephrased as the convolution over the symmetric group $S_n$ of the Weingarten coefficient $\sigma\in S_n \mapsto \Wg([\sigma])$ with the function $\sigma\in S_n \mapsto p_{[\sigma]}(A)$. Up to a global power of $N$, we thus recover the matrix cumulant introduced by Capitaine and Casalis in~\cite{CC06,CC08}, showing that our construction of the generalised precursor from the HCIZ integral coincides with their definition as a convolution over $S_n$.\\

\noindent{\it Remark 2.} The identity \eqref{Kn} should be contrasted with the following one, a direct consequence of Frobenius--Schur's formula \eqref{frobenius} evaluated 
at $\alpha=[n]$,
\be   m_n(A) := \frac{1}{N} \Tr(A^n) = \frac{1}{N}  \sum_{t=0}^{n-1}  (-1)^t s_{\lambda_{n,t}}(A) \,,\ee
which holds true for any finite $N$. For example, for $n=2$
\bea m_2 &=& \inv{N}\Big(   \oh(N^2 m_1^2+N m_2) -  \oh(N^2 m_1^2-N m_2)   \Big) \,, \\
K_2 &=& \inv{N}\Big(   \oh\frac{N^2 m_1^2+N m_2}{1+\inv{N}} -  \oh\frac{N^2 m_1^2-N m_2}{1-\inv{N}}   \Big)=\frac{N^2}{N^2-1}(m_2-m_1^2)\,,\eea
showing the non-trivial role of the weight
$1/C_{\lambda_{n,t}}$ in the sum \eqref{Kn}.

\subsection{Explicit expressions for low-degree generalised precursors}

For concreteness, we gather below the explicit expressions of the precursors $K_n(A)$ up to degree 4:\vspace{-13pt} 
\begin{subequations}\label{eq:K1to4}
\begin{eqnarray}
    K_1(A) &\!=\!& \frac{1}{N} \Tr(A)\,, \\
    K_2(A) &\!=\!& \frac{1}{N^2-1}\bigl( N\, \Tr(A^2) - \Tr^2(A) \bigr)\,, \\
    K_3(A) &\!=\!& \frac{N}{(N^2-1)(N^2-4)} \bigl( N^2 \,\Tr(A^3) - 3N\,\Tr(A)\,\Tr(A^2) + 2 \Tr^3(A)\bigr)\,, \\
    K_4(A) &\!=\!& \frac{N^2}{(N^2-1)(N^2-4)(N^2-9)}\bigl(  N(N^2+1)\,\Tr(A^4) -4(N^2+1)\Tr(A)\,\Tr(A^3) \\
    & &\hspace{130pt} -(2N^2-3)\Tr^2(A^2) + 10 N\,\Tr^2(A)\,\Tr(A^2) - 5 \Tr^4(A) \bigr) \, . \nonumber
\end{eqnarray}
\end{subequations}
In addition, we also include below all the generalised precursors $K_\alpha(A)$ up to degree 3:
\begin{subequations}\label{eq:GenK23}
\begin{eqnarray}
   \hspace{-20pt} K_{(1,1)}(A) &\!=\!& -\frac{1}{N(N^2-1)}\bigl( \Tr(A^2) - N\,\Tr^2(A) \bigr)\,, \\
    \hspace{-20pt} K_{(2,1)}(A) &\!=\!& -\frac{1}{(N^2-1)(N^2-4)} \bigl( 2N\,\Tr(A^3) - (N^2+2)\,\Tr(A)\,\Tr(A^2) + N\, \Tr^3(A)\bigr)\,, \\
    \hspace{-20pt} K_{(1,1,1)}(A) &\!=\!& \frac{1}{N(N^2-1)(N^2-4)} \bigl( 4\Tr(A^3) - 3N\,\Tr(A)\,\Tr(A^2) + (N^2-2)\Tr^3(A)\bigr)\,.
\end{eqnarray}
\end{subequations}


\section{Topological expansion of generalised precursors}
\label{LargeNlim}

\subsection[Large $N$ limit and free cumulants]{Large $\boldsymbol{N}$ limit and free cumulants}
\label{sec:FreeCum}

We now discuss the large $N$ limit of the generalised precursors $K_\alpha$ introduced in the previous section. More precisely, this subsection is devoted to the point 4 of our main Theorem \ref{thm:summary}, stating that this limit is directly expressed in terms of the well-known free cumulants introduced in~\cite{BIPZ,Spe93}. 
To do so, we first need to recall more precisely what is meant 
by the large $N$ limit of $N\times N$ matrices.

\paragraph{Large $\boldsymbol{N}$ limit of matrices.} Let $A$ be a matrix of size $N\times N$, with eigenvalues $\lambda_1, \dots, \lambda_N$. We recall that its $k$-th moment is given by the normalised trace
\begin{eqnarray}
    m_k(A) := \frac{1}{N} \Tr(A^k) = \frac{1}{N} \sum_{i=1}^N \lambda_i^k\,,
\end{eqnarray}
\ie by the average of the $k$-th power of the eigenvalues. In this section, we will be interested in sequences $(A_N)_{N\geq 1}$ of such matrices of growing size $N$ and admitting a good limiting eigenvalue distribution when this size becomes large. For our purposes, this will essentially mean that the moments $m_k(A_N)$ have a finite limit when $N\to\infty$. To ligthen the formulae, we will omit the explicit $N$-dependence in the matrices $A_N$ and use the same notation $A$ as in the previous section. In practice, we will only manipulate invariant polynomials in $A$. When writing the large $N$ limit of such a polynomial, we will implicitly mean re-expressing it in terms of the moments $m_k(A)$ and taking $N\to\infty$ while keeping all the $m_k(A)$ finite. For $\alpha=(\alpha_1,\dots,\alpha_{\ell(\alpha)})$ a partition, it will be useful to introduce
\begin{eqnarray}\label{defmalpha}
 m_\alpha(A) := \prod_{i=1}^{\ell(\alpha)} m_{\alpha_i}(A) = \frac{1}{N^{\ell(\alpha)}} p_\alpha(A)\,.   
\end{eqnarray}

\paragraph{From precursors to free cumulants.} To illustrate these ideas in the context of free cumulant precursors, let us investigate the large $N$ limit of the first few examples. The precursors $K_n(A)$ up to degree 4 have been written in equation \eqref{eq:K1to4} in terms of traces of powers of $A$. Translating to the normalised moments $m_k(A)=\frac{1}{N}\Tr(A^k)$, we then have
\bea 
\nonumber K_1&\!\!\!=\!\!\!& m_1\, 
\\
\nonumber K_2&\!\!\!=\!\!\!&  \frac{N^2}{N^2-1} ( m_2 -m_1^2)\, 
\\ \label{Kmrel} 
K_3 &\!\!\!=\!\!\!&  \frac{N^4}{(N^2-1)(N^2-4)} (m_3-3 m_2 m_1+2m_1^3)\,,
\\
\nonumber K_4 &\!\!\!=\!\!\!&  \frac{N^4}{(N^2-1)(N^2-4)(N^2-9)}\big(  (N^2+1)(m_4 -4 m_3 m_1-2m_2^2 +10 m_2 m_1^2-5 m_1^4) +5  (m_2-m_1^2)^2\big)\,, 
\eea
where we omitted the $A$-dependence of $K_n$ and $m_n$ to simplify the expressions. (See a few more terms in Appendix \ref{TableKkappam}.) From these formulae, it is easy to take the large $N$ limits of the precursors. We find that they are finite and read
\bea \nonumber
 \lim_{N\to\infty}K_1&\!\!\!=\!\!\!& \kappa_1 := m_1\, 
 \\  \nonumber
\lim_{N\to\infty}K_2&\!\!\!=\!\!\!&  \kappa_2 := m_2 -m_1^2\, \\   \nonumber
\lim_{N\to\infty}K_3 &\!\!\!=\!\!\!& \kappa_3 := m_3-3 m_2 m_1+2m_1^3\,, \\ \nonumber
\lim_{N\to\infty}K_4 &\!\!\!=\!\!\!&  \kappa_4 := m_4 -4 m_3 m_1-2m_2^2 +10 m_2 m_1^2-5 m_1^4\,.    \nonumber
\eea
The attentive reader will have recognised in these limits $\kappa_n$ the free cumulants of degree 1 to 4. We note that up to degree 3, $K_n$ coincides with $\kappa_n$ up to a global $N$-dependent factor: this is a low-degree accident, which stops from $n=4$.

The above result in fact holds for any $n$, namely
\begin{eqnarray}\label{eq:limKn}
    \lim_{N\to\infty} K_n = \kappa_n := \sum_{\beta\vdash n} (-1)^{1+\ell(\beta)}\frac{(n+\ell(\beta)-2)!}{(n-1)!\prod_{k=1}^n \hat\beta_k!}\, m_\beta
\end{eqnarray}
is the free cumulant of degree $n$, justifying the name of \textit{free cumulant precursors} for the functions $K_n$. The second equality, \ie the definition of the free cumulant $\kappa_n$ in terms of the moments,  
results from the inversion of the well-known 
formula of Kreweras\cite{Krw}
\begin{eqnarray}\label{eq:Krewe}
    m_n = \sum_{\alpha\vdash n} 
   \frac{n!}{(n+1-\ell(\alpha))! \prod_{k=1}^n \hat \alpha_k!} \, \prod_{i=1}^{\ell(\alpha)} \kappa_{\alpha_i} 
   \,.
\end{eqnarray}

The fact that $K_n$ tends to $\kappa_n$ when $N\to\infty$ 
 was first established
in \cite{IZ80}, making use of the explicit form of the integral $Z$  as a ratio of determinants and of  the differential equation it satisfies. 
(In this reference, the free cumulants $\kappa_n$ are called ``connected planar correlation functions''.)
It was then rederived using the formalism of the dispersionless Toda hierarchy in \cite{ZJ02, ZJZ03}. A diagrammatic proof of that property may also be devised, along the lines used in \cite{ZJZ03}. 
Independently it was elaborated under the assumption that the matrix $B$ is of rank 1, or more
generally of rank ${\rm o}(N)$, by several authors, using a variety of approaches \cite{ZJ99, Collins03, GM, Tan, CollinsSn07}, see also \cite{PB}.
 In the large $N$ limit, under that finite-rank assumption in $B$, the dominant terms
in the expansion of $\log Z$ indeed contain only single traces of $B$. 
The merit of these alternative approaches is that they do not make use of the explicit form of the integral \eqref{hciz} and thus extend to other
cases, like the integration over the orthogonal group, for which no compact determinantal form exists. 
(This had been anticipated in an earlier work with application to spin glasses \cite{MPR}.)

\paragraph{Large $\boldsymbol{N}$ limit of generalised precursors.} Let us now turn to the generalised precursors $K_\alpha$ (associated with arbitrary partitions $\alpha$). For degree 2 and 3, they were expressed in terms of traces in equation \eqref{eq:GenK23}. Switching to moments, we find
\begin{subequations}
\begin{eqnarray}
   \hspace{-20pt} K_{(1,1)} &\!=\!& \frac{1}{(N^2-1)}\bigl( N^2m_1^2 - m_2 \bigr)\,, \label{eq:K11}\\
    \hspace{-20pt} K_{(2,1)} &\!=\!& \frac{N^2}{(N^2-1)(N^2-4)} \bigl((N^2+2)\,m_1 m_2 - N^2 m_1^3 - 2 m_3\bigr)\,, \\
    \hspace{-20pt} K_{(1,1,1)} &\!=\!& \frac{1}{(N^2-1)(N^2-4)} \bigl( N^2(N^2-2)m_1^3 - 3N^2 m_1 m_2 + 4m_3 \bigr)\,.
\end{eqnarray}
\end{subequations}
The large $N$ limit can then be taken readily, yielding again a finite result in terms of free cumulants:
\begin{equation*}
    \lim_{N\to\infty} K_{(1,1)} = m_1^2 = \kappa_1^2\,, \qquad \lim_{N\to\infty} K_{(2,1)} = m_1(m_2-m_1^2) = \kappa_1 \kappa_2\,,\qquad \lim_{N\to\infty} K_{(1,1,1)} = m_1^3 = \kappa_1^3\,.
\end{equation*}
This generalises to higher degrees as the following ``multiplicative'' behavior:
\begin{eqnarray}\label{eq:MultKappa}
    \lim_{N\to\infty} K_\alpha = \kappa_\alpha := \prod_{i=1}^{\ell(\alpha)} \kappa_{\alpha_i}\,.
\end{eqnarray}
This result can be proven by exploiting the well-known fact that the logarithm $\log Z(A,B\,; z)$ of the HCIZ integral is of order ${\rm O}(N^2)$ in the large $N$ limit. Due to the slightly technical nature of this proof, we present it in Appendix \ref{App:MultK}. We note that this multiplicative behavior only holds in the large $N$ limit: at finite $N$, $K_\alpha$ differs from $\prod_{i=1}^{\ell(\alpha)} K_{\alpha_i}$.

Further results, such as the expression of the first generalised precursors in terms of the free cumulants, are presented in Appendix \ref{TableKkappam}.


\subsection[$1/N^2$ corrections and Hurwitz numbers]{$\boldsymbol{1/N^2}$ corrections and Hurwitz numbers}
\label{sec:Hur}

In the previous subsection, we have explained that the large $N$ limit of the generalised precursors $K_\alpha$ is directly related to free cumulants -- see equations \eqref{eq:limKn} and \eqref{eq:MultKappa}. It is then natural to search for the $1/N$ corrections to these results. To answer that question, we first need to discuss the relation between the HCIZ integral and the so-called monotone Hurwitz numbers. 

\paragraph{Monotone Hurwitz numbers.} Recall that the large $N$ limit of the matrices $A$ and $B$ is taken keeping the moments $m_k(A)$ and $m_k(B)$ finite. Quite remarkably, the large $N$ expansion of the HCIZ integral $Z(A,B\,; z)$ in terms of these moments can be seen as the generating function for combinatorial objects called monotone Hurwitz numbers~\cite{GGPN}.

\begin{theorem}\label{thm:Hur}\hspace{-2pt}{\rm{\cite{GGPN}}}\ 
    The coefficients in the power series expansion
    \begin{equation}\label{eq:PowerZ}
        Z(A,B\,; z) = \sum_{n=0}^{+\infty} Z_n(A,B)\,z^n
    \end{equation}
    of the HCIZ integral admits the following large $N$ expansion:\footnote{As explained earlier, the power series expansion \eqref{eq:PowerZ} of the HCIZ integral is absolutely convergent for $z\in\C$. The $1/N$-expansion \eqref{eq:TopZ} of the coefficient of $z^n$ is also convergent, but with a radius of convergence which decreases with $n$ and approaches zero as $n\to\infty$. As such, the large $N$ expansion of the HCIZ integral itself is not convergent and should be understood as an asymptotic series. We will not require such considerations in this work as we will only manipulate the expansion of each $z^n$-coefficient separately.}
    \begin{equation}\label{eq:TopZ}
        Z_n(A,B) =\frac{1}{n!} \sum_{\alpha,\beta\vdash n}(-1)^{\ell(\alpha)+\ell(\beta)}\sum_{g \geq g_{\alpha,\beta}}^\infty N^{2-2g} 
H_g^{\bullet\hspace{1pt},\hspace{1pt}\leq}(\alpha,\beta) \, m_\alpha(A)\, m_\beta(B)\,,
    \end{equation}
where $g_{\alpha\beta}=\text{max}(1-\ell(\alpha),1-\ell(\beta))$ and $H_g^{\bullet\hspace{1pt},\hspace{1pt}\leq}(\alpha,\beta)$ is the \emph{disconnected monotone double Hurwitz number} of genus $g$ and cycle types $\alpha,\beta \vdash n$. The latter is defined as the number of tuples $(\sigma_\alpha,\sigma_\beta,\tau_1,\dots,\tau_r)$ of permutations in $S_n$ such that:
\begin{enumerate}[(i)]
    \item $\sigma_\alpha \sigma_\beta \tau_1 \cdots \tau_r = \text{Id}$ ;
    \item $r=2g-2+\ell(\alpha)+\ell(\beta)$ ;
    \item $\sigma_\alpha$ and $\sigma_\beta$ have cycle types $[\sigma_\alpha]=\alpha$ and $[\sigma_\beta]=\beta$ ;
    \item $\tau_i$ are transpositions ;
    \item writing $\tau_i=(a_i\,b_i)$ with $a_i < b_i$, we have $b_1 \leq \cdots \leq b_r$.
\end{enumerate}
\end{theorem}

\begin{remark}\label{Rmk:Hur}
    The various adjectives used to designate the numbers $H_g^{\bullet\hspace{1pt},\hspace{1pt}\leq}(\alpha,\beta)$ reflect the existence of many different types of Hurwitz numbers considered in the literature.

    The qualifier \textit{double} refers to the fact that $H_g^{\bullet\hspace{1pt},\hspace{1pt}\leq}(\alpha,\beta)$ depends on two partitions $(\alpha,\beta)$. There also exist simple, single or multiple versions of these numbers, depending on 0, 1 or an arbitrary number of such partitions (and which are defined by removing or adding permutations of the corresponding cycle types in the above definition). In this paper, we will mostly be concerned with double Hurwitz numbers. 

    The Hurwitz numbers appearing in Theorem \ref{thm:Hur} are called \textit{monotone} (or sometimes weakly monotone) due to the condition $b_1 \leq \cdots \leq b_r$ imposed in point (v). Relaxing this condition yields the \textit{standard} Hurwitz numbers $H_g^{\bullet}(\alpha,\beta)$, which were extensively studied in the literature but will not be of concern to us for this paper. However, we will later encounter their \textit{strictly monotone} variants $H_g^{\bullet\hspace{1pt},\hspace{1pt}<}(\alpha,\beta)$, defined by switching to strict inequalities $b_1 < \cdots < b_r$ in the condition (v).

    In Theorem \ref{thm:Hur}, we have focused on the algebraic characterisation of Hurwitz numbers, counting solutions of equations in symmetric groups. They also admit a geometric interpretation, counting ramified coverings of the Riemann sphere with $r$ simple branch points and two branch points of cycle types $\alpha$ and $\beta$ (together with appropriate conditions on labelings of sheets to distinguish between standard, monotone and strictly monotone Hurwitz numbers). In this geometric context, the integer $g$ corresponds to the \textit{genus} of the covering surface and the large $N$ expansion \eqref{eq:TopZ} is then also called a \textit{topological expansion}.

    The symbol $\bullet$ in the notations $H_g^{\bullet}(\alpha,\beta)$, $H_g^{\bullet\hspace{1pt},\hspace{1pt}\leq}(\alpha,\beta)$ and $H_g^{\bullet\hspace{1pt},\hspace{1pt}<}(\alpha,\beta)$ indicates that we considered \textit{disconnected} Hurwitz numbers, which count coverings of the sphere by potentially disconnected surfaces. In Theorem \ref{thm:Hur}, the disconnected nature of $H_g^{\bullet\hspace{1pt},\hspace{1pt}\leq}(\alpha,\beta)$ is the reason why the topological expansion \eqref{eq:TopZ} also contains terms of ``negative genus'' $g<0$ (since the minimum genus $g_{\alpha\beta}$ can generally be negative). This is a slight abuse of language, where the genus of a disconnected surface $S=\bigsqcup_i\, S_i$ is defined through the additivity of the Euler characteristic $2-2g(S)=\sum_i (2-2g(S_i))$, allowing for cases with negative $g(S)$. We note that $H_g^{\bullet\hspace{1pt},\hspace{1pt}\leq}(\alpha,\beta)$ vanishes if $g<g_{\alpha\beta}$, translating the fact that covers with a prescribed ramification profile have a minimal genus.

    It is quite natural to also define \textit{connected} Hurwitz numbers $H_g(\alpha,\beta)$, $H_g^{\leq}(\alpha,\beta)$ and $H_g^{<}(\alpha,\beta)$, by adding the condition of connectivity of the covering surface. In the algebraic language, this amounts to requiring that the subgroup of $S_n$ generated by $(\sigma_\alpha,\sigma_\beta,\tau_1,\dots,\tau_r)$ acts transitively on $\lbrace 1,\dots,n\rbrace$. The connected monotone Hurwitz numbers $H_g^{\leq}(\alpha,\beta)$ are then encoded in the topological expansion of the HCIZ free energy, \ie the logarithm of $Z(A,B\,; z)$. The fact that these numbers are non-zero only for non-negative genus is then equivalent to the existence of a finite limit of $\frac{1}{N^2} \log Z(A,B\,; z)$ when $N\to\infty$, as mentioned in the previous subsection.

    To simplify the formulations, we will often omit the adjectives double and disconnected when referring to Hurwitz numbers, since their corresponding variants do not play a role in this paper. In contrast, we will mostly keep track of the monotone or strictly monotone nature of these numbers, as we will encounter them both in what follows.

    Finally, we note that Hurwitz numbers are sometimes defined in the literature with a prefactor $1/n!$ (dividing by the size of the symmetric group $S_n$), in which case they are generally rational numbers. In this paper, we follow an alternative convention with no such prefactor, such that Hurwitz numbers are always non-negative integers.
\end{remark}

The monotone and strictly monotone Hurwitz numbers $H_g^{\bullet\hspace{1pt},\hspace{1pt}\leq}(\alpha,\beta)$ and $H_g^{\bullet\hspace{1pt},\hspace{1pt}<}(\alpha,\beta)$ will play an important role in this subsection. We note that they are symmetric with respect to the exchange of $\alpha$ and $\beta$. For explicitness, we tabulate the monotone Hurwitz numbers (or rather their generating functions) for partitions up to size 4 in Appendix \ref{HurDisconnect}. 

\paragraph{Topological expansion of generalised precursors.} Recall the definition \eqref{eq:Kalpha2} of the generalised precursor $K_\alpha(A)$ from $Z(A,B\,; z)$. Since $m_\alpha(B)=N^{-\ell(\alpha)}p_\alpha(B)$, it is easy to extract the topological expansion of $K_\alpha(A)$ from that \eqref{eq:PowerZ}--\eqref{eq:TopZ} of the HCIZ integral. This gives the following proposition and proves the point 5 in our main Theorem \ref{thm:summary}.

\begin{proposition}
    The generalised precursor $K_\alpha(A)$ admits the 
    topological expansion
    \begin{equation}\label{eq:TopK}
        K_\alpha(A) = \frac{1}{|\mathrm{Cl}_\alpha|} \sum_{g \geq 1-\ell(\alpha)}  N^{2(1-\ell(\alpha)-g)} \left(\sum_{\beta \vdash n} (-1)^{\ell(\alpha)+\ell(\beta)} H_g^{\bullet\hspace{1pt},\hspace{1pt}\leq}(\alpha,\beta)\,  m_\beta(A) \right)\,,
    \end{equation}
    in terms of monotone Hurwitz numbers.\footnote{Note that, as in Theorem \ref{thm:Hur}, we could have restricted the genus expansion to $g\geq g_{\alpha,\beta} = \text{max}\bigl(1-\ell(\alpha),1-\ell(\beta)\bigr)$. For simplicity, we chose here to use a common bound for $g$, independent of $\beta$. When $\beta$ is such that $\ell(\beta)<\ell(\alpha)$, the expansion of $[m_\beta(A)]K_\alpha(A)$ can thus be refined by removing the terms with $1-\ell(\alpha) \leq g < 1-\ell(\beta)$.}
\end{proposition}

We note that the powers of $N$ in the expansion \eqref{eq:TopK} are always even and non-positive. The latter property is consistent with the finiteness of $K_\alpha(A)$ in the large $N$ limit. In fact, the above expansion gives an alternative expression for this limit:
\begin{equation}\label{eq:LimKalphaH}
    \lim_{N\to\infty} K_\alpha(A) = \frac{1}{|\mathrm{Cl}_\alpha|} \sum_{\beta \vdash n} (-1)^{\ell(\alpha)+\ell(\beta)} H_{1-\ell(\alpha)}^{\bullet\hspace{1pt},\hspace{1pt}\leq}(\alpha,\beta)\,  m_\beta(A)\,.
\end{equation}
Of particular interest is the case $[\alpha]=n$, \ie that of the precursor $K_n(A)$:
\begin{equation}
    \lim_{N\to\infty} K_n(A) = \frac{1}{(n-1)!} \sum_{\beta \vdash n} (-1)^{1+\ell(\beta)} H_{0}^{\bullet\hspace{1pt},\hspace{1pt}\leq}([n],\beta)\,  m_\beta(A)\,.
\end{equation}
Recall from equation \eqref{eq:limKn} that we identified this limit with the free cumulant $\kappa_n$. The right-hand side of the above equation then provides an alternative expression for this free cumulant in terms of monotone Hurwitz numbers $ H_{0}^{\bullet\hspace{1pt},\hspace{1pt}\leq}([n],\beta)$, which already appeared in the work~\cite{N15}. Comparing this expression with the one in \eqref{eq:limKn}, we then deduce the following formula:
\begin{equation}\label{eq:H0n}
    H_{0}^{\bullet\hspace{1pt},\hspace{1pt}\leq}([n],\beta) = \frac{(n+\ell(\beta)-2)!}{\prod_{k=1}^n \hat\beta_k!}\,.
\end{equation}
We note that these disconnected monotone Hurwitz numbers $H_{0}^{\bullet\hspace{1pt},\hspace{1pt}\leq}([n],\beta)$ in fact also coincide with the connected ones $H_{0}^{\leq}([n],\beta)$.\footnote{Indeed, $H_{0}^{\bullet\hspace{1pt},\hspace{1pt}\leq}([n],\beta)$ -- resp. $H_{0}^{\leq}([n],\beta)$ -- is extracted from the coefficient of $N^2 m_n(A)m_\beta(B)$ in $Z(A,B\,; z)$ -- resp. $\log Z(A,B\,; z)$. Since we focus on a ``single trace'' $m_n(A)=\frac{1}{N}\Tr(A^n)$ of $A$, the passage to the logarithm does not affect this coefficient and thus these particular disconnected and connected Hurwitz numbers agree.} 
To the best of our knowledge, the expression \eqref{eq:H0n}  of that monotone Hurwitz number is a new result: it would be interesting to obtain a direct combinatorial proof of this formula. Similarly, one can extract the monotone Hurwitz numbers $H_{1-\ell(\alpha)}^{\bullet\hspace{1pt},\hspace{1pt}\leq}(\alpha,\beta)$ by comparing the expressions \eqref{eq:MultKappa} and \eqref{eq:LimKalphaH} of the large $N$ limit of the generalised precursor $K_\alpha$.

Having identified the dominant term ${\rm O}(N^0)$ in the topological expansion \eqref{eq:TopK} of $K_\alpha$ with $\kappa_\alpha$, the remaining terms can thus be seen as explicit $1/N^2$ corrections to the limits \eqref{eq:limKn} and \eqref{eq:MultKappa}. 

\paragraph{The generalised moment-cumulant relation.} Recall the Kreweras relation \eqref{eq:Krewe} expressing moments in terms of free cumulants. This can be compactly rewritten as
\begin{eqnarray}\label{eq:MomCumRel}
    m_n = \sum_{\alpha\vdash n} \text{NC}(\alpha)\,\kappa_\alpha\,,
\end{eqnarray}
where $\text{NC}(\alpha)$ denotes the number of non-crossing set-partitions of $\lbrace 1,\dots,n\rbrace$ formed by $\ell(\alpha)$ subsets of sizes $\alpha_1,\dots,\alpha_{\ell(\alpha)}$. The formula \eqref{eq:MomCumRel} is called the \textit{moment-cumulant relation} of free probability.

It is natural to wonder if there exists a finite $N$ version of this relation, involving the generalised precursors $K_\alpha$ instead of their large $N$ limits $\kappa_\alpha$. Such a generalisation can be derived formally starting from equation \eqref{eq:Kp1} and making use of orthogonality relations \eqref{eq:OrthoChi} of symmetric characters:
\begin{equation}
    m_\alpha = \sum_{\beta,\lambda\vdash n} N^{\ell(\beta)-\ell(\alpha)} \frac{C_\lambda \, | \Cl_\beta|}{n!} \hat\chi_\lambda(\alpha)\hat\chi_\lambda(\beta) K_\beta\,.
\end{equation}
The complicated $N$-dependence of this identity is hidden in the coefficient $C_\lambda$, as defined in \eqref{eq:dC}. In particular, it is not obvious how this relation reduces to the moment-cumulant one \eqref{eq:MomCumRel} in the large $N$ limit, nor how to derive the corresponding $1/N$ corrections. Remarkably, this can be done by inverting the topological expansion \eqref{eq:TopK} in terms of \textit{strictly monotone} Hurwitz number.

\begin{proposition}\label{Prop:TopMK}
    Let $n\geq 1$ and $\alpha\vdash n$. We have
    \begin{equation}\label{eq:TopMK}
        m_\alpha(A) = \frac{1}{|\mathrm{Cl}_\alpha|} \sum_{g \geq 1-\ell(\alpha)}  N^{2(1-\ell(\alpha)-g)} \left(\sum_{\beta \vdash n} H_g^{\bullet\hspace{1pt},\hspace{1pt}<}(\alpha,\beta)\,  K_\beta(A) \right)\,,
    \end{equation}
    where $H_g^{\bullet\hspace{1pt},\hspace{1pt}<}(\alpha,\beta)$ is the \emph{disconnected strictly monotone double Hurwitz number} of genus $g$ and cycle types $\alpha,\beta \vdash n$, as defined above in Remark \ref{Rmk:Hur}.
\end{proposition}

\begin{remark}
    For a given choice of $\alpha,\beta\vdash n$, the strictly monotone Hurwitz numbers $H_g^{\bullet\hspace{1pt},\hspace{1pt}<}(\alpha,\beta)$ vanish for genera such that $2g>n+1-\ell(\alpha)-\ell(\beta)$, so that the expansion \eqref{eq:TopMK} is finite. 
\end{remark}

\begin{proof}
    Let us define the series
    \begin{equation}\label{eq:GenHsm}
        H^{\bullet\hspace{1pt},\hspace{1pt}<}(\alpha,\beta) = \frac{1}{|\Cl_\alpha|} \sum_{g \geq g_{\alpha\beta}} N^{2-2g-\ell(\alpha)-\ell(\beta)} H_g^{\bullet\hspace{1pt},\hspace{1pt}<}(\alpha,\beta)
    \end{equation}
    and
    \begin{equation}\label{eq:GenHm}
        H^{\bullet\hspace{1pt},\hspace{1pt}\leq}(\alpha,\beta) = \frac{1}{|\Cl_\alpha|} \sum_{g \geq g_{\alpha\beta}} (-N)^{2-2g-\ell(\alpha)-\ell(\beta)} H_g^{\bullet\hspace{1pt},\hspace{1pt}\leq}(\alpha,\beta)\,.
    \end{equation}
    These are generating functions for the strictly monotone and monotone Hurwitz numbers of type $(\alpha,\beta)$, with expansion parameter $1/N$ (we note that the power $r=2g-2+\ell(\alpha)+\ell(\beta)$ of $1/N$ appearing in these series keeps track of the number of transpositions in the algebraic definition of Hurwitz numbers in Theorem \ref{thm:Hur} and is thus a natural alternative to the genus). In terms of the monotone generating function, the equation \eqref{eq:TopK} is rewritten as
    \begin{equation}\label{eq:Km}
        K_\alpha = \sum_{\beta\vdash n} N^{\ell(\beta)-\ell(\alpha)} H^{\bullet\hspace{1pt},\hspace{1pt}\leq}(\alpha,\beta)\, m_\beta\,.
    \end{equation}
    The prefactors in equations \eqref{eq:GenHsm} and \eqref{eq:GenHm} might seem slightly strange at first sight, as they break the symmetry $\alpha \leftrightarrow \beta$. However, they have been choosen to ensure a remarkable property (see for instance~\cite[Lemma 2.9]{BCGLS21}), namely the fact that the matrices $\bigl( H^{\bullet\hspace{1pt},\hspace{1pt}<}(\alpha,\beta) \bigr)_{\alpha,\beta\vdash n}$ and $\bigl( H^{\bullet\hspace{1pt},\hspace{1pt}\leq}(\alpha,\beta)\bigr)_{\alpha,\beta\vdash n}$ are inverse one of another, \ie
    \begin{equation}\label{eq:InvH}
        \sum_{\gamma\vdash n} H^{\bullet\hspace{1pt},\hspace{1pt}<}(\alpha,\gamma)  H^{\bullet\hspace{1pt},\hspace{1pt}\leq}(\gamma,\beta) = \delta_{\alpha\beta}\,.
    \end{equation}
    It is then straightforward to invert the relation \eqref{eq:Km}, yielding
    \begin{equation}\label{malpha}
        m_\alpha = \sum_{\beta\vdash n} N^{\ell(\beta)-\ell(\alpha)} H^{\bullet\hspace{1pt},\hspace{1pt}<}(\alpha,\beta)\,K_\beta\,.
    \end{equation}
    Re-inserting the definition \eqref{eq:GenHsm} of $H^{\bullet\hspace{1pt},\hspace{1pt}<}(\alpha,\beta)$ then gives the desired result \eqref{eq:TopMK}.
\end{proof}

To further make the link with the moment-cumulant relation at large $N$, we need to introduce a few new concepts.

\begin{definition}
    Let $\sigma\in S_n$ be a permutation. The number of cycles of $\sigma$ is equal to $\ell([\sigma])$, \ie the length of the partition $[\sigma]$ describing the cycle type of $\sigma$. We then define the \emph{genus} of $\sigma$ by\cite{Jacques}
    \begin{equation}\label{eq:GenPerm}
        2g(\sigma) = n+1 - \ell([\sigma]) - \ell([\sigma^{-1}\zeta_n])\,,
    \end{equation}
    where $\zeta_n=(1\dots n)$ is the canonically ordered $n$-cycle. We further denote by $P_g(\alpha)$ the number of permutations of cycle type $\alpha$ and genus $g$.
\end{definition}

The computation of the numbers $P_g(\alpha)$, \ie the enumeration of permutations by cycle type and genus, is a standard problem in combinatorics, which has been well studied in the literature~\cite{Cori75,CoriH, Hock}. We now claim the following result, proven in the Appendix \ref{App:GenPerm}, which relates this problem to strictly monotone Hurwitz numbers.

\begin{lemma}\label{Lem:GenPermHur}
    Let $g\geq 0$ and $\alpha\vdash n$. We have \begin{equation}
        P_g(\alpha) = \dfrac{1}{(n-1)!} H_g^{\bullet\hspace{1pt},\hspace{1pt}<}([n],\alpha)\,.
    \end{equation}
\end{lemma}

\noindent Combining Proposition \ref{Prop:TopMK} and Lemma \ref{Lem:GenPermHur}, we obtain the main result of this paragraph, forming the point 6 of Theorem \ref{thm:summary}:

\begin{corollary}\label{Coroll1}
    Let $n\geq 1$. We have
    \begin{equation}\label{eq:TopMnK}
        m_n(A) = \sum_{g \geq 0}  N^{-2g} \sum_{\alpha\vdash n} P_g(\alpha)\,  K_\alpha(A)\,.
    \end{equation}
\end{corollary}
Recall that in the large $N$ limit, $K_\alpha$ becomes $\kappa_\alpha$. Moreover, it is well-known~\cite{Cori75} that the number of permutations of cycle type $\alpha$ and genus 0 coincides with the number of non-crossing set-partitions of size $\alpha$, \textit{i.e.,} $P_0(\alpha)=\text{NC}(\alpha)$. Thus,  equation \eqref{eq:TopMnK} reduces to the standard moment-cumulant relation \eqref{eq:MomCumRel} in the large $N$ limit and provides $1/N^2$ corrections to this relation. Recall, however, that for finite $N$, the generalised precursor $K_\alpha$ does not factorise into $\prod_{i=1}^{\ell(\alpha)} K_{\alpha_i}$: in contrast to the equation \eqref{eq:MomCumRel}, the right-hand side of the generalised moment-cumulant relation \eqref{eq:TopMnK} is not directly expressed as a polynomial  
in the $K_m$'s.


\section[Further properties of the $K_\alpha$]{Further properties of the $\boldsymbol{K_\alpha}$}

\subsection{Behaviour under shifts by the identity and removing singletons}
\label{sec:shift}

Writing $A=\hat A + m_1(A) \,\mathbb{I}$ with $\hat A$ traceless, we have 
\be\label{trivial-id}Z(\hat A +m_1(A) \,\mathbb{I}, B\,; z)= Z(\hat A,B\,; z)\, e^{N z\, m_1(A) \,\Tr B}\,. \ee
Therefore, upon identification of the coefficient of $\Tr(B^n)$
\be\label{shift} K_n(A) = K_n(\hat A) + m_1(A)\,\delta_{n,1}\,, \qquad \qquad 1\le n \le N \ \,,\ee
and  if the partition 
$\alpha$ is singleton-free (\ie contains no 1):
\be\label{KAA} 
K_\alpha(A)= K_\alpha(\hat A)\,.\ee
More generally, any partition $\alpha$ can be written as the concatenation $[\alpha_{\text{sf}},1^{\ah_1}]$ of a singleton-free partition $\alpha_{\text{sf}}$ and $\ah_1$ singletons. Projecting the identity \eqref{trivial-id} on $p_\alpha(B)$ then gives
\begin{equation}\label{KAAa}
    K_{\alpha}(A) = K_\alpha(\hat A) + \sum_{r=1}^{\ah_1} {\ah_1 \choose r} m_1^r(A) K_{[\alpha_{\text{sf}},1^{\ah_1-r}]}(\hat A)\,.
\end{equation}
Knowledge of the generalised precursors $K_\alpha(\hat A)$ on traceless matrices is therefore sufficient to fully reconstruct their expressions on any matrix.\\

Now, still with $\hat A$ traceless, $Z(\hat A, B\,; z)$
is invariant under a shift of $B$ by a multiple of the identity: 
\begin{equation*}
    Z(\hat A, B+ b\,\mathbb{I}\,; z)= Z(\hat A,B\,; z)\,.
\end{equation*}
This entails identities between the $K_\alpha(\hat A)$'s.
Indeed, fixing $\alpha=(\alpha_1,\dots, \alpha_\ell)$ and identifying the coefficient of $b\, p_\alpha(B)$  in 
$Z(\hat A, B +  b\,\mathbb{I}\,; z)$, we eventually find 
\be\label{identity1}
K_{(\alpha_1,\dots, \alpha_\ell, 1)}(\hat A) =-\inv{N^2}
\sum_{i=1}^\ell \alpha_i\, K_{(\alpha_1,\dots, \alpha_i+1,\dots, \alpha_\ell)}(\hat A) \,,\ee
hence allowing to recursively remove all singletons 
from the partition $\alpha$. For example,
\begin{equation*}
K_{[2,1]}(\hat A)=-\frac{2}{N^2}K_3(\hat A)\,,\qquad\quad K_{[3,2,1]}(\hat A)=-\frac{1}{N^2}\bigl(3 K_{[4,2]}(\hat A)+2 K_{[3^2]}(\hat A)\bigr)
\end{equation*} 
\begin{equation*}
K_{[2,1^2]}(\hat A) = -\frac{1}{N^2}\bigl(2 K_{[3,1]}(\hat A) +K_{[2^2]}(\hat A) \bigr) = \frac{6}{N^4} K_{4}(\hat A)-\frac{1}{N^2}K_{[2^2]}(\hat A)  \,,
\end{equation*}
and so on and so forth.


\subsection{The generating function of the precursors $K_n$}
\label{genfn}

We now prove the point 7 of the main theorem \ref{thm:summary}, by constructing a generating function for the polynomials $K_n(A)$. To do so, consider the integral 
\begin{equation}\label{eq:IntZRes}
    \frac{1}{N}\int_{\U(N)} DU\,  Z(A,U^\dagger\,; z)\, \Tr\left(\inv{1-x U} \right)
\end{equation}
in which the (normalized) resolvent of the matrix $U$ has been inserted against the HCIZ integral and integration over $U \in \U(N)$ carried out. For $|x|<1$, we can expand the resolvent as
\begin{equation}
    \Tr\left(\inv{1-x U} \right) = \sum_{n=0}^{+\infty} x^n\, p_{[n]}(U)\,,
\end{equation}
with $p_{[n]}(U)=\Tr(U^n)$, while the HCIZ integral admits the expansion \eqref{expand2}. Recalling the orthogonality property \eqref{eq:OrthoP} of the Newton polynomials $p_\alpha$, it is clear that the above integral yields
\begin{equation}
    \frac{1}{N}\int_{\U(N)} DU\,  Z(A,U^\dagger\,; z)\, \Tr\left(\inv{1-x U} \right) = 1 + \sum_{n=1}^N (x z)^n K_n(A) + {\rm O}(x^{N+1})\,,
\end{equation}
hence providing a generating function for the precursors $K_n(A)$. Let us now simplify the expression of this function. Firstly, it is clear that we can set $z=1$ without losing its generating property. Secondly, we now reinsert the integral expression \eqref{hciz} of $Z(A,B\,; 1)$ in the left-hand side to obtain
\begin{equation}
    \frac{1}{N}\int_{\U(N)} DU \int_{\U(N)} DV\,  e^{N \Tr(AV \Ud \Vd)}\, \Tr\left(\inv{1-x U} \right)\,.
\end{equation}
Performing the change of variables $U \mapsto V^\dagger U V$, which leaves the Haar measure $DU$ invariant, $V$ disappears and the integration over it can thus be carried out readily (recalling that we have normalised the total Haar measure to 1). Putting things together, we then obtain
\begin{equation}
   \mathcal{K}^{(N)}(A;x) := \frac{1}{N} \,  \int_{\U(N)} DU\, e^{N\Tr(A\Ud)}\, \Tr\left(\inv{1-x U} \right) = 1 + \sum_{n=1}^N x^n K_n(A) + {\rm O}(x^{N+1}) \,,
\end{equation} 
where the first equality serves as a definition of $\mathcal{K}^{(N)}(A;x)$. This proves point 7 of Theorem \ref{thm:summary}, as announced. In other words, the generating function of the precursors $K_n$ appears as the Fourier--Laplace transform of the resolvent. Since we identified the large $N$ limit of these precursors with the free cumulants in section \ref{LargeNlim}, this implies that the function $\dfrac{1}{x}\bigl( \mathcal{K}^{(N)}(A;x) - 1\bigr)$ tends to the Voiculescu $\mathcal{R}$-transform in this limit. However, a direct proof of that limit would be desirable.


\subsection[Integrating the $K_\alpha$ against a Gaussian weight: the Wick property]{Integrating the $\boldsymbol{K_\alpha}$ against a Gaussian weight: the Wick property}
\label{int-wick}

We now come to point 8 of our main Theorem \ref{thm:summary}. Take $A$ to be a random Hermitian Gaussian matrix with variance $\sigma$, \ie in the ensemble GUE$(N,\sigma)$. For any function $f$ on the space of matrices, its expectation value with respect to this distribution is defined as
\begin{equation}
    \mathbb{E}_{A\sim\text{GUE}(N,\sigma)}\bigl(f(A)\bigr) = \frac{1}{Z_{\text{GUE}}} \int DA \, f(A)\,e^{-\frac{1}{2\sigma^2} N \Tr A^2}\,,
\end{equation}
where $Z_{\text{GUE}}:= \int DA\,  e^{-\frac{1}{2\sigma^2} N \Tr A^2}$. Our goal is to compute the average of the polynomials $K_\alpha$. Obviously, they must be even to give a non-vanishing average, hence we take $\alpha \vdash 2n$. Then

\begin{proposition}\label{prop:GUE}
 The generalised precursors enjoy the Wick property, namely for  $\alpha\vdash 2n$,
 \be\label{intG} \mathbb{E}_{A\sim\text{GUE}(N,\sigma)} \bigl( K_\alpha(A) \bigr) = \delta_{\alpha, [2^n]}\,\sigma^{2n}\,.
 \ee
In particular, $\mathbb{E}_{A\sim\text{GUE}(N,\sigma)} \bigl( K_n(A)\bigr) =\delta_{n,2}\,\sigma^2$.
\end{proposition}

In that respect, the $K_\alpha$ qualify again as precursors of the free cumulants which are known to satisfy a similar property when integrated over a semi-circle law, the celebrated limit of the Gaussian distribution. To be explicit, recall that if
$\rho(x)=\inv{2\pi \sigma^2} \sqrt{4 \sigma^2-x^2}$   
is the asymptotic form of the eigenvalue density of the $\text{GUE}(N,\sigma)$ matrix $A$, we have
\bea\label{Catal} \lim_{N\to \infty} \mathbb{E}_{A\sim\text{GUE}(N,\sigma)}\bigl(m_{2n}(A)\bigr)=\int_{-2 \sigma}^{2 \sigma} x^{2n} \rho(x) dx&=& \mathrm{Catalan}(n)\,  \sigma^{2n}\,,\\
\label{kappanG}
\lim_{N\to \infty} \mathbb{E}_{A\sim\text{GUE}(N,\sigma)}\bigl(\kappa_n(A)\bigr)&=& \delta_{n,2}\,\sigma^2\,,\\
\label{kappalphaG} 
\lim_{N\to \infty} \mathbb{E}\bigl(\kappa_\alpha(A)\bigr)=  \lim_{N\to \infty}  \mathbb{E}\left(\prod_k \kappa_{k}^{\ah_k}(A) \right)=\prod_k  \Bigl(\lim_{N\to \infty}  \mathbb{E}(\kappa_{k}(A))\Bigr)^{\ah_k} 
 &=& \delta_{\alpha, [2^n]}\,\sigma^{2n}\,.\eea
  The term ``Wick property''  originates from Quantum Field Theory and refers to a Feynman diagram interpretation of such results. For a Gaussian free field $\phi$, 
the ``$n$-point function", \ie  the expectation value of $\phi^n$,  ($n$ even), is interpreted  as a sum of pairings between $n$ points. 
In the computation of  $\lim_{N\to \infty}  \mathbb{E}_{A\sim\text{GUE}(N,\sigma)}(m_{2n}(A))$, only non-crossing pairings are allowed,
whose number is a Catalan number, see \eqref{Catal}. Moreover, the free cumulants vanish in mean, except for the one of degree 2, as stated in equation \eqref{kappanG}.  Finally,
\eqref{kappalphaG} follows from the 
factorisation of means of invariant quantities in the large $N$ limit. In the current context, \eqref{intG} says that this property already
applies at finite $N$ for the expectation value of $K_\alpha$.  For finite $N$, $K_\alpha$ is not a product of $K_n$'s, so that the ``Wick nature'' of our result \eqref{intG} should be understood as a formal generalisation referring to the structure of partitions labeling the generalised precursors.

\begin{proof} Taking the GUE average of the expansion \eqref{expand2}, we have
\begin{equation}\nonumber  \mathbb{E}_{A\sim\text{GUE}(N,\sigma)} \bigl( Z(A,B\,; z) \bigr)
 = \sum_{n=0}^{N/2} z^{2n} \sum_{\alpha\vdash 2n} N^{\ell(\alpha)} \mathbb{E}_{A\sim\text{GUE}(N,\sigma)} \bigl( K_\alpha(A) \bigr) \,\frac{|\Cl_\alpha|}{(2n)!}\, p_\alpha(B) + {\rm O}(z^{N+1})\,.
 \end{equation}
On the other hand, a direct computation starting from the HCIZ integral representation \eqref{hciz} yields
\bea \mathbb{E}_{A\sim\text{GUE}(N,\sigma)} \bigl( Z(A,B\,; z) \bigr) &=& Z_{\text{GUE}}^{-1} \int DU \int DA\, e^{-\frac{1}{2\sigma^2} N \Tr A^2}  e^{N z \Tr (U^\dagger A U B)} \nonumber  \\
&=& Z_{\text{GUE}}^{-1} \int DA\, e^{-\frac{1}{2\sigma^2} N \Tr A^2}  e^{N z \Tr (AB)} \nonumber  \\
&=& e^{\oh N \sigma^2 z^2 \Tr B^2} =\sum_n \frac{N^n \sigma^{2n}}{2^n n!}\,  z^{2n}\,\Tr^n B^2 \,,\nonumber  
\eea
on which it is manifest that only $\alpha=[2^n]$ contributes (the second equality made use of the $\U(N)$-invariance of the GUE while the third one is the matrix version of the Gaussian Fourier transform).  Comparing the two expressions of $\mathbb{E}_{A\sim\text{GUE}(N,\sigma)} \bigl( Z(A,B\,; z) \bigr)$ and using $|\Cl_{[2^n]}|/(2n)! = 1/(2^n n!)$, $p_{[2^n]}(B)=\Tr^n B^2$ and $\ell(\alpha)=n$, we obtain the desired result \eqref{intG}.
\end{proof}


\subsection{Relation between cyclic products and generalised precursors}
\label{cyclicpr}
There is a well-known relation in the large $N$ limit between the cyclic product of $n$ matrix 
elements of a matrix $A$ averaged over a $\U(N)$ invariant distribution and the $n$-th free cumulant $\kappa_n(A)$. To establish it, it suffices to 
prove for $A$ deterministic and $A^U:= U A \Ud$ that
\be \label{CMSSrel}
\kappa_n(A) =\lim_{N\to \infty}  N^{n-1}\,\int_{\U(N)} DU\, A^U_{i_1i_2}A^U_{i_2i_3}\cdots A^U_{i_ni_1}  
\ee
for any $n$-plet of pairwise distinct indices $i_1,i_2,\cdots i_n$ and no summation over repeated indices.
This relation seems to have been pointed out first in \cite[Theorem 2.6]{CMSS06}, see also \cite[Theorems 1 and  2]{Maillard-etal}. 
An elegant proof has been given in a recent paper of
Bernard and Hruza \cite{BH24}. In fact their proof extends word for word to finite $N$, substituting the precursor $K_n$ for the free cumulant $\kappa_n$: 
\be \label{CMSSrel2}
 K_n(A) = N^{n-1} \int_{\U(N)} DU\, A^U_{i_1i_2}A^U_{i_2i_3}\cdots A^U_{i_ni_1}  \,.\ee
\begin{proof}\!\!\cite{BH24} Consider the $N\times N$ matrix $\Sigma$ made of the $n\times n$ matrix
 of the cyclic permutation $\Sigma_{i j}=\delta_{i,j+1\!\mod n}$, ($i,j=1,\cdots n$), bordered by zeros. One easily checks that $\Sigma$ is of rank $n$ and that the first non-vanishing trace of $\Sigma^k$ is for $k=n$. It is then clear that the first non-constant coefficient in the expansion \eqref{expand2} of the HCIZ integral $Z(A, \Sigma\,; z)$ is of order $z^n$ and arises from the single trace $p_{[n]}(\Sigma)=\Tr(\Sigma^n)=n$. Therefore,
 \begin{equation*}
  Z(A, \Sigma\,; z) = 1 + z^n N\, K_n(A) + {\rm O}(z^{2n})\,.
 \end{equation*}
 On the other hand, a series expansion of the exponential gives
 \begin{equation*} Z(A, \Sigma\,; z) = \sum_{k=0}^{+\infty} \frac{z^k N^k}{k!} \int DU \, \Tr^k (A^U \Sigma) = 1 + z^n N^n \int_{\U(N)} DU \, A^U_{i_1i_2}A^U_{i_2i_3}\cdots A^U_{i_ni_1} + {\rm O}(z^{2n})
 \end{equation*}
 for any choice of a set of pair-wise distinct indices $i_j$.  Identifying the $z^n$ term yields the result. \end{proof} 
 
This result extends to the case of a permutation $\sigma$ of arbitrary cycle type, establishing the point 9 of our main Theorem \ref{thm:summary}:
\begin{proposition}\label{prop4}
    Let $\sigma\in S_n$. Then, we have
    \be \label{EspKalpha} K_{[\sigma]}(A) = N^{n-\ell([\sigma])}\,  \int_{\U(N)}DU\, A^U_{1\,\sigma(1)}\,A^U_{2\,\sigma(2)}\cdots A^U_{n\,\sigma(n)} \,, \ee
\end{proposition}

\begin{proof}
The integral can be computed from the rule \eqref{Wg} of the Weingarten--Samuel--Collins calculus. After a few manipulations and the use of the identity
\begin{equation*}
    \sum_{j_1,\dots,j_n=1}^N \prod_{b=1}^n A_{j_b\,j_{\sigma(b)}} = p_{[\sigma]}(A)\,,
\end{equation*}
we get
\begin{equation*}
          \int_{\U(N)}DU\, \prod_{i=1}^n A^U_{i\,\sigma(i)}
 =\sum_{\tau\in S_n}    \Wg\bigl([\sigma \tau^{-1}]\bigr)\,  p_{[\tau]} (A)= N^{\ell([\sigma])-n} K_{[\sigma]}(A)\,,
\end{equation*}
where the last equality follows from \eqref{eq:Kp2}.
\end{proof}


\subsection{Integrals of powers of a rotated matrix element}
\label{sec:A11n}

In the previous subsection, we have studied specific products of matrix components $A_{ij}^U$ integrated on conjugacy orbits. We now consider another type of products, namely powers of a single diagonal component $A_{11}^U$. Let $B$ be the rank 1 matrix defined by $B_{ij}=\delta_{i,1}\delta_{j,1}$, such that $A_{11}^U = \Tr(UA\Ud B)$. With this choice of $B$, the HCIZ integral serves as a generating function for the integrals of $(A_{11}^U)^n$:
\begin{equation}\label{eq:ExpAU11}
    Z(A,B\,; z) = \int_{U(N)} DU\, e^{z N A_{11}^U} = \sum_{n = 0}^{+\infty} z^n\frac{N^n}{n!} \int_{\U(N)} DU\, (A_{11}^U)^n\,.
\end{equation}
This can alternatively be computed from the three expansions \eqref{Z2sp}, \eqref{expand} and \eqref{expand2}, noting that our specific choice of $B$ implies $p_\alpha(B)=1$ for all partitions $\alpha$. Using the identity 
 $ \sum_{\alpha\vdash n}|\Cl_\alpha| \hat\chi_\lambda(\alpha)=\delta_{\lambda [n]} $, which follows from the orthogonality \eqref{eq:OrthoChi} of $S_n$ characters, we may then express the integral of $(A_{11}^U)^n$ in terms of Schur functions, Newton polynomials or generalised precursors:
 \bea\label{A11n} \int_{\U(N)} DU\, (A_{11}^U)^n&=& 
 \frac{n!}{\prod_{j=0}^{n-1} (N+j)} s_{[n]}(A) \\
 \nonumber &=&
\prod_{j=0}^{n-1} \inv{(N+j)}  \sum_{\alpha\vdash n} |\Cl_\alpha|\,  p_\alpha(A) \\
\nonumber &=&  \sum_{\alpha\vdash n} |\Cl_\alpha|\, N^{\ell(\alpha)-n}  K_\alpha(A)\,.\eea
We will come back on this result from a slightly different point of view in Section \ref{sec:CGL}.


\section{Sums of unitary conjugacy orbits}
\label{Additivity}

\subsection[Multiplicativity of the HCIZ integral and additivity of $K_n$]{Multiplicativity of the HCIZ integral and additivity of $\boldsymbol{K_n}$}

We now come to a central subject of our work. Let us fix two matrices $A$ and $B$ of size $N\times N$. In this section, we will consider the sum of the (unitary) conjugacy orbits of $A$ and $B$, formed by matrices $UA\Ud+VB\Vd$ with $U,V\in\U(N)$. Our goal will be to characterise the values of various conjugacy-invariant functions on this space averaged over the Haar measure. If $f$ is such a function, we note that $f(UA\Ud+VB\Vd)=f(A+\Ud VB\Vd U)$ so that it is enough to keep $A$ fixed and average over elements of the conjugacy orbit of $B$. Therefore, we are interested in the quantity
\begin{equation}
    \int_{\U(N)} DU \, f(A+UB\Ud)\,.
\end{equation}
The starting point of our discussion is the following result, which shows that the HCIZ integral  is ``multiplicative'' with respect to this operation, while the free cumulant precursors are ``additive'' (proving the point 3 in the summary Theorem \ref{thm:summary}).

\begin{theorem}\label{Thm:ZKSum}
    Let $A,B,C$ be fixed $N\times N$ matrices. Then
    \begin{equation}\label{eq:MultZ}
        \int_{\U(N)} DU \, Z(A+UB\Ud,C\,; z) = Z(A,C\,; z)\,Z(B,C\,; z)\,.
    \end{equation}
    Moreover, for $n\in\lbrace 1,\dots,N\rbrace$, we have
    \begin{equation}\label{eq:AddK}
        \int_{\U(N)} DU \, K_n(A+UB\Ud,C) = K_n(A) + K_n(B)\,.
    \end{equation}
\end{theorem}

\begin{proof}
    The identity \eqref{eq:MultZ} is well-known and is easily proven as follows. By the definition \eqref{hciz} of the HCIZ integral, we have
    \begin{equation}
        \int_{\U(N)} DU \, Z(A+UB\Ud,C\,; z) = \iint\biggl._{\!\U(N)} DU \,DV\, e^{zN\,\Tr(VA\Vd C) + zN\,\Tr(VU B \Ud \Vd C)}\,.
    \end{equation}
    Factorising the exponential and performing the change of variable $U \mapsto V^\dagger U$ (which leaves the Haar measure $DU$ invariant), we get
    \begin{equation}
        \int_{\U(N)} DU \, Z(A+UB\Ud,C\,; z) = \iint\biggl._{\!\U(N)} DU  \,DV\, e^{zN\,\Tr(VA\Vd C)} e^{zN\,\Tr(U B \Ud C)}\,.
    \end{equation}
    The right-hand side factorises into $Z(A,C\,; z)\,Z(B,C\,; z)$, proving the desired result \eqref{eq:MultZ}.

    Recall from equation \eqref{expand2} that the HCIZ integral $Z(X,C\,; z)$ admits a series expansion in powers of $z$ and traces of powers of $C$, starting with 1 and whose coefficient in $z^n\, \Tr(C^{\alpha_1}) \cdots \Tr(C^{\alpha_\ell})$ is proportional to the generalised precursor $K_{(\alpha_1,\dots,\alpha_\ell)}(X)$. To obtain the second result \eqref{eq:AddK} of the theorem, we then project the identity \eqref{eq:MultZ} along $z^n\,\Tr(C^n)$. The key observation is that a single trace $z^n \,\Tr(C^n)$ in the product $Z(A,C\,; z)\,Z(B,C\,; z)$ can only be produced by picking the term proportional to $z^n\,\Tr(C^n)$ in the first factor with the constant term 1 in the second factor, or the opposite. We then get
    \begin{equation*}
        \int_{\U(N)} DU \, [z^n\, \Tr(C^n)] Z(A+UB\Ud,C\,; z) = [z^n\, \Tr(C^n)] Z(A,C\,; z) + [z^n\, \Tr(C^n)] Z(B,C\,; z)\,,
    \end{equation*}
    proving the identity \eqref{eq:AddK}.
\end{proof}

\begin{remark} The additivity property \eqref{eq:AddK} is the main feature characterising $K_n$ as a finite $N$ precursor of the free cumulant $\kappa_n$. Indeed, the latter is well-known to satisfy this property in the large $N$ limit, see equation \eqref{eq:AddFree}. That the additivity of $\kappa_n$ can be alternatively proven from the multiplicative identity \eqref{eq:MultZ} (together with the appearance of $\kappa_n$ in $\log Z$ in the finite rank,
large $N$ limit) was first observed by Guionnet and Ma\"\i da in~\cite{GM}.\end{remark}

Theorem \ref{Thm:ZKSum} shows that the precursors $K_n$ have a particularly simple behaviour with respect to the operation of averaging on sums of conjugacy orbits, namely they are additive. We will study the behaviour of more general polynomials later in Section \ref{Sec:Coproduct}. Before that, we will expand more on this additivity property of $K_n$ and its applications.


\subsection[Higher cumulants of the random variable $\delta K_n$]{Higher cumulants of the random variable $\boldsymbol{\delta K_n}$}
\label{Cumul-deltaK}
Consider the random variable
\be \delta K_n(A,B;U):= K_n(A+UB\Ud)-K_n(A)-K_n(B)\,,\ee
where $U$ is taken randomly and uniformly following the Haar measure of $\U(N)$.
By Theorem \ref{Thm:ZKSum}, 
\be\label{average} \mathbb{E} ( \delta K_n(A,B;U))=0\ee
and it then is natural to examine the higher moments and  (classical)  cumulants of $\delta K_n$, to estimate its fluctuations around zero.
One computes (assuming $A$ and $B$ to be traceless for simplicity)
\bea \mathrm{var}(\delta K_2)&=&  \mathbb{E} \bigl( \delta K_2^2 \bigr)= \frac{4 N^4}{(N^2-1)^3} \kappa_2(A) \kappa_2(B) \,, \\ 
\nonumber  \mathbb{E} ( \delta K_2^3)&=& \frac{16 N^6}{(N^2-1)^4(N^2-4)}{\kappa_3(A) \kappa_3(B)} \approx \frac{16}{N^4}\kappa_3(A) \kappa_3(B) \,,\\[3pt]
\nonumber   \mathbb{E} ( \delta K_2^4)&\approx &\frac{48 \kappa_2^2(A) \kappa_2^2(B)^2}{N^4} \approx 3 \mathrm{var}(\delta K_2)^2\,,\eea
where $\approx$ denotes asymptotics in the $N\to\infty$ limit. This implies that the ``dimensional skewness'' and the ``normalized kurtosis'' of $\delta K_2$ vanish respectively as $1/N$ and $1/N^2$ for large $N$:
\bea  \nonumber \frac{\mathbb{E} ( \delta K_2^3)}{(\mathrm{var}(\delta K_2))^{3/2}} & \approx &\frac{2}{N} \frac{\kappa_3(A) \kappa_3(B)}{(\kappa_2(A) \kappa_2(B))^{3/2}} \,\\
\nonumber \mathrm{kur}(\delta K_2) &:=& \frac{\mathbb{E} ( \delta 
K_2^4)}{(\mathrm{var}(\delta K_2))^2}-3 \\ 
\nonumber & \approx & 
\frac{6}{N^2}\left(\frac{\kappa_4(A) \kappa_4(B)}{\kappa_2(A)^2 \kappa_2(B)^2}  -1 \right) \,\eea
showing that the law of $\delta K_2$
converges at large $N$ to a normal centered distribution of variance $\mathrm{var}(\delta K_2) \approx 4N^{-2}\kappa_2(A) \kappa_2(B)$.\\

We reach a similar conclusion for $\delta K_n$ at higher $n$.
For example
\bea
\mathrm{var} (\delta K_3)&=&{\frac{9 N^8}{(N^2-1)^3(N^2-9)^2}} \\
&& \hspace{10pt}\big(\kappa_4(A)  \kappa_2(B) + 2 \kappa_3(A) \kappa_3(B) + \kappa_2(A)
\kappa_4(B) +\kappa_2(A) \kappa_2(B)(\kappa_2(A)+\kappa_2(B))\big) \,,\nonumber\\
 \nonumber \mathrm{var} (\delta K_4)&\approx&   
 \frac{8}{N^2} \Big( 2( \kappa_6(A) \kappa_2(B)+\kappa_2(A) \kappa_6(B))+4(\kappa_5(A)\kappa_3(B)+\kappa_3(A)\kappa_5(B))+
 6\kappa_4(A)\kappa_4(B)\\ \nonumber
 && \hspace{25pt} +4 \kappa_2(A) \kappa_2(B)(\kappa_4(A)+\kappa_4(B))+4\kappa_3(A)^2 \kappa_2(B)+4 \kappa_2(A) \kappa_3(B)^2 \\[3pt]  \nonumber
 && \hspace{25pt} +2\kappa_4(A) \kappa_2(B)^2 +2 \kappa_2(A)^2  \kappa_4(B)
 +8 \kappa_3(A) \kappa_3(B) (\kappa_2(A)+\kappa_2(B))\\[3pt] \nonumber
 && \hspace{25pt} +\kappa_2(A) \kappa_2(B)(2\kappa_2(A)^2+3 \kappa_2(A) \kappa_2(B) +2 \kappa_2(B)^2)\Big)\,.
\eea
More generaly, for large $N$, $ \mathrm{var}(\delta K_n) ={\rm O}(N^{-2})$ and the higher  (classical) cumulants of $\delta K_n$ are more and more suppressed: $c^{(2k)}(\delta K_n)={\rm O}(N^{-2k})$, $c^{(2k+1)}(\delta K_n)={\rm O}(N^{-2k-2})$, see \cite{CMSS06}.


\subsection{An illustration: Horn's problem}
\label{Horn}
Horn's problem deals with the spectral properties of the sum of two (say Hermitian) $N\times N$ matrices of given spectrum. 
In particular, if one takes $A$ and $B$ randomly and independently chosen on their U$(N)$ orbits, what can be said about the spectrum of the random 
variable $A+B$ ? It suffices to consider, for $A$ and $B$ fixed, the random variable $A+U B U^\dagger$, with $U$ randomly taken in $\U(N)$ with the Haar measure.

 While at finite $N$ there is a  non trivial distribution of the eigenvalues of the sum,  the problem simplifies drastically in the large $N$ limit:
as proved by Voiculescu  \cite{Voi91}, $A$ and $UB\Ud$ may be regarded as free variables  and  the distribution of eigenvalues of $A+U B U^\dagger$ concentrates on a limit described by free convolution, with respect to which the free cumulants $\kappa_n$ are additive. In a notation parallel to Section \ref{Cumul-deltaK}, the random variable
 \bea  \label{DeltaKappa}  \delta\kappa_n(A,B;U)&:=& \kappa_n(A+U B U^\dagger)-\kappa_n(A)-\kappa_n(B)
\eea 
satisfies
\be  \label{limit} \lim_{N\to\infty}  \delta\kappa_n = 0 
\ee 
in probability, since $\delta \kappa_n$ asymptotically has a centered normal law of variance $\mathrm{O}(N^{-2})$
\cite{CMSS06}. In particular, in the $N\to\infty$ limit, the additivity $\kappa_n(A+U B U^\dagger)\approx \kappa_n(A)+\kappa_n(B)$ holds at the level of random variables, not only on average. Since the precursor $K_n$ converges to the free cumulant $\kappa_n$ in this limit, we note that we also have $K_n(A+U B U^\dagger)\approx K_n(A)+K_n(B)$ as $N\to\infty$.

 \begin{figure}[h] \begin{center}   
\includegraphics[width=.3\textwidth]
{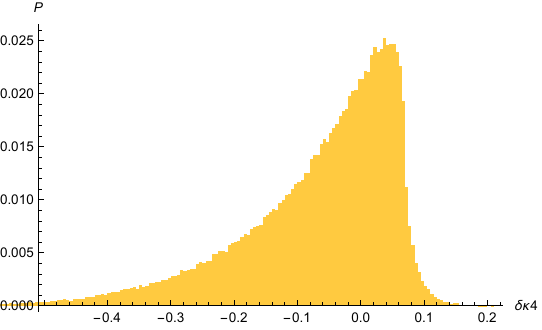}
\quad \includegraphics[width=.3\textwidth]
{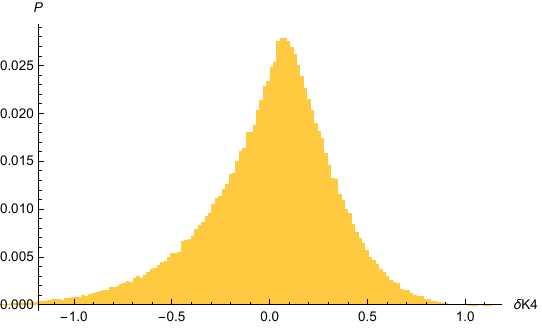}\quad \includegraphics[width=.3\textwidth]
{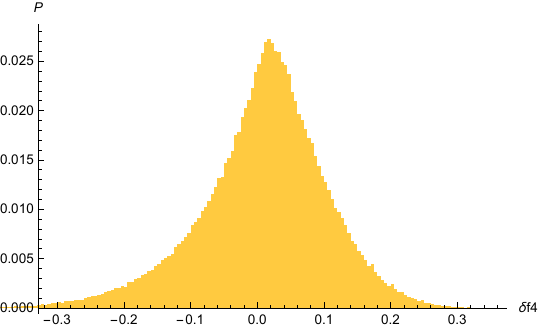}
\\ [15pt] \caption{\label{hornfig} 
 A numerical experiment with  $4\times 4$  matrices  $A$ and $B$ of regularly spaced spectrum on $[-1,1]$ and a sample of 200,000 Haar-distributed matrices $U$, showing the
histograms of $\delta \kappa_4(U)$ (left),
$\delta K_4(U)$ (middle) and $\delta f_4(U)$ (right). 
The improvement of the latter two with respect to the former is manifest: smaller
skewness, distribution closer to Gaussian, etc. The broader range of values of   $\delta K_4$ with respect to  $\delta \kappa_4$ is due to the prefactor $N^4(N^2 + 1) /((N^2 - 1) (N^2 - 4) (N^2 - 9)$
in \eqref{Kmrel}. } \end{center}\end{figure}

At finite $N$, the eigenvalues of $A+U B U^\dagger$ follow a complicated distribution in a certain convex polytope. It is then natural to search for simpler ways to approximate and visualise this distribution. The above results suggest to look at $\kappa_n(A+UB\Ud)$ and $K_n(A+UB\Ud)$, which are symmetric polynomials of these eigenvalues and which we expect to spread around $\kappa_n(A)+\kappa_n(B)$ and $K_n(A)+K_n(B)$ respectively, with more and more precision as $N$ grows. We note that, at finite $N$, the additivity $\E(K_n(A+UB\Ud))=K_n(A)+K_n(B)$ (on average) only holds for $K_n$, and not for $\kappa_n$. We thus expect the former to provide a better approximation of the Horn distribution compared to the latter. This is precisely quantified by the distributions of the random variables $\delta\kappa_n(U)$ and $\delta K_n(U)$ around zero. We performed some numerical experiments for various choices of $A$ and $B$ (picked from different ensembles: uniform spectrum, GUE, etc), comparing the distribution of these two variables.  Already for a small size $N=4$, the improvement of $K_n$ compared to $\kappa_n$ is manifest; see Figure \ref{hornfig}, where we also show data for the 4-th order finite free cumulant of \cite{AP}, which for traceless matrices reads $f_4={m}_4 -(2N - 3)/(N - 1) {m}_2^2$.

\subsection{Behaviour of more general polynomials and the orbit coproduct}
\label{Sec:Coproduct}

\paragraph{The orbit coproduct.} The previous subsections were mainly devoted to the additivity property \eqref{eq:AddK} of the precursors $K_n$ with respect to averaging over sums of conjugacy orbits. It is natural to wonder what is the behaviour of more general invariant polynomials under this operation. To formalise this question, let us introduce some notations and definitions. We will denote by $\AcN$ the algebra of polynomials on $N\times N$ matrices invariant under the action of $\U(N)$ by conjugacy. This is an associative, commutative and unital algebra. It is naturally $\mathbb{Z}_{\geq 0}$-graded with respect to the degree, in the sense that the decomposition
\begin{eqnarray}
    \AcN = \bigoplus_{n \geq 0} \AcN_n\,, \qquad \AcN_n := \bigl\lbrace f\in\AcN\,|\, \text{deg}(f)=n \bigr\rbrace\,,
\end{eqnarray}
satisfies $\AcN_n \cdot \AcN_{n'} \subseteq \AcN_{n+n'}$.

For $f\in\AcN$, recall that our object of interest is the quantity
\begin{equation}\label{eq:Coproduct}
    \Delta f(A,B) = \int_{\U(N)} DU \, f(A+UB\Ud)\,,
\end{equation}
seen as a function of the two matrices $A$ and $B$. It is clear that $\Delta f(A,B)$ is a polynomial in the entries of $A$ and $B$. Moreover, it is invariant under independent conjugations of both $A$ and $B$ by unitary matrices. It then takes the form of a finite sum of products of an invariant polynomial of $A$ with an invariant polynomial of $B$. Abstractly, $\Delta f$ can thus be seen as an element of the tensor product $\AcN \otimes \AcN$.\footnote{An element $h = \sum_i f_i \otimes g_i$ in the tensor factor $\AcN \otimes \AcN$ is identified with an invariant polynomial in two matrix variables by taking $h(A,B) = \sum_i f_i(A)g_i(B)$.} The operation of averaging over the sum of two conjugacy orbits then defines a coproduct $\Delta : \AcN \to \AcN \otimes \AcN$, which we will call the \textit{orbit coproduct} and which will be the main protagonist of this subsection.\\

We note that $\Delta$ is not a morphism with respect to the product of invariant polynomials and thus does not define a Hopf algebra structure on $\AcN$. We gather some of the basic properties of this coproduct in the following proposition.

\begin{proposition}\label{prop:Coproduct}
    The orbit coproduct $\Delta$
    \begin{enumerate}
        \item is cocommutative,
        \item is coassociative,
        \item is graded, \ie $\Delta\bigl(\AcN_n\bigr) \subseteq \bigoplus_{k=0}^n \AcN_k \otimes \AcN_{n-k}$,
        \item admits the precursors $K_n$ as primitive elements,  \ie
\begin{equation}\label{eq:KPrim}
    \Delta K_n = K_n \otimes 1 + 1 \otimes K_n\,.
\end{equation}
    \end{enumerate}
\end{proposition}

\begin{proof}
    The cocommutativity means that for any $f\in\AcN$, $\Delta f$ is symmetric with respect to the exchange of the two tensor factors. Coming back to the functional formulation, this is equivalent to $\Delta f(A,B)=\Delta f(B,A)$, which is obvious from the definition \eqref{eq:Coproduct} of $\Delta$ and the invariance of $f$ under conjugation. The coassociativity $(\Delta\otimes\text{Id})\circ\Delta=(\text{Id}\otimes\Delta)\circ\Delta$ translates to the identity
    \begin{equation}
        \int_{\U(N)} DU\, \Delta f(A+UBU^\dagger,C) = \int_{\U(N)} DU\,\Delta f(A,B+UCU^\dagger)
    \end{equation}
    for all matrices $A,B,C$. This is easily verified, with these integrals taking the common value
    \begin{eqnarray}
        \iint\biggl._{\!\U(N)} DU\,DV\, f(A+UBU^\dagger+VCV^\dagger)\,,
    \end{eqnarray}
    which measures the average of $f$ on the sum of the three conjugacy orbits of $A$, $B$ and $C$. The graded property of $\Delta$ follows directly by going back to its functional definition. Finally, the primitivity of the precursors $K_n$ is an equivalent reformulation of their additivity property \eqref{eq:AddK}. 
\end{proof}

\begin{remark}\label{rem:StandardCoproduct}
    The algebra $\AcN$ can also be seen as the algebra of symmetric polynomials in $N$ variables, by passing to eigenvalues of matrices. We note that the algebra of symmetric polynomials in an arbitrary number of variables is equipped with a well-known cocommutative coassociative graded coproduct, defined by splitting any such polynomial into a sum of products of polynomials in two subsets of the variables (see for instance~\cite[Example I.5.25]{Mac}). Although this coproduct shares some common properties with the one $\Delta$ studied here, it also differs from it on various aspects. To start with, $\Delta$ is defined on symmetric functions of a fixed number $N$ of variables and will depend explicitely on $N$, in contrast with this other coproduct. Moreover, the latter defines a Hopf algebra structure whereas we observed above that $\Delta$ does not. We will comment on other similarities and differences of these two coproducts throughout this subsection.
\end{remark}

\paragraph{The Newton basis of $\boldsymbol{\AcN}$.} To understand the coproduct of more general polynomials than the precursors, it will be useful to discuss some natural bases of $\AcN$. For $\alpha$ a partition, recall that $p_\alpha(A)=\prod_{i=1}^{\ell(\alpha)} \Tr(A^{\alpha_i})$. If $n\leq N$, it is clear that $\lbrace p_\alpha \rbrace_{\alpha\vdash n}$ forms a basis of $\AcN_n$. Such a family is not independent anymore for $n> N$, due to the Newton identities relating traces of powers of $A$. However, it is easy to extract a basis from this family by removing the traces $\Tr(A^k)$ with $k>N$. Define $\Pc_n$ as the set of partitions $\alpha\vdash n$ such that $\alpha_i\leq N$ for all $i\in\lbrace 1,\dots,\ell(\alpha)\rbrace$, \ie partitions formed only by integers smaller or equal to $N$. We also let $\Pc=\bigsqcup_{n\geq 0} \Pc_n$.\footnote{\label{foot:Empy}Note that we include $n=0$ in this union, with the convention that $\Pc_0=\lbrace \emptyset \rbrace$ is the set formed by the empty partition $\emptyset$. The corresponding Newton polynomial is simply the constant function $p_{\emptyset}(A)=1$.} It is then clear that $\lbrace p_\alpha \rbrace_{\alpha\in\Pc_n}$ is a basis of $\AcN_n$ and $\lbrace p_\alpha \rbrace_{\alpha\in\Pc}$ is a basis of $\AcN$. We will refer to it as the Newton basis. We note that it is homogeneous, in the sense that each element $p_\alpha$ is of a well-defined degree $d(\alpha)=\alpha_1+\dots+\alpha_{\ell(\alpha)}$ (which coincides with the standard notion of degree of a partition). We will denote by $\Pi_{\alpha\beta}\null^\gamma$ the structure constants of the algebra $\AcN$ in this basis, defined by
\begin{equation}\label{eq:ProdNewton}
    p_\alpha\,p_\beta = \sum_{\gamma\in\Pc} \Pi_{\alpha\beta}\null^\gamma\,p_\gamma\,.
\end{equation}
For degrees $d(\gamma)$ smaller than $N$, these structure constants are quite simple:
\begin{equation}\label{eq:PiConc}
    \Pi_{\alpha\beta}\null^\gamma = \left\lbrace \begin{array}{ll}
        1 & \text{ if } \gamma \text{ is the concatenation of } \alpha \text{ and } \beta\,,  \\
        0 & \text{ otherwise}\,. 
    \end{array} \right.
\end{equation}

\paragraph{The HCIZ integral and the dual basis.} We define the \textit{dual} $\lbrace p^\alpha_\ast \rbrace_{\alpha\in\Pc_n}$ of the Newton basis by expanding the HCIZ integral as
\begin{equation}\label{eq:Dual}
    Z(A,B\,; z) = \sum_{\alpha\in\Pc} z^{d(\alpha)}\,p^\alpha_\ast(A)\,p_\alpha(B)\,.
\end{equation}
This produces another homogeneous basis of $\AcN$. By equation \eqref{expand2}, the dual polynomials $p^\alpha_\ast$ essentially coincide with the generalised precursors for degrees smaller or equal to $N$, up to a normalisation constant:
\begin{equation}\label{eq:pstarK}
    p^\alpha_\ast(A) = \frac{N^{\ell(\alpha)}|\Cl_\alpha|}{n!} K_\alpha(A)\,, \qquad  \forall \, \alpha\in\Pc_n\text{ with } n\leq N .
\end{equation}
For degrees bigger than $N$, the generalised precursors were not defined and the dual polynomials $p^\alpha_\ast(A)$ are thus new objects.\\

Let us also define numbers $(\eta^{\alpha\beta})_{\alpha,\beta\in\Pc}$ and $(\eta_{\alpha\beta})_{\alpha,\beta\in\Pc}$ by the following decompositions of the HCIZ integral:
\begin{equation}\label{eq:eta}
    Z(A,B\,; z) = \sum_{\alpha,\beta\in\Pc} z^{d(\alpha)}\,\eta^{\alpha\beta}\, p_\alpha(A)\,p_\beta(B) = \sum_{\alpha,\beta\in\Pc} z^{d(\alpha)}\,\eta_{\alpha\beta}\, p^\alpha_\ast(A)\,p^\beta_\ast(B)\,.
\end{equation}
The following properties are essentially straightforward.

\begin{lemma} ~
    \begin{enumerate}
        \item $(\eta^{\alpha\beta})_{\alpha,\beta\in\Pc}$ and $(\eta_{\alpha\beta})_{\alpha,\beta\in\Pc}$ are symmetric and block-diagonal by degrees, \ie $\eta^{\alpha\beta}=\eta_{\alpha\beta}=0$ if $d(\alpha) \neq d(\beta)$.
        \item $(\eta^{\alpha\beta})_{\alpha,\beta\in\Pc}$ and $(\eta_{\alpha\beta})_{\alpha,\beta\in\Pc}$ are inverse of one another, in the sense that
        \begin{equation}
            \sum_{\gamma\in\Pc} \eta^{\alpha\gamma}\,\eta_{\gamma\beta} = \delta^{\alpha}_{\;\,\beta}\,.
        \end{equation}
        \item For $\alpha\in\Pc$, we have
        \begin{equation}\label{eq:ppstar}
            p^\alpha_\ast = \sum_{\beta\in\Pc} \eta^{\alpha\beta}\, p_\beta \qquad \text{ and } \qquad p_\alpha = \sum_{\beta\in\Pc} \eta_{\alpha\beta}\, p^\beta_\ast\,.
        \end{equation}
        \item The map $\eta : \AcN \times \AcN \to \C$, defined on the Newton basis by $\eta(p_\alpha,p_\beta)=\eta_{\alpha\beta}$ and extended on $\AcN \times \AcN$ by bilinearity, is a non-degenerate symmetric bilinear form on $\AcN$.
        \item The bases $\lbrace p_\alpha \rbrace_{\alpha\in\Pc}$ and $\lbrace p^\alpha_\ast \rbrace_{\alpha\in\Pc}$ are dual with respect to the pairing $\eta(\cdot,\cdot)$.
    \end{enumerate}
\end{lemma}

\begin{remark}
    The property 1 means that, for a given $\alpha\in\Pc$, only a finite number of $\eta^{\alpha\beta}$ and $\eta_{\alpha\beta}$ are non-zero, when $d(\beta)=d(\alpha)$. In particular the sums appearing in the points 2 and 3 are in fact finite and can be truncated to sums over $\Pc_{d(\alpha)}$.

    For degrees $d(\alpha),d(\beta)$ smaller than $N$, an explicit expression of $\eta^{\alpha\beta}$ can be extracted from the expansion \eqref{expand} of the HCIZ integral:
    \begin{equation}
        \eta^{\alpha\beta} = \sum_{\lambda\vdash n\atop \ell(\lambda) \leq N } \frac{1}{C_\lambda} 
\sum_{\alpha,\beta\vdash n}\frac{|\Cl_\alpha |\cdot|\Cl_\beta |}{n!^2} \hat\chi_\lambda(\alpha)\hat\chi_\lambda(\beta)\,, \qquad \forall \, \alpha,\beta\in\Pc_n \text{ with }  1 \leq n\leq N .
    \end{equation}
    A similar expression for the inverse $\eta_{\alpha\beta}$ can be derived using orthogonality relations \eqref{eq:OrthoChi} of symmetric characters. Reinserting this formula into equation \eqref{eq:ppstar} then essentially gives back the relation \eqref{eq:Kp1} between the Newton polynomials $p_\alpha$ and the generalised precursors $K_\alpha$, up to the normalisation constants in \eqref{eq:pstarK}. In contrast, when the degree $d(\alpha)=d(\beta)=n$ is strictly greater than $N$, one cannot extract such a direct expression of $\eta^{\alpha\beta}$ from the expansion \eqref{expand} of the HCIZ integral, since the latter involves all Newton polynomials $\lbrace p_\alpha \rbrace_{\alpha\vdash n}$ and not only the linearly independent ones $\lbrace p_\alpha \rbrace_{\alpha\in\Pc_n}$. To compute $\eta^{\alpha\beta}$, one would in principle need to replace the redundant $p_\alpha$'s using Newton identities. We will not attempt this here and will simply manipulate $\eta^{\alpha\beta}$ and  $\eta_{\alpha\beta}$ through their formal characterisation \eqref{eq:eta}.

    The numbers $\eta^{\alpha\beta}$ and $\eta_{\alpha\beta}$ depend in a complicated way on $N$. The large $N$ expansion of $\eta^{\alpha\beta}$ is essentially a generating function for monotone Hurwitz numbers $H^{\bullet\hspace{1pt},\hspace{1pt}\leq}_g(\alpha,\beta)$, while the expansion of $\eta_{\alpha\beta}$ encodes their strictly monotone variants $H^{\bullet\hspace{1pt},\hspace{1pt}<}_g(\alpha,\beta)$. This shows that the bilinear form $\eta(\cdot,\cdot)$ is also natural from the point of view of combinatorics and enumerative geometry.
\end{remark}

\paragraph{Explicit expression of the orbit coproduct.}  Recall that our goal is to describe the orbit coproduct \eqref{eq:Coproduct}. We are now in position to state the main result of this subsection, which gives an explicit and simple expression of this coproduct in the dual Newton basis.

\begin{theorem}\label{Thm:DeltaDual}
    Let $\gamma\in\Pc$. We have
    \begin{equation}\label{eq:DeltaDual}
        \Delta(p^\gamma_\ast) = \sum_{\alpha,\beta\in\Pc}  \Pi_{\alpha\beta}\null^\gamma\,p^\alpha_\ast \otimes p^\beta_\ast\,,
    \end{equation}
    where $\Pi_{\alpha\beta}\null^\gamma$ are the structure constants of the Newton basis, defined by equation \eqref{eq:ProdNewton}. In other words, for any matrices $A,B$:
    \begin{equation}
        \int_{\U(N)} DU\, p^\gamma_\ast(A+UB\Ud) = \sum_{\alpha,\beta\in\Pc}  \Pi_{\alpha\beta}\null^\gamma\,p^\alpha_\ast(A)\, p^\beta_\ast(B)\,
.    \end{equation}
\end{theorem}

\begin{proof}
    We start with the decomposition \eqref{eq:Dual} of the HCIZ integral and take the average over a sum of conjugacy orbits in its first argument:
    \begin{equation}
        \sum_{\gamma\in\Pc} z^{d(\gamma)}\,\Delta (p^\gamma_\ast)(A,B)\,p_\gamma(C)  = \int_{\U(N)} DU \,Z(A+UB\Ud,C\,; z)\,.
    \end{equation}
    We now use the multiplicativity \eqref{eq:MultZ} of the HCIZ integral to rewrite the right-hand side as $Z(A,C\,; z)Z(B,C\,; z)$, hence
    \begin{equation}
        \sum_{\gamma\in\Pc} z^{d(\gamma)}\,\Delta (p^\gamma_\ast)(A,B)\,p_\gamma(C)  = \left( \sum_{\alpha\in\Pc} z^{d(\alpha)}\,p^\alpha_\ast(A)\,p_\alpha(C) \right) \left( \sum_{\beta\in\Pc} z^{d(\beta)}\,p^\beta_\ast(B)\,p_\beta(C) \right)\,.
    \end{equation}
    Expanding and using the product rule \eqref{eq:ProdNewton}, we get
    \begin{equation}
        \sum_{\gamma\in\Pc} z^{d(\gamma)}\,\Delta (p^\gamma_\ast)(A,B)\,p_\gamma(C) = \sum_{\alpha,\beta,\gamma\in\Pc} z^{d(\alpha)+d(\beta)}\,\Pi_{\alpha\beta}\null^\gamma\,p^\alpha_\ast(A) p^\beta_\ast(B)\,p_\gamma(C)\,.
    \end{equation}
    The grading property of the algebra $\AcN$ implies that $\Pi_{\alpha\beta}\null^\gamma$ is non-vanishing only if $d(\alpha)+d(\beta)=d(\gamma)$. Recalling also the functional interpretation of the tensor product, we finally get
    \begin{equation}
        \sum_{\gamma\in\Pc} z^{d(\gamma)}\,\Delta (p^\gamma_\ast)(A,B)\,p_\gamma(C) = \sum_{\gamma\in\Pc} z^{d(\gamma)}\left(\sum_{\alpha,\beta\in\Pc}\Pi_{\alpha\beta}\null^\gamma\,(p^\alpha_\ast\otimes p^\beta_\ast)(A,B)\right)\,p_\gamma(C)\,.
    \end{equation}
    Extracting the coefficient of $z^{d(\gamma)}\,p_\gamma(C)$ on both side then proves the theorem.
\end{proof}

Theorem \ref{Thm:DeltaDual} provides an explicit formula for the orbit coproduct $\Delta$, which is also quite convenient since the structure constants $\Pi_{\alpha\beta}\null^\gamma$ of the Newton basis are particularly simple (recall in particular their expression \eqref{eq:PiConc} valid for degrees $d(\gamma)\leq N$). An important example of this formula is obtained by specialising to $\gamma=[n]$ with $n\leq N$. In this case, there are only two non-vanishing structure constants $\Pi_{\alpha\beta}\null^{[n]}$, which are $\Pi_{[n]\,\emptyset}\null^{[n]}=1$ and $\Pi_{\emptyset\,[n]}\null^{[n]}=1$, where we recall that $\emptyset$ is the empty partition of degree $0$ (see footnote \ref{foot:Empy}). The result \eqref{eq:DeltaDual} then gives $\Delta p^{[n]}_\ast = p^{[n]}_\ast \otimes 1 + 1 \otimes p^{[n]}_\ast$. This is to be expected since the dual polynomial $p^{[n]}_\ast$ is proportional to the precursor $K_n$ by equation \eqref{eq:pstarK}: we then recover the primitivity \eqref{eq:KPrim} of $K_n$. The Theorem \ref{Thm:DeltaDual} can thus be thought of as a vast generalisation of the additivity of the precursors $K_n$. If $\gamma\vdash n$ with $n\leq N$, we can use it to compute the orbit coproduct of the generalised precursor $K_\gamma$ using its relation \eqref{eq:pstarK} with $p^\gamma_\ast$. After a few manipulations, this gives\footnote{Note that when $\text{concatenation}(\alpha,\beta)=\gamma$, we have $\ell(\alpha)+\ell(\beta)=\ell(\gamma)$, explaining why powers of $N$ decouple from equation \eqref{eq:OrbitK}. Moreover, we stress that the sum over the partitions $\alpha,\beta$ also includes the empty partition $\emptyset$, for which $d(\emptyset)=0$, $|\Cl_\emptyset|=1$, $\text{concatenation}(\alpha,\emptyset)=\alpha$ and $K_{\emptyset}(A)=1$ by convention.}
\begin{equation}\label{eq:OrbitK}
    \Delta K_\gamma(A,B)=\int DU\,K_\gamma(A+U B\Ud) = \sum_{\alpha,\beta \text{ such that} \atop  \text{concatenation}(\alpha,\beta)=\gamma} \frac{|\Cl_\alpha|}{d(\alpha)!}\,\frac{|\Cl_\beta|}{ d(\beta)!} \,\frac{d(\gamma)!}{|\Cl_\gamma|}\,K_\alpha(A)\,K_\beta(B)\,.
\end{equation}

In principle, Theorem \ref{Thm:DeltaDual} allows to compute the value of any invariant polynomial averaged over sums of conjugacy orbits. To do so in practice, one first have to express this polynomial in terms of the dual Newton basis $\lbrace p^\alpha_\ast \rbrace_{\alpha\in\Pc}$: this is the price to pay for the apparent simplicity of the formula \eqref{eq:DeltaDual}, where the complexity of the coproduct is hidden in the choice of basis in which it is expressed. One can translate these results for the coproduct into a more standard basis, such as the Newton one. In this formulation, the expression for the coproduct then becomes more involved:

\begin{corollary}\label{cor:DeltaNewton}
Let $\alpha\in\Pc$. Then we have
\begin{equation}
    \Delta(p_\alpha) = \sum_{\beta,\gamma,\alpha',\beta',\gamma'\in\Pc} \eta_{\alpha\alpha'}\,\eta^{\beta\beta'}\,\eta^{\gamma\gamma'} \Pi_{\beta'\gamma'}\null^{\alpha'}\,p_\beta\otimes p_\gamma\,.
\end{equation}
\end{corollary}

\paragraph{Algebra/coalgebra duality.} Theorem \ref{Thm:DeltaDual} can be rephrased more abstractly as follows.

\begin{corollary}\label{cor:DeltaDual2}
    The algebra $(\AcN,\cdot)$ is dual to the coalgebra $(\AcN,\Delta)$ with respect to the non-degenerate pairing $\eta(\cdot,\cdot)$.\\
    If $\lbrace f_\alpha\rbrace_{\alpha\in\mathcal{I}}$ and $\lbrace f^\alpha_\ast\rbrace_{\alpha\in\mathcal{I}}$ are homogeneous bases of $\AcN$ dual with respect to $\eta(\cdot,\cdot)$, then
    \begin{eqnarray}
        \Delta(f^\gamma_\ast) = \sum_{\alpha,\beta\in\mathcal{I}} \Lambda_{\alpha\beta}\null^\gamma\,f^\alpha_\ast \otimes f^\beta_\ast\,, \qquad \text{ where } \qquad f_\alpha\,f_\beta = \sum_{\gamma\in\mathcal{I}} \Lambda_{\alpha\beta}\null^\gamma\,f_\gamma\,.
    \end{eqnarray}
\end{corollary}

Let us explain the content of this Corollary. The coalgebra structure $(\AcN,\Delta)$ is equivalent to an algebra structure $(\AcN\null^\ast,\Delta^\ast)$ on the dual space, where $\Delta^\ast: \AcN\null^\ast \otimes \AcN\null^\ast \to \AcN\null^\ast$ defines an associative and commutative product. On the other hand, the presence of the non-degenerate pairing $\eta(\cdot,\cdot)$ on $\AcN$ provides a natural identification between $\AcN$ and its dual $\AcN\null^\ast$. The corollary then states that this identification is also an isomorphism of algebra from $(\AcN,\cdot)$ to $(\AcN\null^\ast,\Delta^\ast)$. The proof is straightforward when writing the explicit definitions of these different algebraic structures in the Newton basis and its dual.

The main practical consequence of Corollary \ref{cor:DeltaDual2} is the second formulation of this duality property, namely that the structure constants of the product of $\AcN$ in any given homogeneous basis $\lbrace f_\alpha\rbrace_{\alpha\in\mathcal{I}}$ coincide with the structure constants of the orbit coproduct $\Delta$ on the dual basis $\lbrace f^\alpha_\ast\rbrace_{\alpha\in\mathcal{I}}$. This extends the result \eqref{eq:DeltaDual} to arbitrary bases of $\AcN$ and could also be derived naturally from the multiplicativity \eqref{eq:MultZ} of the HCIZ integral and its expansion
\begin{equation}
    Z(A,B\,; z) = \sum_{\alpha\in\mathcal{I}} z^{\text{deg}(f_\alpha)}\,f^\alpha_\ast(A)\,f_\alpha(B)\,.
\end{equation}

\begin{remark}
    Recall from Remark \ref{rem:StandardCoproduct} that the algebra of  symmetric polynomials is equipped with a natural coproduct. Let us further discuss the similarities and differences between this coproduct and the orbit one $\Delta$ defined here. While the latter is dual to the product of symmetric functions through the pairing $\eta(\cdot,\cdot)$, the former is known to be dual to the same product but with respect to another well-known pairing on symmetric functions, called the Hall scalar product. With respect to this Hall pairing, the Newton basis is orthogonal, so that the standard coproduct is already quite simple in this basis. In particular, the ``single traces'' $p_{[n]}$ turn out to be primitive elements. This is in contrast with our case at hand, where the primitive elements are the dual polynomials $p^{[n]}_\ast$ (proportional to the precursors), since the pairing $\eta(\cdot,\cdot)$ is quite more involved and in particular not diagonal in the Newton basis.
\end{remark}

\paragraph{Expressions for other bases.} So far, we have described the orbit coproduct $\Delta$ and the bilinear pairing $\eta(\cdot,\cdot)$ in the Newton basis and its dual. There are other natural bases of $\AcN$, as for instance the Schur polynomials $\lbrace s_\lambda \rbrace_{\ell(\lambda)\leq N}$.\footnote{This Schur basis $\lbrace s_\lambda \rbrace_{\ell(\lambda)\leq N}$ is labeled by partitions of length smaller or equal to $N$: this corresponds to Young diagrams with at most $N$ rows. In contrast, the Newton basis $\lbrace p_\alpha \rbrace_{\alpha\in\Pc}$ is labeled by the partitions $\Pc$ whose constituents are all smaller or equal to $N$, corresponding to Young diagrams with at most $N$ columns. For a fixed degree $n$, there is the same number of these two types of diagrams, as expected since they label two different bases of $\AcN_n$.} This Schur basis is particularly adapted to expand the HCIZ integral, as made obvious by the diagonal expansion \eqref{Z2ss}. In the language of this section, this means that this basis is orthogonal with respect to the bilinear pairing $\eta(\cdot,\cdot)$:
\begin{equation}\label{eq:EtaSchur}
    \eta(s_\lambda,s_\mu) = C_\lambda\,\delta_{\lambda\mu}\,, \qquad \text{ for } \ell(\lambda),\ell(\mu) \leq N\,.
\end{equation}
In particular, the dual of the Schur basis is itself up to a proportionality factor: $s^\lambda_\ast = C_\lambda^{-1}s_\lambda$. It is well-known that the structure constants of $\AcN$ in the Schur basis are the Littlewood--Richardson coefficients $c_{\lambda,\mu}\null^\nu$:
\begin{equation}
    s_\lambda\,s_\mu = \sum_{\ell(\nu)\leq N} c_{\lambda,\mu}\null^\nu\,s_\nu\,.
\end{equation}
It is now straightforward to write the orbit coproduct in terms of the Schur basis, as an application of Corollary \ref{cor:DeltaDual2}.

\begin{proposition}
    Let $\nu$ be a partition with $\ell(\nu)\leq N$. We then have
    \begin{equation}\label{eq:DeltaSchur}
        \Delta(s_\nu) = \sum_{\ell(\lambda),\ell(\mu) \leq N} \frac{C_\nu}{C_\lambda\,C_\mu}\,c_{\lambda,\mu}\null^\nu\,s_\lambda \otimes s_\mu\,.
    \end{equation}
\end{proposition}

We finally consider yet another family of invariant functions, the elementary symmetric ones $\lbrace e_n \rbrace_{n=0,\dots,N}$, appearing in the characteristic polynomial $\det(x\,\text{Id}-A) = \sum_{n=0}^N (-1)^{N-n}\,x^n e_{N-n}(A)$ -- note that these generate $\AcN$ as an algebra but do not form a linear basis of it. The expression of the orbit coproduct on these functions can be derived from the results of~\cite{Marcus}, yielding the following proposition.

\begin{proposition}\label{prop:e}
    Let $n\in\lbrace 1,\dots,N\rbrace$. We then have
    \begin{equation}\label{eq:DeltaE}
        \Delta(e_n) = \sum_{0\leq i,j\leq n\atop i+j=n} \frac{(N-i)!(N-j)!}{N!(N-i-j)!}\, e_i \otimes e_j\,.
    \end{equation}
\end{proposition}

\begin{remark}
    Let us compare again our results with those concerning the standard coproduct and the Hall scalar product on symmetric functions. It is well-known that Schur polynomials $s_\lambda$ are orthonormal with respect to the Hall scalar product, in contrast with the orthogonality property \eqref{eq:EtaSchur} found here for the pairing $\eta(\cdot,\cdot)$. As a consequence, the Schur polynomials are self-dual with respect to the Hall pairing and the standard coproduct of $s_\nu$ simply reads $\sum_{\lambda,\mu} c_{\lambda,\mu}\null^\nu\,s_\lambda \otimes s_\mu$ (\ie removing the ``content'' factors compared to the equation \eqref{eq:DeltaSchur} for the orbit coproduct).

    This also suggests a relation between the pairing $\eta(\cdot,\cdot)$ and the Hall scalar product. Namely, the former can be obtained from the latter by inserting a linear operator, which acts diagonally on the Schur basis $\lbrace s_\lambda\rbrace_{\ell(\lambda)\leq N}$ by the factor $C_\lambda$. Such an operator is in fact well-known in the representation theory of symmetric and unitary groups and their Schur-Weyl duality: it is built as a natural combination of the Jucys-Murphy elements and plays an important role in the abstract formulation of the  Weingarten--Samuel--Collins calculus (see for instance~\cite{Collins03,CollinsSn06,MatNov,ZJ09}).

    Let us finally comment on the elementary symmetric polynomial $e_n$, whose orbit coproduct is given by equation \eqref{eq:DeltaE}. For the standard coproduct on symmetric functions, this formula is replaced by the even simpler one $\sum_{i+j=n} e_i \otimes e_j$.
\end{remark}


\section{Averaged precursors and sums of independent random matrices}
\label{sec:Proba}

\subsection{Random matrices and averaged polynomials}
\label{sec:Random}

So far, the generalised precursors $K_\alpha$ have been defined as invariant polynomials on $N\times N$ matrices, \ie as elements of the algebra $\AcN$ in the notations of Section \ref{Sec:Coproduct}. In particular, this definition is purely algebraic and does not require any probabilistic considerations. However, having in mind applications to the theories of random matrices and free probability, we will now consider the setup where we are given a random matrix $A$ and we want to study the scalar random variable $K_\alpha(A)$ and its properties, such as its expectation value.

Let us recall or introduce a few notations. We will denote by $\McN$ the space of $\U(N)$-invariant probability measures on $N\times N$ matrices. If $A$ is a random matrix with distribution $\mu\in\McN$, we will write $A\sim \mu$. For any invariant polynomial $f\in\AcN$, one can then consider the $\C$-valued random variable $f(A)$. Of particular importance will be its expectation value
\begin{equation}
    \fb[\mu] := \E_{A\sim\mu}\bigl( f(A) \bigr) = \int f(A)\,\mu(A)\,,
\end{equation}
where the last expression denotes the Lebesgue integral of $f$ on the space of $N\times N$ matrices, with respect to the measure $\mu$. As the notation suggests, this defines a functional $\fb : \McN \to \C$ on the space of matrix probability measures, which assigns a number $\fb[\mu]$ to any $\mu\in\McN$. In particular, one can define the \textit{averaged generalised precursors} $\Kb_\alpha$ in this way. Note that we have already computed the value of these functionals $\Kb_\alpha$ on a very specific choice of measure in $\McN$, namely the Gaussian Unitary Ensemble measure $\text{GUE}(N,\sigma)$, see Section \ref{int-wick}. We then found $\Kb_\alpha[\text{GUE}(N,\sigma)]=\delta_{\alpha,[2^n]}\sigma^{2n}$ (see Proposition \ref{prop:GUE}). This illustrates on a particular example how the functionals $\Kb_\alpha$ assign a number to a measure.

If $\fb$ and $\gb$ are two such averaged polynomials, one can naturally construct the linear combination $a\fb+b\gb$ (with $a,b\in\C$) and the product $\fb\,\gb$, which are also functionals on $\McN$. We will denote by $\OcN$ the space of functionals obtained by taking finite products and linear combinations of averaged polynomials, which is then an associative commutative unital algebra.\footnote{The unit element of $\AcN$ is the constant polynomial equal to 1. Its average is the constant functional which assigns 1 to each measure, defining the unit element of $\OcN$. By a slight abuse of notations, we will denote these unit
elements by 1.} Such functionals can be thought of as natural ``observables'' for random (one-)matrix models. We note that the map $\AcN\to\OcN, f \mapsto\fb$ is not a morphism of algebra, since
\begin{equation}
    \overline{fg}[\mu] \neq \fb[\mu]\, \gb[\mu]
\end{equation}
for generic polynomials $f,g\in\AcN$ and measure $\mu\in\McN$. Indeed, the average of a product is not the product of the averages in general. As an algebra, $\OcN$ is in fact quite more involved than $\AcN$. More precisely, if $\lbrace 1 \rbrace \cup \lbrace f_\alpha \rbrace_{\alpha\in\mathcal{I}}$ is a linear basis of $\AcN$ (as a vector space), then the averaged polynomials $\lbrace \fb_\alpha \rbrace_{\alpha\in\mathcal{I}}$ form a set of algebraically independent generators of $\OcN$ (as an algebra).\\

Among invariant polynomials, a particular role is played by the normalised traces $m_n(A)=\frac{1}{N}\Tr(A^n)$. Following the notation introduced above, we can consider their expectation values $\overline{m}_n[\mu]$: these are the moments of the matrix distribution $\mu$, in the standard sense of random matrix theory. We stress that so far, we have essentialy only used the purely algebraic objects $m_n$, rather than the integrated quantities $\overline{m}_n$. By a slight abuse of terminology, we used a similar name for the moment of a matrix $m_n(A)$ and the moment of a distribution $\overline{m}_n[\mu]$, letting the context and the different notations distinguish these two concepts. Similarly, we called $\kappa_n(A)$ the $n$-th free cumulant of the (deterministic) matrix $A$, while we call $\overline{\kappa}_n[\mu]$ the $n$-th free cumulant of the distribution $\mu$.\footnote{With this definition, the free cumulant $\overline{\kappa}_n[\mu]$ of a distribution is given by the expectation value of a certain polynomial built from normalised traces. For instance, the free cumulant of degree 2 is $\overline{\kappa}_2[\mu]=\overline{m}_2[\mu]-\overline{m_1^2}[\mu]$. In the theory of free probability, it would be more standard to define the free cumulants as polynomials of integrated moments $\overline{m}_k[\mu]$. For example, the free cumulant of degree 2 would be $\overline{m}_2[\mu]-\overline{m_1}[\mu]^2$, which does not coincide with the quantity $\overline{\kappa}_2[\mu]=\overline{m}_2[\mu]-\overline{m_1^2}[\mu]$ introduced above. This apparent contradiction can in fact be resolved. Indeed, the theory of free probability typically deals with random matrices in the large $N$ limit, for which fluctuations are suppressed. For instance, in this limit, we have $\overline{m_1^2}[\mu] \approx \overline{m_1}[\mu]^2$, so that the two notions of degree 2 free cumulant discussed above coincide. Our definition of $\overline{\kappa}_n[\mu]$ can thus be thought of as a specific choice for extending the standard definition of free cumulants to finite size matrices, with the focus put on the associated invariant polynomials $\kappa_n(A)$ and thus on algebraic considerations.}


\subsection{Sums of independent random matrices and convolution}

Suppose that we are given two independent $\U(N)$-invariant random matrices $A\sim \mu$ and $B\sim\nu$ and consider $C=A+B$. We will denote by $\mu\star\nu\in\McN$ the distribution of this random matrix, which is the \textit{convolution} of the measures $\mu$ and $\nu$. The goal of this subsection is to answer the following question: suppose we know the values of functionals $F\in\OcN$ on the measures $\mu$ and $\nu$, can we find the value of these functionals on the convolution $\mu\star\nu$. Recall that any functional in $\OcN$ is built from algebraic combinations of averaged invariant polynomials $\fb$, with $f\in\AcN$. To find the value of such functionals on $\mu\star\nu$, it is therefore enough to know the value of all averaged polynomials $\fb$ on this convolution. Our main object of interest is then
    \begin{equation}\label{eq:fConvolDef}
        \fb[\mu\star\nu] = \E_{C\sim\mu\star\nu}\bigl(f(C)\bigr) = \iint f(A+B)\,\mu(A)\,\nu(B)\,.
    \end{equation}
In other words, we are asking whether we can reconstruct the expectation values of invariant polynomials of $C=A+B$ from the expectation values of invariant polynomials of $A\sim\mu$ and $B\sim\nu$. The main result of this subsection is then the following theorem, which answers this question using the orbit coproduct introduced in Section \ref{Sec:Coproduct}.

\begin{theorem}\label{Thm:fConvolution}
Let $f\in\AcN$ be an invariant polynomial and $\mu,\nu\in\McN$ be $\U(N)$-invariant measures. We write the orbit coproduct of $f$ as $\Delta f=\sum_{i=1}^p g_i\otimes h_i$, with $g_i,h_i\in\AcN$. Then\footnote{In the second member of equation \eqref{eq:fbConv}, we have extended the ``averaging operator'' $\overline{\,\cdot\,} : \AcN \to \OcN$ to $\AcN \otimes \AcN \to \OcN \otimes \OcN$ and identified elements of $\OcN \otimes \OcN$ with functionals of two measures, by evaluating the first and second tensor factor on the first and second measure respectively.}
\begin{equation}\label{eq:fbConv}
    \fb[\mu\star\nu] = \overline{\Delta f}[\mu,\nu] = \sum_i \overline{g_i}[\mu]\,\overline{h_i}[\nu]\,.
\end{equation}
\end{theorem}

\begin{proof}
    We want to compute the quantity \eqref{eq:fConvolDef}. Since the measure $\nu$ is $\U(N)$-invariant, we can freely insert a conjugation of $B$ by $U\in \U(N)$ and integrate over the Haar measure, yielding:
    \begin{equation*}
        \fb[\mu\star\nu] = \iint \left( \int_{\U(N)} DU\,f(A+UB\Ud) \right)\,\mu(A)\,\nu(B)\,.
    \end{equation*}
    We recognise in the inner integral the definition of the orbit coproduct \eqref{eq:Coproduct} of $f$. Writing it as $\Delta f = \sum_{i=1}^p g_i \otimes h_i$, we get
    \begin{equation*}
        \fb[\mu\star\nu] = \iint \Delta f(A,B)\,\mu(A)\,\nu(B) = \sum_{i=1}^p \left( \int g_i(A) \mu(A) \right) \left( \int h_i(B) \nu(B) \right) = \sum_i \overline{g_i}[\mu]\,\overline{h_i}[\nu] \,,
    \end{equation*}
    where we have used Fubini's theorem to factorise the integrals in the second equality (in the probablistic language, we used the independence of the random variables $A$ and $B$). This proves the theorem.
\end{proof}

Theorem \ref{Thm:fConvolution} directly relates the expectation value of averaged polynomials on convoluted measures with the orbit coproduct $\Delta$ on the polynomials themselves. Since the results of Section \ref{Sec:Coproduct} allow one to compute the orbit coproduct of any invariant polynomial, the Theorem \ref{Thm:fConvolution} then provides a concrete (albeit involved) procedure to compute $F[\mu\star\nu]$ for any functional $F\in\OcN$, answering the question raised at the beginning of this subsection. In Section \ref{Sec:Coproduct}, equation \eqref{eq:OrbitK}, we have argued that the generalised precursors $K_\gamma$ have a quite simple orbit coproduct. In the present setup, we then get the following.

\begin{corollary}\label{cor:ConvK}
    Let $\mu,\nu\in\McN$ be invariant measures on $N\times N$ matrices and $\gamma$ be a partition of degree smaller or equal to $N$. Then
    \begin{equation}\label{eq:ConvK}
        \Kb_\gamma[\mu\star\nu] =  \sum_{\alpha,\beta \text{ such that} \atop  \text{concatenation}(\alpha,\beta)=\gamma} \frac{|\Cl_\alpha|}{d(\alpha)!}\,\frac{|\Cl_\beta|}{ d(\beta)!} \,\frac{d(\gamma)!}{|\Cl_\gamma|}\,\Kb_\alpha[\mu]\,\Kb_\beta[\nu]\,.
    \end{equation}
    In particular, for $n\leq N$,
    \begin{equation}
        \Kb_n[\mu\star\nu] = \Kb_n[\mu] + \Kb_n[\nu]\,.
    \end{equation}
\end{corollary}

This shows that the averaged precursors $\Kb_n$ are additive with respect to the convolution $\star$. We will develop further on this in the next subsection. The formula \eqref{eq:ConvK} then provides a generalisation of this additivity property for averaged generalised precursors $\Kb_\gamma$. For concreteness, let us give an explicit example of the use of this formula, by taking $\gamma=[1^2]$. A straightforward computation then gives
\begin{equation}\label{eq:ConvK11}
    \Kb_{[1^2]}[\mu\star\nu] = \Kb_{[1^2]}[\mu] + \Kb_{[1^2]}[\nu] + 2\,\Kb_{1}[\mu]\,\Kb_{1}[\nu]\,.
\end{equation}

\begin{remark}\label{rem:CoproductConv}
    The results of this subsection can be rephrased more formally as follows. There exists a \emph{``convolution coproduct''} $\Db : \OcN \to \OcN \otimes \OcN$ such that
    \begin{equation}
        F[\mu\star\nu] = \Db(F)[\mu,\nu]\,, \qquad \forall\,F\in\OcN\,, \;\; \forall\,\mu,\nu\in\McN\,.
    \end{equation}
    This coproduct is a morphism of algebra, with respect to the standard product $\cdot$  of functionals in $\OcN$, since
    \begin{equation*}
        \Db(aF+bG)[\mu,\nu]=a\,F[\mu\star\nu]+b\,G[\mu\star\nu]=a\,\Db(F)[\mu,\nu]+b\,\Db(G)[\mu,\nu]\,,
    \end{equation*}
    \begin{equation*}
        \Db(F\cdot G)[\mu,\nu]=F[\mu\star\nu]\,G[\mu\star\nu]=\Db(F)[\mu,\nu]\,\Db(G)[\mu,\nu]\,,
    \end{equation*}
    for all $F,G\in\OcN$ and $a,b\in\C$. Following Theorem \ref{Thm:fConvolution}, $\Db$ is then uniquely characterised by its values $\Db(\fb) = \overline{\Delta f}$ on the averaged polynomials $\fb$, since the latter are generators of $\OcN$.

    We note that $(\OcN,\cdot,\Db)$ defines a bialgebra, since the convolution coproduct is a morphism with respect to the product of functionals. This is in contrast with the algebra of invariant polynomials $\AcN$ equipped with the orbit coproduct $\Delta$ (see Section \ref{Sec:Coproduct}).
\end{remark}

\subsection{Collins--Gurau--Lionni cumulants and additive functionals}
\label{sec:CGL}

\paragraph{Definition.} Having derived Corollary \ref{cor:ConvK}, it is natural to wonder whether there exist other functionals in $\OcN$ which are additive with respect to the convolution $\star$, in addition to the averaged precursors $\Kb_n$. The answer to that question is positive and is a reformulation of the recent work~\cite{CGL} of Collins, Gurau and Lionni.\footnote{Note that their work generalises some of these results from random matrices to random tensors.} We summarise these results here and take the opportunity to discuss the similarities and the differences between their approach and ours. The key ingredient is
\begin{equation}\label{eq:Zbar}
    \Zb_{B,z}[\mu] = \E_{A\sim\mu}\bigl[e^{zN\,\Tr(AB)}\bigr] = \int e^{zN\,\Tr(AB)}\,\mu(A)\,,
\end{equation}
where $A$ is a random matrix which is integrated along a given measure $\mu\in\McN$, while $B$ is kept deterministic and will serve as a bookkeeping variable. The Collins-Gurau-Lionni (CGL) cumulants $\Kcgl_\alpha[\mu]$ are then extracted from the expansion of the logarithm of this quantity along Newton polynomials of $B$ -- see~\cite[Equation (4.22)]{CGL}:\footnote{Here, we have expanded $\log \Zb_{B,z}[\mu]$ along the Netwon basis $\lbrace p_\alpha(B)\rbrace_{\alpha\in\Pc}$, as defined in Section \ref{Sec:Coproduct} (see above equation \eqref{eq:ProdNewton}). We have used a different normalisation than~\cite{CGL}, which simplifies some of the discussions below.}
\begin{equation}\label{eq:CGL}
    \log \Zb_{B,z}[\mu] =  \sum_{\alpha\in\Pc} z^{d(\alpha)}  \frac{N^{2-\ell(\alpha)}|\Cl_\alpha |}{d(\alpha)!} \Kcgl_\alpha[\mu]\,  p_\alpha(B)\,.
\end{equation}
We stress that the quantities $\Kcgl_\alpha$ are not invariant polynomials but functionals of the measure $\mu$. Because we have taken the logarithm after the integration over $A$, the $\Kcgl_\alpha$ are generally not expectation values of polynomials, but rather polynomials of such expectation values.

\paragraph{Relation with $\boldsymbol{\Kb_\alpha}$.} To make this more explicit, let us make the link with the averaged generalised precursors $\Kb_\alpha$ defined earlier. To do so, we recall that the distribution of the random matrix $A$ is assumed to be $\U(N)$-invariant. We can therefore freely insert a conjugation of $A$ by $U\in \U(N)$ in the definition \eqref{eq:Zbar} and integrate over the Haar measure, making the HCIZ integral \eqref{hciz} appear. We then get
\begin{equation}
    \Zb_{B,z}[\mu] = \E_{A\sim\mu}\bigl[Z(A,B\,; z)\bigr] = \int Z(A,B\,; z)\,\mu(A)\,.
\end{equation}
Now, recall from equation \eqref{expand2} that the HCIZ integral also has a natural expansion in terms of Newton polynomials $p_\alpha(B)$, using the generalised precursors $K_\alpha(A)$ (at least up to order $z^N$). We then get
\begin{equation}\label{eq:ZbarKb}
    \Zb_{B,z}[\mu] = \sum_{\alpha \text{ with } \atop d(\alpha)\leq N} z^{d(\alpha)}  \frac{N^{\ell(\alpha)}|\Cl_\alpha |}{d(\alpha)!} \Kb_\alpha[\mu]\,  p_\alpha(B) + {\rm O}(z^{N+1})\,.
\end{equation}
It is then a standard exercise to pass to the logarithm, reorganise the series along the Newton basis $p_\alpha(B)$ and compare to equation \eqref{eq:CGL}. This expresses the CGL cumulants $\Kcgl_\alpha[\mu]$ as polynomials in the averaged generalised precursors $\Kb_\beta[\mu]$ (of degrees $d(\beta)\leq d(\alpha)\leq N$). The simplest are the coefficients of single traces $p_{[n]}(B)=\Tr(B^n)$ (with $n\leq N$), which are not affected when passing to the logarithm. We then simply have
\begin{equation}\label{eq:CGLn}
    \Kcgl_{[n]}[\mu]=\Kb_n[\mu]\,,
\end{equation}
\ie the CGL cumulant associated with the partition $[n]$ coincides with the averaged precursor $\Kb_n$. The ones corresponding to arbitrary partitions are more involved. The first few read
\begin{subequations}
\begin{eqnarray}
    \Kcgl_{[1^2]}[\mu] &=& N^2\bigl(\Kb_{[1^2]}[\mu] - \Kb_1[\mu]^2\bigr)\,,\label{eq:CGL11} \\
    \Kcgl_{[2,1]}[\mu] &=& N^2\bigl(\Kb_{[2,1]}[\mu] - \Kb_1[\mu]\, \Kb_2[\mu]\bigr)\,,\\
    \Kcgl_{[1^3]}[\mu] &=& N^4\bigl(\Kb_{[1^3]}[\mu] - 3\,\Kb_1[\mu]\, \Kb_{[1^2]}[\mu] + 2\,\Kb_1[\mu]^3\bigr)\,.
\end{eqnarray}
\end{subequations}
In particular, it is clear that these general $\Kcgl_\alpha$ are not obtained as expectation values of invariant polynomials but rather as algebraic combinations of such expectation values. More precisely, for $d(\alpha)\leq N$, one can check that $\Kcgl_\alpha$ is a polynomial in the $\Kb_\beta$ with $d(\beta)\leq d(\alpha)$, whose coefficients are independent of $N$ up to a global factor $N^{2\ell(\alpha)-2}$. This general relation, which we will not need here, can be written in terms of the Möbius function of the lattice of set-partitions.\footnote{The CGL cumulants $\Kcgl_\alpha$ of degrees $d(\alpha)>N$ are not directly expressed in terms of $\Kb_\beta$, since the latter are only defined for $d(\beta)\leq N$. However, they can be expressed in terms of averaged dual Newton polynomials $\overline{p^\ast_\beta}$, following the expansion \eqref{eq:Dual} of the HCIZ integral introduced in Section \ref{Sec:Coproduct}. Indeed, these functionals $\overline{p^\ast_\beta}$ coincide (up to a global factor) with $\Kb_\beta$ for $d(\beta)\leq N$ but generalise them for higher degrees.}

\paragraph{Additivity.} The main property of the CGL cumulants is their additivity with respect to the convolution $\star$, as established in~\cite[Proposition 4.9]{CGL}. Namely, we have
\begin{equation}\label{eq:CGLAdd}
    \Kcgl_\alpha[\mu\star\nu] = \Kcgl_\alpha[\mu] + \Kcgl_\alpha[\nu]\,, \qquad\; \forall\,\alpha\in\Pc\,, \quad \forall\,\mu,\nu \in\McN\,.
\end{equation}
This is easily proven by observing that
\begin{equation*}
    \Zb_{C,z}[\mu\star\nu] = \iint e^{zN\,\Tr((A+B)C)}\mu(A)\,\nu(B) = \Zb_{C,z}[\mu]\,\Zb_{C,z}[\nu]
\end{equation*}
and passing to the logarithm. The above identity is the analogue of the multiplicative property \eqref{eq:MultZ} of the HCIZ integral. The main difference is that the latter involves completely deterministic matrices $A,B,C$ and an integration over a conjugacy orbit on one side, while the former involves already integrated objects (keeping only the bookkeeping matrix $C$ deterministic). This is what allows the passage to the logarithm and thus the addivity of the CGL cumulants (while one could not derive any analogue result from the logarithm of \eqref{eq:MultZ} since the logarithm and the integral cannot be exchanged). This illustrates the crucial distinction between the algebraic setup of the previous section (with invariant polynomials) and the probabilistic setup of the current section (with integrated quantities, \ie functionals of the measure). In the former case, the search for polynomials which are additive under the addition of conjugacy orbits only leads to the precursors $K_n(A)$. In the latter case, functionals which are additive under the addition of independent random matrices include the averaged precursors $\Kb_n[\mu]$ plus many other functionals $\Kcgl_\alpha[\mu]$, built from various algebraic combinations of averaged polynomials.\\

For length-one partitions $\alpha=[n]$ ($n\leq N$), the additivity \eqref{eq:CGLAdd} should not come as a surprise. Indeed, as pointed out in equation \eqref{eq:CGLn}, the corresponding CGL cumulants $\Kcgl_{[n]}$ coincide with the averaged precursors $\Kb_n$, which were proven to be additive in Corollary \ref{cor:ConvK}. For concreteness, let us explicitly check the additivity for a more general partition $\alpha=[1^2]$. The corresponding CGL cumulant was computed in equation \eqref{eq:CGL11}. We then have
\begin{eqnarray}
    \Kcgl_{[1^2]}[\mu\star\nu] = N^2\bigl(\Kb_{[1^2]}[\mu\star\nu] - \Kb_1[\mu\star\nu]^2\bigr)\,.
\end{eqnarray}
The averaged precursor $\Kb_1$ is additive, while the behaviour of $\Kb_{[1^2]}$  with respect to the convolution is a bit more involved and was derived in equation \eqref{eq:ConvK11}. Combining those, we get
\begin{equation}
    \Kcgl_{[1^2]}[\mu\star\nu] = N^2\bigl(\Kb_{[1^2]}[\mu] + \Kb_{[1^2]}[\nu] + 2\,\Kb_{1}[\mu]\,\Kb_{1}[\nu] - (\Kb_1[\mu]+\Kb_1[\nu])^2\bigr)\,.
\end{equation}
Expanding the square, it is clear that the crossed terms (mixing $\mu$ and $\nu$) cancel, leaving the desired result $\Kcgl_{[1^2]}[\mu\star\nu]=\Kcgl_{[1^2]}[\mu]+\Kcgl_{[1^2]}[\nu]$. As a consistency check, we have performed similar verifications of the additivity of $\Kcgl_\alpha$ for all partitions up to degree $d(\alpha)=8$. 

\begin{remark}
    The additivity \eqref{eq:CGLAdd} of CGL cumulants $\Kcgl_\alpha$ means that they are primitive elements for the convolution coproduct $\Db$ introduced in remark \ref{rem:CoproductConv}.
\end{remark}

 \paragraph{Discussion.} Let us offer an alternative point of view on the CGL cumulants and their additivity. Any linear form on the space of $N\times N$ matrices can be realised as a map $\lambda_C: A \mapsto N\,\Tr(AC)$ for a given matrix $C$. If $A\sim\mu$ is a random matrix, the quantity
\begin{equation}
    \Zb_{C,z}[\mu] = \E_{A\sim \mu}\bigl( e^{z\,\lambda_C(A)} \bigr)
\end{equation}
can be interpreted as the exponential generating function of the moments of the scalar random variable $\lambda_C(A)$. The logarithm $\log\Zb_{C,z}[\mu]$ is then the generating function of the classical cumulants $c^{(n)}_{A\sim\mu}\bigl(\lambda_C(A)\bigr)$ of this random variable. By equation \eqref{eq:CGL}, these $c^{(n)}_{A\sim\mu}\bigl(\lambda_C(A)\bigr)$ are therefore also expressed as linear combinations of the CGL cumulants $\Kcgl_\alpha[\mu]$ (with $C$-dependent coefficients). Reciprocally, letting $C$ vary allows one to reconstruct all CGL cumulants from the data of the classical cumulants $c^{(n)}_{A\sim\mu}\bigl(\lambda_C(A)\bigr)$.

If $A\sim\mu$ and $B\sim\nu$ are independent random matrices, the variables $\lambda_C(A)$ and $\lambda_C(B)$ are obviously independent, and we then have
\begin{equation}
    c^{(n)}_{A+B\sim\mu\star\nu}\bigl(\lambda_C(A+B)\bigr) = c^{(n)}_{A\sim\mu}\bigl(\lambda_C(A)\bigr) + c^{(n)}_{B\sim\nu}\bigl(\lambda_C(B)\bigr)
\end{equation}
by the additivity of classical cumulants. This provides an alternative route towards the additivity of the CGL cumulants. We note however that the latter have an interesting advantage compared to the classical cumulants $c^{(n)}_{A\sim\mu}\bigl(\lambda_C(A)\bigr)$. Indeed, since we are restricting ourselves to $\U(N)$-invariant measures $\mu\in\McN$, these standard cumulants are not all independent. For instance, the linear forms $A \mapsto A_{11}$, $A\mapsto A_{22}$ and $A\mapsto \frac{1}{N}\Tr(A)$ all lead to the same cumulants, due to the $\U(N)$-invariance. By keeping $C$ as an auxiliary variable and expanding $\Zb_{C,z}[\mu]$ along the Newton basis $\lbrace p_\alpha(C)\rbrace_{\alpha\in\Pc}$, we extract a basis of linearly independent functionals with the desired additivity property, namely the CGL cumulants $\Kcgl_\alpha[\mu]$.

\paragraph{Large $\boldsymbol{N}$ limit.} Let us now discuss the large $N$ limit of averaged precursors and CGL cumulants. As usual, we suppose that we are given a $N$-dependent measure $\mu\in\McN$ with an appropriately defined $N\to\infty$ limit $\mu_\infty$, such that the (averaged) moments $\mb_n[\mu]$ all have a finite limit $\mb_n[\mu_\infty]$. As explained in equation \eqref{eq:limKn}, the precursor $K_n(A)$ tends to the free cumulant $\kappa_n(A)$ (in terms of the invariant polynomials $m_k(A)$). Taking expectation values, the averaged precursor $\Kb_n[\mu]$, which coincides with the CGL cumulant $\Kcgl_{[n]}[\mu]$ by equation \eqref{eq:CGLn}, then has a finite $N\to\infty$ limit:
\begin{equation}\label{eq:LimCGLn}
    \lim_{N\to\infty} \Kcgl_{[n]}[\mu] = \lim_{N\to\infty} \Kb_{n}[\mu] = \kb_n[\mu_\infty]\,. 
\end{equation}
The right-hand side of this equation is the $n$-th free cumulant of the limiting distribution $\mu_\infty$, as standardly defined in free probability.

The large $N$ behaviour of CGL cumulants $\Kcgl_\alpha$ associated with more general partitions is quite more subtle. Let us explain this on a simple example, by considering $\alpha=[1^2]$. The corresponding CGL cumulant was expressed in equation \eqref{eq:CGL11}. Translating this expression in terms of moments, using equations \eqref{Kmrel} and \eqref{eq:K11}, we get
\begin{equation}
    \Kcgl_{[1^2]}[\mu] = \frac{N^2}{N^2-1}(\overline{m_1^2-m_2})[\mu] + N^2 \bigl(\overline{m_1^2}[\mu] -  \overline{m_1}[\mu]^2 \bigr)\,.
\end{equation}
The first term has a finite $N\to\infty$ limit ``on the nose''. The second term is more subtle, due to the factor $N^2$. However, the combination in the brackets is the variance $\text{var}_{A\sim \mu}\left(\frac{1}{N}\Tr(A)\right)$, which is well known to be suppressed as ${\rm O}(N^{-2})$ when $N\to\infty$. The whole cumulant $\Kcgl_{[1^2]}[\mu]$ thus has a finite and non-zero limit. We note that the coefficient $N^2$ in the formula \eqref{eq:CGL11} was crucial to obtain a quantity of order ${\rm O}(N^0)$ and not ${\rm O}(N^{-2})$.

More generally, the normalisation in equation \eqref{eq:CGL} has been chosen such that the CGL cumulants $\Kcgl_\alpha$ have a finite and non-zero $N\to\infty$ limit. More precisely, according to~\cite[Equation (3.27)]{CGL}, $\lim_{N\to\infty}\Kcgl_\alpha[\mu]$ becomes the higher-order free cumulant~\cite{CMSS06} associated with the partition $\alpha=(\alpha_1,\dots,\alpha_{\ell})$ (in a certain normalisation). This is consistent with the addivity property \eqref{eq:CGLAdd} of $\Kcgl_\alpha$, since higher-order free cumulants are known to be additive quantities with respect to the free convolution.\vspace{-2pt}

\paragraph{An application.} We end the section with a brief application of this formalism. In Section \ref{sec:A11n}, we have studied the behaviour of powers of the top diagonal component of a matrix sampled over a conjugacy orbit. Let us translate this result in the language of this subsection. Suppose that $A\sim \mu$ is a random matrix with $\U(N)$-invariant measure $\mu\in\McN$. The top component $A_{11}$ can be written as $A_{11}=\Tr(AB)$, where $B=(\delta_{i1}\delta_{j1})_{1\leq i,j\leq N}$ is the rank one matrix used in Section \ref{sec:A11n}. The generating function for the $n$-th moments of the random variable $A_{11}$ is\vspace{-2pt}
\begin{equation}
    \Zb_{B,z}[\mu] = \E_{A\sim \mu}[e^{zN A_{11}}] = \sum_{n=0}^{+\infty} \frac{z^n N^n}{n!} \E_{A\sim \mu}(A_{11}^n)\,,\vspace{-2pt}
\end{equation}
where $\Zb_{B,z}[\mu]$ is defined (for general $B$) in equation \eqref{eq:Zbar}. Using the expansion \eqref{eq:ZbarKb} of this object, with the fact that $p_\alpha(B)=1$ for all partitions $\alpha$ for our specific choice of $B$, we get\vspace{-2pt}
\begin{equation}\label{eq:mA11}
     \E_{A\sim\mu}(A_{11}^n) =  \sum_{\alpha\vdash n} |\Cl_\alpha|\, N^{\ell(\alpha)-n}  \Kb_\alpha[\mu]\,,\vspace{-2pt}
\end{equation}
for $n\leq N$. This is the averaged analogue of the result \eqref{A11n} found earlier.

Instead of the moments of $A_{11}$, let us now consider its (classical) cumulants $c^{(n)}_{A\sim\mu}(A_{11})$. It is well-known that these cumulants are extracted from the logarithm
\begin{equation}
    \log \Zb_{B,z}[\mu] = \sum_{n=1}^{+\infty} \frac{z^n N^n}{n!} c^{(n)}_{A\sim\mu}(A_{11})
\end{equation}
of the moments generating function. This logarithm is exactly the quantity from which the CGL cumulants are defined, following equation \eqref{eq:CGL}. We then find
\begin{equation}\label{eq:cA11}
     c^{(n)}_{A\sim\mu}(A_{11}) =  \sum_{\alpha\vdash n} |\Cl_\alpha|\, N^{2-\ell(\alpha)-n}  \Kcgl_\alpha[\mu]\,.
\end{equation}
We thus get parallel expressions \eqref{eq:mA11} and \eqref{eq:cA11} for the moments and cumulants of $A_{11}$, respectively expressed in terms of averaged generalised precursors $\Kb_\alpha$ and CGL cumulants $\Kcgl_\alpha$.

Recall that the $\Kcgl_\alpha$'s have a finite limit when $N\to\infty$. The dominant term in equation \eqref{eq:cA11} is therefore the contribution from the partition with smallest length, \ie $\alpha=[n]$. From equation \eqref{eq:LimCGLn}, we get
\begin{equation}
    c^{(n)}_{A\sim\mu}(A_{11}) \approx \frac{(n-1)!}{N^{n-1}}\, \kb_n[\mu_\infty]\,\quad \text{ when } N\to\infty,
\end{equation}
in terms of the $n$-th free cumulant $\kb_n[\mu_\infty]$. We recover this way a well-known identity~\cite{BoBo}, which was initially derived in~\cite{GM} using the large $N$ expansion of the HCIZ integral for a finite rank matrix $B$. The equation \eqref{eq:cA11} then expresses the $1/N$ corrections to this identity in terms of the CGL cumulants.


\section{Discussion}
\label{open}

\subsection[Finite $N$ free convolution associated with precursors?]{Finite $\boldsymbol{N}$ free convolution associated with precursors?}
\label{sec:convolution}

Another version of finite $N$ analogues of free cumulants has been studied by Arizmendi and Perales~\cite{AP}. Let us explain the similarities and differences with our approach. Their work is based on the notion of finite $N$ free convolution, introduced by Marcus, Spielman and Srivastava in~\cite{Marcus}, which starts with any two $\U(N)$-orbits of Hermitian matrices and builds a third one. Such an orbit is uniquely determined by the spectrum $(\alpha_1,\dots,\alpha_N)$ of the matrices $A$ within, or equivalently by their characteristic polynomial\vspace{-2pt}
\begin{equation}
    P_A(x)=\det(x\,\mathbb{I}-A) = \prod_{i=1}^N (x-\alpha_i) = \sum_{k=0}^{N} (-1)^k x^{n-k} e_k(A)\,.\vspace{-2pt}
\end{equation}
This is a real-rooted monic polynomial in $x$, whose coefficients are the elementary invariant polynomials $e_k(A)$. The finite $N$ free convolution of~\cite{Marcus} is defined as follows. Consider two Hermitian orbits with representative $A$ and $B$, characterised by the polynomials $P_A(x)$ and $P_B(x)$. The convoluted orbit is then formed by Hermitian matrices with characteristic polynomial\vspace{-2pt}
\begin{equation}
    (P_A \boxplus_N P_B)(x) := \int_{\U(N)} DU\,P_{A+UB\Ud}(x)\,,\vspace{-2pt}
\end{equation}
obtained by taking the average of characteristic polynomials over the sum of the $\U(N)$-orbits of $A$ and $B$. Two remarkable results of~\cite{Marcus} are then the following:\vspace{-2pt}
\begin{itemize}\setlength\itemsep{3pt}
    \item $P_A \boxplus_N P_B$ is a real-rooted monic polynomial,
    \item there exists $U_0\in \U(N)$ such that $P_A \boxplus_N P_B = P_{A+U_0B\Ud_0}$.\vspace{-2pt}
\end{itemize}
The first property ensures that $P_A \boxplus_N P_B$ can be interpreted as the characteristic polynomial of a conjugacy orbit of Hermitian matrices. The second one ensures that these Hermitian matrices lie exactly within the sum of the orbits of $A$ and $B$. This last result can also be phrased in terms of the roots $\alpha_i,\beta_i,\gamma_i\in \R$ $(1\leq i \leq N)$ of $P_A$, $P_B$ and $P_A \boxplus_N P_B$, as the fact that $\gamma_i$ satisfy the Horn inequalities~\cite{Ho62,Kly,KT99, Fu} associated with the eigenvalues $\alpha_i,\beta_i$ of $A$ and $B$. In this context, the finite $N$ cumulants of Arizmendi and Perales~\cite{AP} are invariant polynomials $f_n$ of degree $n$ such that $f_n(C)=f_n(A)+f_n(B)$ for any matrix $C$ in the convoluted orbit of $A$ and $B$. They are thus additive with respect to the finite $N$ convolution. This additivity property however, is different from the one \eqref{eq:AddK} of the precursors $K_n$ studied in this paper. 

Let us rephrase the notion of finite free convolution in terms of the elementary polynomials $e_n(A)$. A Hermitian matrix $C$ is in the convoluted orbit of $A$ and $B$ if and only if
\begin{equation}\label{eq:ConvE}
    e_n(C) = \int_{\U(N)} DU\,e_n(A+UB\Ud) = \sum_{0\leq i,j\leq n\atop i+j=n} \frac{(N-i)!(N-j)!}{N!(N-i-j)!}\, e_i(A) \, e_j(B) \,,
\end{equation}
for all $n\in\lbrace 1,\dots,N\rbrace$ (the last expression follows from the results of~\cite{Marcus}, see also Proposition \ref{prop:e}). This formulation raises a natural question: can we define a variant of the finite $N$ free convolution using another set of invariant polynomials than the elementary ones $e_n$? In particular, can we do so using the precursors $K_n$ considered in this paper? In this scenario, the convoluted orbit of $A$ and $B$ would then be formed by the matrices $C$ such that
\begin{equation}\label{eq:ConvK2}
     K_n(C) = \int_{\U(N)} DU\,K_n(A+UB\Ud) = K_n(A) + K_n(B)\,,
\end{equation}
for all $n\in\lbrace 1,\dots,N\rbrace$ (where the last equality is nothing but the additivity property \eqref{eq:AddK} of the the precursors). This condition can be rephrased as a system of $N$ algebraic equations on the eigenvalues $(\gamma_1,\dots,\gamma_N)$ of $C$ in terms of the eigenvalues $\alpha_i,\beta_i$ of $A$ and $B$. This system always admits a solution, unique up to permutation, so that the matrix $C$ exists, at least as a complex matrix. However, two important questions remain:\vspace{-2pt}
\begin{itemize}\setlength\itemsep{3pt}
    \item are the eigenvalues $(\gamma_1,\dots,\gamma_N)$ real?
    \item if they are, do they satisfy Horn's inequalities?\vspace{-2pt}
\end{itemize}
The first property would ensure that $C$ can be taken to be Hermitian while the second would mean $C$ belongs to the sum of the $\U(N)$-orbits of $A$ and $B$.

For $N=3$, one checks that the systems of equations defined by \eqref{eq:ConvE} and \eqref{eq:ConvK2} are equivalent, so that the two notions of convolution coincide and the answer to the above questions is positive. However, for $N\geq 4$, it is easy to find matrices $A$ and $B$ such that $C$ defined through \eqref{eq:ConvK2}
has eigenvalues that are not all real. For example, 
$A=\mathrm{diag}(6,5,4,-15)\, ,\ B=\mathrm{diag}(12,-3,-4,-5)$ lead to $C\approx U\mathrm{diag}(14.72 , -16.69, .98\pm .72\mathrm{i})  \Ud$. Moreover, there also exist cases where the eigenvalues of $C$ are real but violate Horn's inequalities. We conclude that, in contrast with the method of \cite{Marcus, AP}, the precursors $K_n$ are not associated with a finite $N$ convolution on Hermitian matrices.  It would be interesting to determine whether the elementary polynomials $e_n$ are the only ones which define such a well-behaved convolution.


\subsection{General orbit coproduct}

In Section \ref{Sec:Coproduct}, we have defined and studied the orbit coproduct $\Delta$, which described how conjugacy-invariant polynomials behave with respect to averaging over sums of $\U(N)$-orbits. This construction admits a straightforward extension to a more general setup, which we now briefly sketch. Let $G$ be a compact group and $V$ be a real representation of $G$. We will denote by $\mathcal{A}_V$ the algebra of polynomials on $V$ which are invariant under the action of $G$. We define the orbit coproduct $\Delta_V : \mathcal{A}_V \to \mathcal{A}_V \otimes \mathcal{A}_V$ by\vspace{-1pt}
\begin{equation}
    \Delta_V f(a,b) = \int_G DU\,f(a+U.b)\,,
\end{equation}
for all $f\in\mathcal{A}_V$ and all $a,b\in V$. The case where $G=\U(N)$ and $V$ is the space of $N\times N$ (complex or Hermitian) matrices, on which $G$ acts by conjugation, reduces to the setup considered in Section \ref{Sec:Coproduct}. In the general case, it is well-known that the representation $V$ always admits a $G$-invariant bilinear form $\langle \cdot,\cdot\rangle : V\times V \to \R$ (which can be obtained by averaging, but is generally not unique, even up to normalisation). From this bilinear form, we define the following integral:\vspace{-1pt}
\begin{equation}
    Z_V(a,b\,; z) = \int_G DU\,e^{z\,\langle a,\,U.b\rangle}\,,\vspace{-1pt}
\end{equation}
depending on $z\in\C$ and $a,b\in V$. It is clear that $Z_V(a,b\,; z)$ is invariant under the actions of $G$ on $a$ and $b$ (independently). Moreover, the coefficient of $z^n$ in its power series expansion can be written in terms of invariant polynomials of degree $n$ of both $a$ and $b$. The function $Z_V(a,b\,; z)$ generalises the HCIZ integral to our extended setup. One easily checks that it satisfies the following multiplicative identity, which is a generalisation of the property \eqref{eq:MultZ} of the HCIZ integral:\vspace{-1pt}
\begin{equation}
    \int_G DU\, Z_V(a+U.b,c\,; z) = Z_V(a,c\,; z)\, Z_V(b,c\,; z)\,.\vspace{-1pt}
\end{equation}
This identity allows one to determine the orbit coproduct $\Delta_V f$ of all the invariant polynomials $f\in\mathcal{A}_V$ which appear in the power series expansion of $Z_V(a,b\,; z)$, following an analogous method to the one developed in Section \ref{Sec:Coproduct}. This way, many of the results found in this section generalise to this extended setup. Similarly, most methods developed in Section \ref{sec:Proba} can also be generalised to describe the behaviour of averaged invariant polynomials with respect to the convolution of $G$-invariant measures on $V$. It would be interesting to explore these constructions in more detail and study explicit examples.

\section*{Acknowledgments}{It is a pleasure to acknowledge stimulating discussions with D. Bernard, J. Bouttier, \'E. Br\'ezin, P. Di Francesco, M. Ma\"ida, S.  Tarricone and P. Zinn-Justin. We are particularly indebted to P. Biane for drawing our attention to the point discussed in Section \ref{sec:convolution} and to R. Speicher for pointing out the reference~\cite{CC06}.}


\appendix
\section[Large $N$ multiplicativity of generalised precursors]{Large $\boldsymbol{N}$ multiplicativity of generalised precursors}
\label{App:MultK}

This appendix is devoted to the proof of equation \eqref{eq:MultKappa}. Our goal is therefore to study the large $N$ behaviour of the generalised precursors $K_\alpha$. Recall that these polynomials are defined only up to degree $N$. Throughout the appendix, we will fix an integer $n_0$ and will focus on the $K_\alpha$'s of degree lesser or equal to $n_0$: to ensure that they are well-defined, the equations that we will write will always assume that $N\geq n_0$ (which does not affect our argument since we are interested in large $N$ behaviours). We start with the expansion \eqref{expand2} of the HCIZ integral in terms of $K_\alpha(A)$ and $p_\alpha(B)$. Translating the latter to normalised moments, we get\footnote{Recall that for a partition $\alpha=(\alpha_1,\dots,\alpha_{\ell})$, we denote by $d(\alpha)=\alpha_1+\dots+\alpha_{\ell}$ its degree, such that $\alpha\vdash d(\alpha)$. The sum in equation \eqref{eq:appZK} then runs over all the partitions of integers between $1$ and $n_0$.}
\begin{eqnarray}\label{eq:appZK}
    Z(A,B\,; z) = 1 + \sum_{\substack{\text{partitions } \alpha\\ 1 \leq d(\alpha) \leq n_0}} {z^{d(\alpha)}}   
 N^{2\ell(\alpha)} \frac{|\Cl_\alpha |}{n!} K_\alpha(A)  m_\alpha(B) + {\rm O}(z^{n_0+1})\,,
\end{eqnarray}
where  $m_\alpha$  has been defined in \eqref{defmalpha}.
Recall that we consider the large $N$ limit with moments $m_k(A)$ and $m_k(B)$ kept finite. To show that $K_\alpha(A)$ has a well-defined limit, we thus need to prove that the coefficient $[z^{d(\alpha)} m_\alpha(B)]Z(A,B\,; z)$ of the HCIZ integral is of order ${\rm O}(N^{2\ell(\alpha)})$. Moreover, computing its dominant term would provide the expression of $\lim_{N\to\infty} K_\alpha(A)$. This is the main task we achieve in this appendix. 

The key ingredient we will need is the HCIZ free energy
\begin{equation}
    F(A,B\,; z) = \frac{1}{N^2} \log\bigl( Z(A,B\,; z) \bigr)\,.
\end{equation}
It is well known~\cite{Collins03, GZ, GGPN} that this quantity admits a topological expansion in large $N$ with only non-positive powers of $N$. For our purposes, it will be useful to write it as
\begin{equation}\label{eq:ExpandF}
    F(A,B\,; z) = \sum_{\substack{\text{partitions } \beta \\ 1 \leq d(\beta) \leq n_0}} z^{d(\beta)} L_\beta(A)\, m_\beta(B) + {\rm O}(z^{n_0+1})\,,
\end{equation}
where the polynomials $L_\beta(A)$ are of order ${\rm O}(1)$ when $N\to\infty$ and admit a topological expansion
\begin{equation}\label{eq:NExpandL}
    L_\beta(A) = \sum_{g\geq 0} N^{-2g} L_{\beta,g}(A)\,.
\end{equation}
These polynomials are related in a complicated fashion to the generalised precursors $K_\alpha(A)$, by taking the logarithm of \eqref{eq:appZK} and expanding the result in $m_\beta(B)$. The fact that $L_\beta(A)$ is finite in the large $N$ limit is non-trivial and results on various algebraic relations between the coefficients of the $1/N$-expansion of the $K_\alpha(A)$'s. Here, we will reverse the logic: we start with the fact that $L_\beta(A)$ admits a finite expansion \eqref{eq:NExpandL} and derive the corresponding consequences on the expansion of $K_\alpha(A)$.\\

Exponentiating \eqref{eq:ExpandF}, we get
\begin{equation}\label{eq:FtoZ}
    Z(A,B\,; z)  = \prod_{\substack{\text{partitions } \beta \\ 1 \leq d(\beta) \leq n_0}} \exp\left( z^{d(\beta)} N^2\,L_\beta(A)\, m_\beta(B) \right) + {\rm O}(z^{n_0+1})\,.
\end{equation}
Recall that $m_\beta(B)$ is a shorthand notation for the product $\prod_{i=1}^{\ell(\beta)} m_{\beta_i}(B)$. We will first focus on the ``single trace'' terms in the above expression, \ie the coefficient of $z^n m_n(B)$ for some $n\leq n_0$. It is clear that this coefficient can only be produced by the linear term corresponding to $\beta=[n]$ in the exponentials, such that
\begin{equation}
    [z^n m_n(B)] Z(A,B\,; z) = N^2\,L_{[n]}(A)\,.
\end{equation}
On the other hand, by the relation \eqref{eq:appZK}, this coefficient coincides with $\frac{N^2}{n}\,K_{n}(A)$. The finiteness of $L_{[n]}(A)$ in the large $N$ limit is therefore consistent with the observation that $K_n(A)$ is also finite. More precisely, we then relate the ``genus 0'' component of $L_{[n]}(A)$ with the limit of the precursor $K_n(A)$, which coincides with the free cumulant $\kappa_n(A)$ by equation \eqref{eq:limKn}:
\begin{equation}\label{eq:LimLK}
    L_{[n],0}(A) = \frac{1}{n} \lim_{N\to \infty} K_n(A) = \frac{\kappa_n(A)}{n}\,.
\end{equation}

We now turn to the more general coefficients in the expression \eqref{eq:FtoZ}. Expanding the exponentials and the products yields terms of the form
\begin{equation}\label{eq:ProdL}
    \frac{z^{b_1\,d(\beta^{(1)})+\cdots+b_p\,d(\beta^{(p)})}}{b_1!\cdots b_p!} N^{2(b_1+\dots+b_p)} L_{\beta^{(1)}}(A)^{b_1}\cdots L_{\beta^{(p)}}(A)^{b_p} \, m_{\beta^{(1)}}(B)^{b_1}\cdots m_{\beta^{(p)}}(B)^{b_p}\,,
\end{equation}
where the $\beta^{(i)}$'s are pairwise distinct partitions and the $b_i$'s positive integers. This contributes to the coefficient of $z^{d(\alpha)}\, m_\alpha(B)$ associated with the partition $\alpha$ built by concatenation of the $\beta^{(i)}$'s repeated $b_i$ times -- $\alpha$ is then of degree $d(\alpha)=\sum_i b_i\,d(\beta^{(i)})$ and length $\ell(\alpha)=\sum_i b_i\,\ell(\beta^{(i)})$. In general, there are many different combinations $(\beta^{(i)},b_i)$ which concatenate to the same given partition $\alpha$, so that the full coefficient of $z^{d(\alpha)} m_\alpha(B)$ in $Z(A,B\,; z)$ is a quite complicated object. To get something more tractable, we will focus on the dominant term of this coefficient when $N\to\infty$. Recall that the $L_{\beta^{(i)}}(A)$'s are all of order ${\rm O}(1)$ in this limit, so that \eqref{eq:ProdL} is of order ${\rm O}(N^{2(b_1+\dots+b_p)})$. This gets more dominant as we increase $b_1+\dots+b_p$, hence as the concatenation gets finer. The most dominant term therefore stems for the finest way of writing $\alpha$ as a concatenation. Using the notation\footnote{A quick comment on conventions is in order. For $\alpha\vdash n$, in section \ref{notations} we have often used the notation $\alpha=[1^{\hat{\alpha_1}}\dots n^{\hat{\alpha}_n}]$, keeping track of each number $k \in\lbrace 1,\dots,n\rbrace$ but allowing its multiplicity $\hat{\alpha}_k$ to be $0$ if it does not appear in $\alpha$. Here, we choose to omit all the numbers which do not appear in $\alpha$ and thus write $\alpha=[k_1^{\hat{\alpha}_{k_1}} \dots k_p^{\hat{\alpha}_{k_p}}]$, where the multiplicities $\hat{\alpha}_{k_i}$ are then non-zero.} $\alpha=[k_1^{\hat{\alpha}_{k_1}} \dots k_p^{\hat{\alpha}_{k_p}}]$ introduced in Section \ref{notations}, this finest concatenation is obtained from the choice where each $\beta_i=[k_i]$ is of minimal length $1$ and $b_i=\hat{\alpha}_{k_i}$ is the number of times $k_i$ appears in $\alpha$. The corresponding power of $N$ is then $2\sum_i \hat{\alpha}_{k_i} = 2\ell(\alpha)$, \ie twice the length of the partition $\alpha$. In conclusion, we find
\begin{equation}
    \bigl[z^{d(\alpha)} m_\alpha(B)\bigr] Z(A,B\,; z) = N^{2\ell(\alpha)} \left( \frac{1}{\hat{\alpha}_{k_1}!\cdots \hat{\alpha}_{k_p}!} L_{[k_1],0}(A)^{\hat{\alpha}_{k_1}} \cdots L_{[k_p],0}(A)^{\hat{\alpha}_{k_p}} + {\rm O}(N^{-1}) \right)\,.
\end{equation}
On the other hand, the relation \eqref{eq:appZK} tells us that this coefficient is also equal to $N^{2\ell(\alpha)} \frac{|\Cl_\alpha|}{n!} K_\alpha(A)$. As expected, this confirms the finiteness of $K_\alpha(A)$ at large $N$ and provides its limit:
\begin{equation}
    \lim_{N\to\infty} K_{\alpha}(A) = \frac{1}{|\Cl_\alpha|} \frac{n!}{\hat{\alpha}_{k_1}!\cdots \hat{\alpha}_{k_p}!} L_{[k_1],0}(A)^{\hat{\alpha}_{k_1}} \cdots L_{[k_p],0}(A)^{\hat{\alpha}_{k_p}}\,.
\end{equation}
Earlier in equation \eqref{eq:LimLK}, we have identified $L_{[k],0}(A)$ with $\kappa_k(A)/k$, thus implying
\begin{equation}
    \lim_{N\to\infty} K_{\alpha}(A) = \frac{1}{|\Cl_\alpha|} \frac{n!}{\prod_{i=1}^p k_i^{\hat{\alpha}_{k_i}}\hat{\alpha}_{k_i}!} \prod_{i=1}^p \kappa_{k_i}(A)^{\hat{\alpha}_{k_i}}\,.
\end{equation}
Using the expression\footnote{Switching from the notation $\alpha=[k_1^{\hat{\alpha}_{k_1}} \dots k_p^{\hat{\alpha}_{k_p}}]$ used here (with non-zero multiplicities $\hat{\alpha}_{k_i}$) to the one $\alpha=[1^{\hat\alpha_1}\dots n^{\hat\alpha_n}]$ used in section \ref{notations} (with potentially vanishing multiplicities $\hat{\alpha}_{k}$), we note that
\begin{equation*}
    |\Cl_\alpha| = \frac{n!}{\prod_{k=1}^n k^{\hat\alpha_k} \hat\alpha_k!} = \frac{n!}{\prod_{i=1}^p k_i^{\hat{\alpha}_{k_i}}\hat{\alpha}_{k_i}!}\,.
\end{equation*}}
\eqref{eq:SizeClass} of $|\Cl_\alpha|$ and recalling that $\hat{\alpha}_{k_i}$ is the number of times $k_i$ appears in the partition $\alpha=(\alpha_j)_{j=1,\dots,\ell(\alpha)}$, we simply get
\begin{equation}
    \lim_{N\to\infty} K_{\alpha}(A) = \prod_{i=1}^p \kappa_{k_i}(A)^{\hat{\alpha}_{k_i}} = \prod_{j=1}^{\ell(\alpha)} \kappa_{\alpha_j}(A)\,.
\end{equation}
This ends the proof of equation \eqref{eq:MultKappa} and thus this appendix.


\section{Expression of precursors for low degrees}
\label{TableKkappam}
 We list first the expression of $K _n$ in terms of the free cumulants,
for low values of $n\le N$:

\bea \nonumber
\nonumber K_1&=&\kappa_1\,,\\  \nonumber
K_2&=& 
{\frac{N^2}{N^2-1}} \kappa_2 \,, 
 \\ \nonumber 
K_3 &=&   { 
\frac{N^4}{(N^2-1)(N^2-4)}} \kappa_3\,,\\
\nonumber
\label{listKn} K_4 &=& 
\frac{N^4}{(N^2-1)(N^2-4)(N^2-9)} \big(   (N^2+1)\kappa_4 + 5   \kappa_2^2\big)\,,  \\ \nonumber
K_5
&=& \frac{N^6}{(N^2-1)(N^2-4)(N^2-9)(N^2-16)} \big((N^2+5)\kappa_5 +35 \kappa_2\kappa_3) \big)\,,  \\ \nonumber
K_6&=& \frac{N^6}{\tiny (N^2-1)(N^2-4)(N^2-9)(N^2-16)(N^2-25)}
\\ 
\nonumber && \qquad { \big((N^4 +15N^2+8)\kappa_6 +84  (N^2+2)\kappa_4 \kappa_2 +56 (N^2 - 1)  \kappa_3^2   +14\kappa_2^3(N^2+20)\big) }\,. 
\eea

 We may also express the $K_n$  in terms of the normalized moments of $A$,
 $m_n:=\inv{N} \Tr A^n$.
 These expressions may be quite cumbersome and it is profitable to remember there that the $K_n(A)$, $n>1$, are insensitive to a shift of $A$ by a multiple of the identity matrix, see Section \ref{sec:shift}. Writing $A=\hat A + m_1(A) \mathbb{I}$, with
$\hat A$ traceless, we can express the 
$K_n(A)=K_n(\hat A)$ ($n>1$) in terms of the $\hat m_n=\frac{1}{N}\Tr \hat A^n$, from which the general
expression in terms of $m_n(A)$ may be recovered by
substituting
\begin{equation*}
    \hat m_n = \sum_{r=0}^n {n\choose r} \, (-m_1)^{r} \,   m_{n-r} \,.
\end{equation*} 
We then give below the expression of the first six precursors:
\bea \nonumber K_1&\!\!\!=\!\!\!&m_1\,,
\\  \nonumber
K_2 &\!\!\!=\!\!\!& \frac{N^2}{N^2-1}  \hat m_2\,,
\\  \nonumber
 K_3 &\!\!\!=\!\!\!& \frac{N^4}{(N^2-1)(N^2-4)} \hat m_3\,,
 \\  
 \label{listKnmhat} 
 K_4 &\!\!\!=\!\!\!&   \frac{N^4}{(N^2-1)(N^2-4)(N^2-9)} \big(  (N^2+1)\hat m_4  - (2 N^2-3)  \hat m_2^2\big) \,,
 \\ \nonumber
 K_5 &\!\!\!=\!\!\!& \frac{N^6}{(N^2-1)(N^2-4)(N^2-9)(N^2-16)} \big((N^2+5)\hat m_5 -5  (N^2-2)\hat m_3 \hat m_2 \big) \,,
 \\  \nonumber
 K_6&\!\!\!=\!\!\!& \frac{N^6}{(N^2 -1)(N^2-4)(N^2-9)(N^2-16)(N^2-25)}
 \\
\nonumber && \quad \big((N^4 +15N^2+8)\hat m_6 -6 (N^4+N^2-20)\hat m_4 \hat m_2 -  (3N^4-11 N^2 +80) \hat m_3^2   +7(N^4-7 N^2)\hat m_2^3\big) \,. \eea
We may also express the moments $m_n(A)$ in terms of the generalised precursors $K_\alpha(A)$:
\bea  \nonumber m_1 &\!\!\!=\!\!\!& K_ 1 \,, \\
 \nonumber m_2 &\!\!\!=\!\!\!& 
K_{2}+K_{(1,1)} \,,
\\  
 \label{mnfunctionK} m_3 &\!\!\!=\!\!\!&
 \frac{N^2+1}{N^2} K_{3} + 3 K_{(2,1)}+K_{(1,1,1)} \,,\\  
 \nonumber m_4 &\!\!\!=\!\!\!&
 \frac{N^2+5}{N^2} K_{4}+4\frac{N^2+1}{N^2} K_{(3,1)}+\frac{2N^2+1}{N^2}K_{(2,2)}+6 K_{(2,1,1)}+K_{(1,1,1,1)} \,,\\  
 \nonumber m_5 &\!\!\!=\!\!\!&
\frac{N^4 + 15 N^2 + 8}{N^4}K_{5}+5\frac{N^2+5}{N^2}K_{(4,1)}+5\frac{N^2+3}{N^2}K_{(3,2)}\\ \nonumber &&\qquad +10\frac{N^2+1}{N^2} K_{(3,1,1)}+5\frac{2N^2+1}{N^2}K_{(2,2,1)}+10 K_{(2,1,1,1)}+K_{(1,1,1,1,1)}\,,\eea
and so on and so forth. Following  \eqref{malpha}, 
\begin{equation*}    
m_n = \sum_{\beta\vdash n}N^{\ell(\beta)-1} H^{\bullet\,,\, <}([n],\beta) K_\beta
\end{equation*} 
and the coefficients $H^{\bullet\,,\, <}([n],\beta)$, that are read off 
\eqref{mnfunctionK},
may be regarded as deformed Kreweras coefficients. 
They are generating functions of strictly monotone Hurwitz numbers $H_g^{\bullet\,,\, <}([n],\beta)$ or of $P_g(\alpha)$,
see Proposition \ref{Prop:TopMK} and
Corollary \ref{Coroll1}.

We have seen in Appendix \ref{App:MultK} that the $K_\alpha(\hat A)$ factorise in the large $N$ limit. 
What are the subdominant terms? Due to the identities \eqref{KAAa}-\eqref{identity1}, we can restrict to traceless $\hat A$ and singleton-free partitions. For the first few, we find 
\bea K_{(2,2)}(\hat A)&=&  K_2^2(\hat A) + C_{(2,2)}(\hat A)\,,\\
\nonumber  K_{(3,2)}(\hat A)&=& K_3(\hat A) K_2(\hat A)+C_{(3,2)} (\hat A)\,, \\
\nonumber  K_{(4,2)}(\hat A)&=& K_4(\hat A) K_2(\hat A)+C_{(4,2)} (\hat A)\,,
\eea
where
\begin{eqnarray*}
    C_{(2,2)} &=& \frac{-2}{N^2-1}\(\frac{2N^2-3}{N^2} K_4+K_{(2,2)}\)\,, \\
    C_{(3,2)} &=& \frac{-3}{N^2-1} \(\frac{2N^2-4}{N^2}K_5+K_{(3,2)}\)\,, \\
    C_{(4,2)} &=& \frac{-4}{N^2-1}\(\frac{2N^2-5}{N^2} K_6+ 2 K_{(4,2)} +K_{(3,3)} \)\,.
\end{eqnarray*}
As expected, these objects are of order ${\rm O}(N^{-2})$ when $N\to\infty$. Note that in the problem studied in Section \ref{cyclicpr}, these $C_\alpha$ appear as (classical) cumulants of the r.h.s. of Proposition \ref{prop4}.


\section{Tables of disconnected weakly monotone Hurwitz numbers}
\label{HurDisconnect}

The matrices $H^{\bullet\,,\,\le}(\alpha,\beta)$ of equation \eqref{eq:GenHm}, for $\alpha,\beta\vdash n$,  which are generating functions of the $H^{\bullet\,,\,\le}_g(\alpha,\beta)$,   are listed below, for  $n=1,\dots, 4$.
There, rows and columns are labelled by partitions of $n$, listed by reverse lexicographic order (thus for $n=4$: $[4], [3,1], [2^2], [3,1^2], \cdots, [1^4]$).
$$\left(
\begin{array}{c}
 1 \\
\end{array}
\right),\qquad \left(
\begin{array}{cc}
 \frac{N^2}{N^2-1} & -\frac{N}{N^2-1} \\
 -\frac{N}{N^2-1} & \frac{N^2}{N^2-1} \\
\end{array}
\right),\qquad \left(
\begin{array}{ccc}
 \frac{N^4}{(N^2-1)(N^2-4)} & -\frac{3 N^3}{(N^2-1)(N^2-4)} & \frac{2 N^2}{(N^2-1)(N^2-4)} \\[3pt]
 -\frac{2 N^3}{(N^2-1)(N^2-4)} & \frac{N^2 (N^2+2)}{(N^2-1)(N^2-4)} &
   -\frac{N^3}{(N^2-1)(N^2-4)} \\[3pt]
 \frac{4 N^2}{(N^2-1)(N^2-4)} & -\frac{3 N^3}{(N^2-1)(N^2-4)} & \frac{N^4-2 N^2}{(N^2-1)(N^2-4)}
   \\
\end{array}
\right)\,,$$
$$\hspace{-20pt}{\scriptstyle 
\inv{(N^2-1)(N^2-4)(N^2-9)}
\tiny  \left(
\begin{array}{ccccc}
N^4 (N^2+1) & -4N^3 (N^2+1)& -N^3 \left(2 N^2-3\right)& 10 N^4 & -{5 N^3} 
\\
-3N^3 (N^2+1) & N^2(N^4+3 N^2+12) & 3 N^2(2 N^2-3) &-3N^3(N^2+1) & N^2(2 N^2-3)
 \\
-2N^3(2N^2-3) & 8N^2(2N^2-3) & N^2(N^4-6 N^2+18)&-2N^3(N^2+6)& N^2(N^2+6)
\\
10 N^4& -4N^3 (N^2+1)\ & - N^3 \left(N^2+6\right) & N^2 (N^2+1)& -N^3 (N^2-4)
   \\
 -{30 N^3}& 8N^2 \left(2 N^2-3\right)& 3 N^2 (N^2+6)& -6N^3(N^2-4)& N^2(N^4-8 N^2+6)
  \\
\end{array}
\right)\,.}$$
Moreover, one checks that 
$\sum_{\alpha\vdash n} H^{\bullet, <}([n],\alpha) H^{\bullet, \le}(\alpha,\beta) = \delta_{\beta,[n]}$ with the matrix elements determined in Appendix \ref{TableKkappam}, as expected from the identity \eqref{eq:InvH}.


\section{Strictly monotone Hurwitz numbers and counting permutations by genus}
\label{App:GenPerm}

The goal of this appendix is to prove the Lemma \ref{Lem:GenPermHur}, relating certain strictly monotone Hurwitz numbers with the counting of permutations by genus and cycle types. In fact, we will prove a slightly more general result, as follows.

\begin{proposition}\label{Prop:StrictMon}
    The strictly monotone Hurwitz number $H_g^{\bullet\hspace{1pt},\hspace{1pt}<}(\alpha,\beta)$ is divisible by $|\mathrm{Cl}_\alpha|$ and $|\mathrm{Cl}_\beta|$. Moreover, fixing a permutation $\sigma_\beta$ in the conjugacy class $\mathrm{Cl}_\beta$, we have
    \begin{equation}
        \frac{1}{|\mathrm{Cl}_\beta|} H_g^{\bullet\hspace{1pt},\hspace{1pt}<}(\alpha,\beta) = |\Pi_g(\alpha \hspace{1pt}|\hspace{1pt} \sigma_\beta)|\,,
    \end{equation}
    where\footnote{By construction, $[\sigma_\alpha]=\alpha$ if $\sigma_\alpha\in \mathrm{Cl}_\alpha$: we could have thus replaced the condition in equation \eqref{eq:Pi} by $2g = n+2 - \ell(\alpha) - \ell(\beta) - \ell([\sigma_\alpha^{-1}\sigma_\beta])$.}
    \begin{equation}\label{eq:Pi}
        \Pi_g(\alpha \hspace{1pt}|\hspace{1pt} \sigma_\beta) = \bigl\lbrace \sigma_\alpha \in \mathrm{Cl}_\alpha\,\bigl|\,  n+2 - \ell([\sigma_\alpha]) - \ell([\sigma_\beta]) - \ell([\sigma_\alpha^{-1}\sigma_\beta]) = 2g \bigr\rbrace
    \end{equation}
    and we recall that $\ell([\sigma])$ measures the number of cycles of $\sigma\in S_n$.
\end{proposition}

Before proving this proposition, let us explain how it encompasses the Lemma \ref{Lem:GenPermHur}. To do so, we specialise the above result to $\beta=[n]$, choosing the representative of $\Cl_{[n]}$ to be the $n$-cycle $\zeta_n=(1\dots n)$. Using $|\Cl_{[n]}|=(n-1)!$ and $\ell([n])=1$, we then have
\begin{equation*}
    \frac{1}{(n-1)!} H_g^{\bullet\hspace{1pt},\hspace{1pt}<}(\alpha,[n]) = |\Pi_g(\alpha \hspace{1pt}|\hspace{1pt} \zeta_n)|\,,  \;\quad
        \Pi_g(\alpha \hspace{1pt}|\hspace{1pt} \zeta_n) = \bigl\lbrace \sigma_\alpha \in \Cl_\alpha\,\bigl|\,  n+1 - \ell([\sigma_\alpha]) - \ell([\sigma_\alpha^{-1}\zeta_n]) = 2g  \bigr\rbrace\,.
\end{equation*}
From the definition \eqref{eq:GenPerm} of the genus of a permutation, we recognise in $\Pi_g(\alpha \hspace{1pt}|\hspace{1pt} \zeta_n)$ the set of permutations of cycle type $\alpha$ and genus $g$, whose size is $P_g(\alpha)$. This then proves Lemma \ref{Lem:GenPermHur}.

\begin{proof}
We will need some notations and preliminary results. For $r\in\lbrace 1,\dots,n-1\rbrace$, we let $\Sigma_{n,r}$ be the set of $r$-uples $(\tau_1,\dots,\tau_r)$ of transpositions in $S_n$ such that $\tau_i=(a_i\,b_i)$ with $a_i< b_i$ and $b_1 < \cdots < b_r$. If a permutation $\sigma\in S_n$ can be written as a product $\sigma=\tau_1\cdots\tau_r$ with $(\tau_1,\dots,\tau_r)\in\Sigma_{n,r}$, then we say that $\tau_1\cdots\tau_r$ is a primitive factorisation of $\sigma$ of length $r$. The following result is standard (see for instance~\cite[Lemma 1.5.1]{G2012}) and will play a crucial role in our proof.

\begin{lemma}\label{Lem:PrimFac}
    Every permutation $\sigma\in S_n$ admits a unique primitive factorisation. Moreover, the length of this factorisation is equal to $n-\ell([\sigma])$. In other words, the map
    \begin{equation*}
        \begin{array}{ccc}
            \Sigma_{n,r} & \longrightarrow & \bigl\lbrace \sigma\in S_n\;| \; \ell([\sigma]) = n-r \bigr\rbrace  \\
            (\tau_1,\dots,\tau_r) & \longmapsto & \tau_1\cdots\tau_r
        \end{array}
    \end{equation*}
    is a bijection for every $r\in\lbrace 1,\dots,n-1\rbrace$.
\end{lemma}

Coming back to the proof of Proposition \ref{Prop:StrictMon}, we now fix the integers $n\geq 1$, $g\geq 0$ and two partitions $\alpha,\beta\vdash n$. We then let $r=2g-2+\ell(\alpha)+\ell(\beta)$. By a slight reformulation of its definition (see Theorem \ref{thm:Hur} and the Remark \ref{Rmk:Hur} below), the strictly monotone Hurwitz number $H_g^{\bullet\hspace{1pt},\hspace{1pt}<}(\alpha,\beta)$ can be expressed as the size of the set
\begin{equation*}
    \bigl\lbrace (\sigma_\alpha,\sigma_\beta)\in\Cl_\alpha\times\Cl_\beta,\,(\tau_1,\dots,\tau_r) \in \Sigma_{n,r}\;\bigl| \; \sigma_\alpha^{-1} \sigma_\beta = \tau_1 \cdots \tau_r \bigr\rbrace\,.
\end{equation*}
The condition imposed on these permutations is equivalent to $\sigma_\alpha^{-1} \sigma_\beta$ admitting a primitive factorisation of length $r$. By Lemma \ref{Lem:PrimFac}, the above set is then bijective to
\begin{equation*}
    \bigl\lbrace (\sigma_\alpha,\sigma_\beta)\in\Cl_\alpha\times\Cl_\beta\;\bigl| \; \ell([\sigma_\alpha^{-1} \sigma_\beta]) = n-r \bigr\rbrace\,.
\end{equation*}
Recalling that $r=2g-2+\ell(\alpha)+\ell(\beta)$, this is also equal to
\begin{equation}
    \bigl\lbrace (\sigma_\alpha,\sigma_\beta)\in\Cl_\alpha\times\Cl_\beta\;\bigl| \; n+2-\ell([\sigma_\alpha]) - \ell([\sigma_\beta]) - \ell([\sigma_\alpha^{-1} \sigma_\beta]) = 2g \bigr\rbrace\,.
\end{equation}
We can further decompose this set with respect to the choice of element $\sigma_\beta$, yielding a bijection with the disjoint union
\begin{equation}\label{eq:UnionPi}
    \bigsqcup_{\sigma_\beta\,\in\,\Cl_\beta} \Pi_g(\alpha\hspace{1pt}|\hspace{1pt}\sigma_\beta)\,,
\end{equation}
where $\Pi_g(\alpha\hspace{1pt}|\hspace{1pt}\sigma_\beta)$ was defined in equation \eqref{eq:Pi}. Let us fix one representative $\sigma_\beta$ in the conjugacy class $\Cl_\beta$. Any other representative is of the form $\rho\sigma_\beta\rho^{-1}$ for some $\rho\in S_n$. It is then clear that $\sigma_\alpha \mapsto \rho\sigma_\alpha\rho^{-1}$ is an isomorphism between $\Pi_g(\alpha\hspace{1pt}|\hspace{1pt}\sigma_\beta)$ and $\Pi_g(\alpha\hspace{1pt}|\hspace{1pt}\rho\sigma_\beta\rho^{-1})$. The different sets in the union \eqref{eq:UnionPi} are thus all bijective to a fixed one $\Pi_g(\alpha\hspace{1pt}|\hspace{1pt}\sigma_\beta)$ and in particular are all of the same size. The cardinality of this union is then equal to
\begin{equation}
    |\Cl_\beta| \times |\Pi_g(\alpha\hspace{1pt}|\hspace{1pt}\sigma_\beta)|\,.
\end{equation}
Recalling the sequence of bijections that led us to the set \eqref{eq:UnionPi}, we note that this cardinality is none other than the Hurwitz number $H_g^{\bullet\hspace{1pt},\hspace{1pt}<}(\alpha,\beta)$, thus proving Proposition \ref{Prop:StrictMon}.
\end{proof}


\end{document}